\documentclass{elsarticle}
\usepackage{graphicx}
\usepackage{amsmath,amssymb,amsthm}
\usepackage{lineno}
\usepackage{comment}
\usepackage{algorithm}
\usepackage{mathrsfs}

{\catcode`[=\active
\gdef\activatealgo{\let\>\qquad\let\=\leftarrow
        \catcode`[=\active\def[##1]{\ifmmode \lbrack##1]\else{\bf ##1}\fi}}
}

\def\poly{\textit{poly}}
\def\tw{{\it tw}}
\def\ds{{\it ds}}
\def\vc{{\it vc}}
\def\col{{\it 3col}}

\def\qr{{\it qr}}
\def\FO{\ensuremath{{\rm FO}}}
\def\MSO{{\rm MSO}}
\def\Pow{{\ms P}}
\def\root{{\it root}}
\def\subgames{{\it subgames}}
\def\subgame{{\it subgame}}
\def\next{{\it next}}

\newcommand{\linmso}{\ensuremath{\mbox{LinMSO}}}
\newcommand{\mso}{\ensuremath{\mbox{MSO}}}

\def\STR{{\mathcal{STR}}}

\def\redgame{{\ensuremath{\it reduce}}}
\def\OPT{{\it OPT}}
\def\interpreted{{\ensuremath{\it interpreted}}}
\def\nullary{{\ensuremath{\it null}}}
\def\unaryrelations{{\ensuremath{\it unary}}}
\def\relations{{\ensuremath{\it rel}}}
\def\vocabulary{{\ensuremath{\it vocabulary}}}

\def\true{{\ensuremath{\top}}}
\def\false{{\ensuremath{\bot}}}
\def\nil{{\ensuremath{\rm nil}}}
\def\mcg{{\ensuremath{\mathcal{MC}}}}
\def\emcg{{\ensuremath{\mathcal{EMC}}}}
\def\eval{{\it eval}}
\def\reduce{{\ensuremath{\it reduce}}}
\def\combine{{\ensuremath{\it combine}}}
\def\forget{{\ensuremath{\it forget}}}
\def\convert{{\ensuremath{\it convert}}}
\def\val{{\ensuremath{\it val}}}

\newcommand{\ms}[1]{\ensuremath{\mathscr{#1}}}

\newcommand{\adj}{\ensuremath{{\it adj}}}
\newcommand{\inc}{\ensuremath{{\it inc}}}
\newcommand{\arity}{\ensuremath{{\it arity}}}

\newcommand{\taugraph}{\ensuremath{\tau_{\text{Graph}}}}

\newtheorem{lemma}{Lemma}
\newtheorem{invariant}{Invariant}
\newtheorem{theorem}{Theorem}
\newtheorem{corollary}{Corollary}
\newtheorem{example}{Example}
\theoremstyle{definition}
\newtheorem{definition}{Definition}

\begin{document}

\begin{frontmatter}
\title{Courcelle's Theorem -- A Game-Theoretic Approach\tnoteref{thanks}}
\tnotetext[thanks]{Supported by the DFG under grant RO 927/8}

\author{Joachim Kneis}
\ead{kneis@cs.rwth-aachen.de}
\author{Alexander Langer}
\ead{langer@cs.rwth-aachen.de}
\author{Peter Rossmanith}
\ead{rossmani@cs.rwth-aachen.de}
\address{Department of Computer Science, RWTH Aachen University}

\begin{abstract}
Courcelle's Theorem states that every problem definable in Monadic
Second-Order logic can be solved in linear time on structures of bounded
treewidth, for example, by constructing a tree automaton that recognizes
or rejects a tree decomposition of the structure.  Existing, optimized software
like the MONA tool can be used to build the corresponding tree automata,
which for bounded treewidth are of constant size.
Unfortunately, the constants involved can become extremely large -- every
quantifier alternation requires a power set construction for the automaton.
Here, the required space can become a problem in practical applications. 

In this paper, we present a novel, direct approach based on model checking games,
which avoids the expensive power set construction.  
Experiments with an implementation are promising, and we can solve problems on graphs
where the automata-theoretic approach fails in practice.
\end{abstract}

\end{frontmatter}

Courcelle's celebrated theorem essentially states that
every problem definable in Monadic Second-Order logic (MSO) can be
solved in linear time on graphs of bounded treewidth~\cite{Cou90}.  However,
the multiplicative constants in the running time, which depend
on the treewidth and the MSO-formula, can be extremely large~\cite{FG04}.

\begin{theorem}[\cite{Cou90,FG04}]
\label{thm:courcelle}
Let $\mathcal P$ be an MSO problem and $w$ be a positive integer.
There is an algorithm $A$ and a function $f\colon \mathbf N \times \mathbf N \to \mathbf N$
such that for every graph $\mathcal G = (V, E)$ of order $n := |V|$ and
treewidth at most $w$, $A$ solves $\mathcal P$ on input $\mathcal G$ in time $f(\|\varphi\|, w) \cdot n$, where
$\varphi$ is the MSO formula defining~$\mathcal P$ and $\|\varphi\|$ is its length.
Furthermore, unless ${\rm P = NP}$, the function $f$ cannot be upper bounded by an iterated exponential of bounded height in terms
of~$\varphi$ and~$w$.
\end{theorem}

This result has been generalized by Arnborg, Lagergren, and Seese
to Extended MSO~\cite{ALS91}, and by Courcelle and Mosbah to
Monadic Second-Order evaluations using semiring
homomorphisms~\cite{CM93}.  In both cases, an MSO-formula with free
set variables is used to describe a property, and satisfying
assignments to these set variables are evaluated in an appropriate way.

Courcelle's Theorem is usually proved as follows: In time
only dependent on~$\varphi$ and the treewidth~$w$, a tree automaton $\mathcal A$ is constructed that
accepts a tree decomposition of width~$w$ if and only if the corresponding graph satisfies the formula.
This construction can either be done explicitly, by actually constructing the
tree automaton (see, e.g., \cite{ALS91,Cou90a,DF99,FFG02,FG06,Wey08,Cou11}),
or implicitly via auxiliary formulas obtained by applying the
Feferman--Vaught Theorem~\cite{FV59} extended to MSO~\cite{Cou90,Gur79}
(see, e.g., \cite{Cou90,Fri01,Mak04,Gro07,Cou11}).

In a practical setting, the biggest strength of Courcelle's Theorem
is at the same time its largest weakness:  MSO logic has extremely
large expressive power, and very short formulas can be used to
encode NP-hard problems.  This is used in~\cite{FG04} to prove
non-elementary worst-case lower bounds for the multiplicative constants
in the linear running time.
Even worse, these lower bounds already hold for the class of trees,
i.e., graphs of treewidth one.

On the other hand, these are worst-case lower bounds
for very special classes of formulas and trees, and thus there is
a good chance that in practice problems can be solved much faster.
In fact, existing software like the MONA tool~\cite{KM01,KMS00} for Weak Second-Order
logic on two successors (WS2S) is surprisingly successful
even though it is subject to the same theoretical lower bounds.

The automata-theoretic approach is therefore a promising starting point
for practical applications of Courcelle's Theorem, particularly
since advanced and optimized tools like MONA can be used as a black
box for the majority of the work, and techniques like minimizing tree automata
are very well understood.

There are, however, some cases where the automata-theoretic approach
is infeasible in practice, i.e., when the automata (or set of
auxiliary formulas) are too large to be practically computable.
This can even happen
when the final minimal automata are small, but intermediate automata
cannot be constructed in reasonable time and space (note that each
quantifier alternation requires an automaton power set construction).

In his thesis~\cite{Sog08}, Soguet has studied the sizes of tree
automata corresponding to various problems for small
clique-width~\cite{CMR00}.\footnote{Both, treewidth and clique-width, can be
defined in terms of graph grammars (hyperedge replacement grammars
for treewidth, and vertex replacement grammars for clique-width; see
e.g.\ a recent survey~\cite{Cou08}), and in both cases, tree automata can
be used to recognize parse trees of graphs.}
The automata were generated using MONA, and in many cases, the
corresponding automata were surprisingly small, thanks to the
well-understood minimization of tree automata.
On the other hand, even for graphs of 
clique-width three, MONA was unable to construct the corresponding
tree automata for the classical \textsc{3-Colorability} problem.
Even worse, the same happened for simple problems such as deciding
whether the graph is connected or if its maximum degree is two. 

These negative results are somewhat unsatisfying because
the respective algorithm already fails in the first phase,
when the automaton is
constructed.  The first phase, however,
only depends on the treewidth
(or clique-width in above cases) and the formula (i.e., the problem),
but is independent of the actual input graph.
On the other hand, when running the tree automaton on most graphs
arising from practical problems, only
few states are actually visited.

Recently, there have been a few approaches to this problem, see,
e.g.,~\cite{GPW08a,GPW10,CD10,CD10a}.
For example, the approach of~\cite{CD10,CD10a} avoids
an explicit construction of the tree automaton.  Instead, the state-transition
function is computed \emph{on-the-fly}.  Experiments indicate practical feasibility.
Courcelle~\cite{Cou10a} introduces \emph{special tree-width}, where
the corresponding automata are easier to construct.

In this paper, we present a novel, game-theoretic approach,
where the input structure is taken into account
from the beginning via model checking games (cf.,~\cite{Hin73,Gra07,Gra11}).
Therefore, only the amount
of information is stored that is needed by the algorithm to solve the problem
on this explicit input, and, in some sense, transitions between nodes
of the tree decompositions are as well computed \emph{on-the-fly}.
We particularly avoid the expensive power set construction.

We hope that the approach can be used in those cases, where the automata
are too large to be constructed in practice, but the input graphs itself
are simple enough.   In fact, first experiments are promising.  Using the
generic approach, we can, for example, solve the \textsc{3-Colorability}
problem on grids of size $6 \times 33$ (treewidth~6) in about 21 seconds and
with 8 MB memory usage on standard PC hardware, and the \textsc{Minimum
Vertex Cover} problem on the same graph in less than a second and only 1 MB of
memory usage.   We note that the automata construction using MONA
in \cite{Sog08} already failed for $2 \times n$ grids (clique-width~3).

\subsection*{Related Work}

We briefly survey other approaches to Courcelle's Theorem.
We already mentioned that, given the MSO formula $\varphi$, one can
construct a finite-state bottom-up tree automaton that accepts a tree
decomposition of the input graph~$G$ if and only if $G \models
\varphi$.  This is sometimes called the automata theoretic approach.
A direct construction of the tree automata
is described in, e.g.,~\cite{Wey08} or~\cite[Chapter 6]{Cou11}.
In \cite{AF93,DF99} a Myhill-Nerode type argument is used to show
that the treewidth parse tree operators admit a right congruence
with finitely many congruence classes.  The method of test sets can then
be used to construct the tree automaton. 
One can also use a reduction to the classical model checking
problem for MSO on labeled trees~\cite{ALS91,FFG02,FG06}.
It is well-known~\cite{TW68,Don70} that this problem can be solved by
constructing suitable finite-state tree automata.
This approach is favorable if one likes to use existing software such as
the MONA tool~\cite{KM01}.

A model theoretic approach is based on variants of the
Feferman--Vaught Theorem~\cite{FV59}:
If a graph $G$ can be
decomposed into components $G_1$ and $G_2$, then from the input
formula~$\varphi$ one can construct a suitable \emph{reduction
sequence} consisting of Boolean combinations (\emph{and}, \emph{or}, \emph{not})
of finitely many
formulas that hold in $G_1$ and $G_2$ if and only if $\varphi$ holds
in $G$ (cf.,~\cite{Cou90,Gur79,Mak04,Cou11}).  One can therefore
use dynamic programming on the tree decomposition to compute the
\emph{$q$-theory} of $G$, i.e., set of formulas of quantifier rank at
most~$q$ that hold in~$G$ (cf.,~\cite{Fri01,Gro07,Mak04}).
Similarly, one can also inductively compute the set of satisfying assignments to the
input formula~\cite{CM93}.


We are not aware of any implementations of Courcelle's Theorem based
on the Feferman--Vaught approach.  The construction of all possible
reduction sequences for MSO formulas ``obviously is not
practical''~\cite[Section 1.6]{Mak04}.   The algorithms presented
in~\cite{Fri01,Mak04} are therefore infeasible in practice.
However, from~\cite{CM93} we get that computing the particular reduction sequence for
the input formula~$\varphi$ suffices.  Some lower bounds are known for the
necessary conversions into disjunctions~\cite{MRW05}, but it would
still be interesting to see how this approach behaves in practice.

A few authors studied practical aspects of the automata theoretic approach.
It is mentioned in~\cite{DF99} that a Myhill--Nerode based program has
been implemented as part of an M.Sc.\ thesis, which unfortunately
does not seem to be publicly available. 
The MONA tool~\cite{KM01} is a well-known and optimized implementation
for the tree automata construction.  The space required to construct
the automata with MONA still turns out to cause severe problems in
practical applications~\cite{Sog08,GPW10}.
One idea~\cite[Chapter 6]{Cou11}
is to use precomputed automata for commonly used predicates such
as ${\it Conn}(X)$ expressing that the set~$X$ is connected.
Note however that the ${\it Conn}(X)$ automaton
requires $2^{2^{\Theta(k)}}$ states
for graphs of clique-width~$k$~\cite[Chapter 6]{Cou11}.
An automatic translation into Monadic
Datalog is proposed in~\cite{GPW10}.   Some experiments 
indeed suggest feasibility in practice; their
prototype implementation was, however, obtained by manual construction
and
not by an automatic transformation from the underlying MSO formula.
In~\cite{CD10,CD10a} the power set construction is avoided by considering
existential formulas only.  The automata thus remain non-deterministic, but
of course standard methods to simulate runs of the automata apply.
Since the state transition function is given only implicitly, the automaton
is essentially computed \emph{on-the-fly} while recognizing a clique-decomposition.
Experiments have been conducted on graphs of comparably high clique-width and
the approach is quite promising.
In fact,  the lack of feasible algorithms to compute the
necessary clique-width parse trees seems to be the major limitation.
To ease the specification of such \emph{fly-automata},
Courcelle~\cite{Cou10a,Cou11a} introduces \emph{special tree-width}.
Special tree-width lies between path-width and treewidth, but the
automata are significantly smaller and easier to construct than
those for treewidth.

In this article, we present a new approach that neither uses automata theoretic
methods nor uses a Feferman--Vaught style splitting theorem.  Instead, we essentially
evaluate the input formula on the graph using a simple recursive model checking
algorithm.  In what follows, we shall outline this approach.

\subsection*{Overview}

Our starting point is the \emph{model checking game} for MSO
(Definition~\ref{def:model-checking-game}), a pebble game between two
players called the \emph{verifier} and the \emph{falsifier} also known as
the Hintikka game~\cite{Hin73}.  The verifier
tries to prove that the formula holds on the input structure, while the
falsifier tries to prove the opposite.  In the game, the verifier moves on
existential formulas ($\vee$, $\exists$), while the falsifier moves on
universal formulas ($\wedge$, $\forall$).

This game can in a natural way be identified with a simple algorithm
that \emph{evaluates} the formula on the input structure in a recursive manner.
If, for example, the formula is $\exists R \psi(R)$ for a set variable~$R$,
the algorithm checks whether $\psi(U)$ holds for all sets $U$.
In this sense, the computation tree of this simple algorithm can be
interpreted as the \emph{unfolding} (cf.,~\cite{McM95}) of the model checking game.
On a structure with $n$ elements, this straight-forward recursive
model-checking algorithm takes time $O((2^n + n)^q)$ for a formula
of quantifier rank~$q$.  By dynamic programming on the tree
decomposition, we can improve this to time linear in $n$ on structures of
bounded treewidth.

This works as follows:  We traverse the tree decomposition
of the input structure $\ms A$ bottom-up.  At each node of the tree
decomposition we preliminary try to evaluate the formula $\varphi$
on $\ms A$ using the model checking game on the ``current''
substructure~$\ms A'$ of $\ms A$.  To this end, we allow ``empty''
assignments $x := \nil$ to first order variables $x$.  Such
empty assignments correspond to objects in $\ms A$ that are not
contained in $\ms A'$ and are to be assigned in later steps.  Then,
two things may happen:

\begin{itemize}
\item We can \emph{already now} determine whether $\ms A \models \varphi$ or $\ms A \not\models \varphi$.

If, for instance, the formula $\col$ encodes the \textsc{3-Colorability} problem and
even $\ms A'$ is not three-colorable, it locally violates $\col$ and we can derive $\ms A \not\models \col$.

\item We cannot \emph{yet} determine whether $\ms A \models \varphi$ or $\ms A \not\models \varphi$.

	For example, if the formula encodes \textsc{Dominating Set} problem, then a vertex~$v$ in the ``current'' bag
	might be undominated in the current subgraph, but we do not know whether in the ``future'' another vertex
	might dominate~$v$.
\end{itemize}

The first case is formalized in Lemma~\ref{lem:emcg-eval-introduce} and Lemma~\ref{lem:emcg-eval-union}.
In the second case, we found a ``witness,'' i.e., a subgame that we were unable to evaluate. 
We then will re-visit those undetermined subgames
during the course of the dynamic programming until we finally arrive in the root of the tree decomposition, where
all subgames become determined.

The next crucial observation is that MSO and FO formulas with
bounded quantifier rank have limited capabilities to distinguish structures
(formally captured in the $\equiv_q$-equivalence of structures,
cf.~\cite{EF99}).  We exploit this fact and show that we can delete
redundant equivalent subgames (cf., Algorithm~\ref{alg:reduce})
for a suitable definition of equivalence
(cf., Definition~\ref{def:equivalent-games}).  We can then show that, assuming
a fixed formula and bounded treewidth, the number of reduced, non-equivalent games
is bounded by a constant (Lemma~\ref{lem:number}), which allows us to obtain
running times linear in the size of the tree decomposition.

While this game-theoretic approach is subject to the same non-elementary
lower bounds as the other approaches, the actual number of ways to
play the model checking game highly depends on the input graph.  For example,
if the graph does not contain, say, a triangle, then the players will never move
to a set of nodes that induce a triangle, while a tree automaton
must work for all graphs.  This observation is reflected in practical experiments, where
the actual number of entries considered is
typically much smaller than the corresponding worst-case bound.

\section{Preliminaries}
\label{sec:prel}

The power set of a set $U$ is denoted by $\Pow(U)$.  The disjoint union
of two sets $U_1, U_2$ is denoted by~$U_1 \uplus U_2$.
We assume that trees are rooted and denote the root of a tree $\mathcal T$ by
$\root(\mathcal T)$. 
For every $t \in \mathbf N$, $\exp^t(\cdot)$ is a $t$-times iterated
exponential, i.e., $\exp^0(x) = x$ and $\exp^t(x) = 2^{\exp^{t-1}(x)}$.

For a set $U$ and object $x$, we let $(x \in U)$ be defined as
$$
(x \in U) = \begin{cases}
0 & \text{if $x \notin U$} \\
1 & \text{if $x \in U$.}
\end{cases}
$$

To avoid cluttered notation, we may, for elements $s_1, \ldots, s_l$ and
$t_1, \ldots, t_m$, abbreviate $\bar s := \{s_1, \ldots, s_l\}$,
$(\bar s, s') := \bar s \cup \{s'\}$, and
$(\bar s, \bar t) := \bar s \bar t := \bar s \cup \bar t$.

\subsection{Structures}

We fix a countably infinite set of \emph{symbols}. 
Each symbol $S$ has an \emph{arity}~$r = \arity(S) \ge 0$.
We distinguish between \emph{nullary} symbols with arity zero
and \emph{relation symbols} that have arity greater
than zero.  Relation symbols with arity one are called \emph{unary}.
For convenience, we shall denote relation symbols by capital letters
and nullary symbols by lower case letters.

A \emph{vocabulary} $\tau$ is a finite set of symbols.
We denote by $\nullary(\tau)$ the set of nullary symbols in $\tau$,
by $\relations(\tau)$ the set of relation symbols in $\tau$, and
by $\unaryrelations(\tau)$ the set of unary relation symbols in $\tau$.
Let $\arity(\tau) = \max\{\,\arity(R) \mid R \in \relations(\tau)\,\}$
be the maximum arity over all relation symbols in $\tau$.
If $\nullary(\tau) = \emptyset$, we call $\tau$ \emph{relational}.

Let $\tau$ be a vocabulary.
A \emph{structure~$\ms A$ over $\tau$} (or $\tau$-structure) is a
tuple $\ms A = \left(A, (R^{\ms A})_{R \in \relations(\tau)}, (c^{\ms A})_{c \in \nullary(\tau)}\right)$, 
where $A$ is a finite set called the \emph{universe} of $\ms A$, and $(R^{\ms
A})_{R \in \relations(\tau)}$ and $(c^{\ms A})_{c \in \nullary(\tau)}$ are
\emph{interpretations} of the $\tau$-symbols in~$\ms A$.
Here, $R^{\ms A} \subseteq A^{\arity(R)}$ for each relation symbol
$R \in \relations(\tau)$.  For a nullary symbol $c \in \nullary(c)$
we either have $c^{\ms A} \in A$ and say that $c$ is \emph{interpreted}
in~$\ms A$, or we write $c^{\ms A} = \nil$ and say that~$c$ is
\emph{uninterpreted}.
The set of nullary symbols interpreted in $\ms A$ is denoted by
$\interpreted(\ms A)$.  
If all symbols are interpreted, we say the
structure is \emph{fully interpreted}, and \emph{partially interpreted}
otherwise.  We note that a related concept of \emph{partially equipped signatures}
has been used in, e.g.,~\cite{AF93,DF99,GH10}.

The set of all $\tau$-structures is denoted by $\STR(\tau)$.
We shall always denote structures in script letters $\ms A, \ms B, \ldots$ and
in roman letters $A, B, \ldots$ their corresponding universes.
If the universe is empty, then we say that the structure is \emph{empty}.
Structures over a
relational vocabulary $\tau$ are called \emph{relational structures}.

For a structure $\ms A$, we denote by $\vocabulary(\ms A)$ the vocabulary of $\ms A$.
For sets $\bar R = \{R_1, \ldots, R_l\} \subseteq \relations(\tau)$ and
$\bar c = \{c_1, \ldots, c_m\} \subseteq \nullary(\tau)$, we let
$\bar R^{\ms A} := \{\,R^{\ms A} \mid R \in \bar R\,\}$, and
$\bar c^{\ms A} := \{\,c^{\ms A} \mid c \in {\bar c \cap \interpreted(\ms A)} \,\}$
be their corresponding interpretations.

\begin{example}
\label{ex:taugraph}
A graph $(V, E)$ can in a natural way be identified with a structure~$\ms G$
over the vocabulary $\taugraph = (\adj)$, where $\adj$ represents the binary adjacency
relation.  The universe of $\ms G$ is $V$, and we interpret $\adj$ as
$\adj^{\ms G} = E$ in~$\ms G$.
\end{example}

Let $\tau$ be a vocabulary and $\{R_1, \ldots, R_l, c_1, \ldots, c_m\}$
be a set of symbols, each of which is not contained in~$\tau$.
The vocabulary $\tau' = (\tau , R_1, \ldots, R_l, c_1, \ldots, c_m)$ is
called an \emph{expansion} of~$\tau$.  Similarly, if $\ms A$ is a $\tau$-structure and
$\ms A'$ is a $\tau'$-structure that agrees with $\ms A$ on $\tau$, i.e.,
$R^{\ms A} = R^{\ms A'}$ for each $R \in \relations(\tau)$ and
$c^{\ms A} = c^{\ms A'}$ for each $c \in \nullary(\tau)$, then
we call $\ms A'$ a $\tau'$-expansion of~$\ms A$.
If~$\ms A$ is a $\tau$-structure, and $U_1, \ldots, U_l$ are relations over~$A$, such that
$U_i \subseteq A^{\arity(R_i)}$, $1 \le i \le l$, and $u_1, \ldots, u_m \in A \cup \{\nil\}$, we write
$\ms A' = (\ms A, U_1, \ldots, U_l, u_1, \ldots, u_m)$ to indicate that $\ms A'$ is
a $\tau'$-expansion of $\ms A$, such that $R_i^{\ms A'} = U_i$, $1 \le i \le l$, and
$c_j^{\ms A'} = u_j$, $1 \le j \le m$.

Let $\ms A$ be a $\tau$-structure and
$\bar a = \{a_1, \ldots, a_m\} \subseteq A$.
Then $\ms A[\bar a]$ is the \emph{substructure of $\ms A$ induced by
$\bar a$}, where $\ms A[\bar a]$ has universe $\bar a$, 
for each relation symbol $R \in \tau$ we have
	$R^{\ms A[\bar a]} = R^{\ms A} \cap \bar a^{\arity(R)}$,
and nullary symbols $c$ are interpreted as $c^{\ms A[\bar a]} = c^{\ms A}$ if $c^{\ms A} \in \bar a$
and become uninterpreted otherwise.

Two $\tau$-structures $\ms A$ and $\ms B$ over the same
vocabulary~$\tau$ are \emph{isomorphic}, denoted by $\ms A \cong \ms B$,
if there is an \emph{isomorphism}
$h\colon A \to B$, where $h$ is a bijection between $A$ and $B$ and
\begin{itemize}
\item $c \in \interpreted(\ms A)$ if and only if $c \in \interpreted(\ms B)$
	for all $c \in \nullary(\tau)$,
\item $h(c^{\ms A}) = c^{\ms B}$ for every nullary symbol $c \in \interpreted(\tau)$, and
\item for every relation symbol $R \in \tau$ and $a_1, \ldots,
	a_p \in A$, where $p = \arity(R)$,
	$$(a_1, \ldots, a_p) \in R^{\ms A} \qquad \text{iff}
	\qquad (h(a_1), \ldots, h(a_p)) \in R^{\ms B}.$$
\end{itemize}


\begin{definition}[Compatibility, Union]
We call two $\tau$-structures $\ms A_1$ and $\ms A_2$ \emph{compatible}, if 
for all nullary symbols $c \in \interpreted(\ms A_1) \cap \interpreted(\ms A_2)$ we
have $c^{\ms A_1} = c^{\ms A_2}$ and
the identity $x \mapsto x$ is an isomorphism between
$\ms A_1[A_1 \cap A_2]$ and $\ms A_2[A_1 \cap A_2]$.

In this case, we define the \emph{union} of $\ms A_1$ and $\ms A_2$,
denoted by $\ms A_1 \cup \ms A_2$, as the $\tau$-structure with universe
$A := A_1 \cup A_2$ and interpretations
$R^{\ms A_1 \cup \ms A_2} := R^{\ms A_1} \cup R^{\ms A_2}$ for every
relation symbol $R \in \tau$.  Nullary symbols
$c \in \nullary(\tau)$ with $c^{\ms A_1} = c^{\ms A_2} = \nil$ remain uninterpreted
in $\ms A_1 \cup \ms A_2$; otherwise $c^{\ms A_1 \cup \ms A_2} = c^{\ms A_i}$ if
$c \in \interpreted(\ms A_i)$ for some $i \in \{1,2\}$.
\end{definition}

\subsection{Treewidth and Tree Decompositions}
\label{sec:treewidth}

Tree decompositions and treewidth were introduced by
Robertson and Seymour~\cite{RS86}
in their works on the Graph Minors Project, cf.~\cite{DF99,FG06,Die10}.

A \emph{tree decomposition} of a relational $\tau$-structure $\ms A$ is a tuple
$(\mathcal T, \mathcal X)$, where $\mathcal T = (T, F)$ is a rooted
tree and $\mathcal X = (X_i)_{i \in T}$ is a collection
of subsets $X_i \subseteq A$, such that
\begin{itemize}
\item $\bigcup_{i \in T} X_i = A$,
\item for all $p$-ary relation symbols $R \in \tau$ and all
	$(a_1, \ldots, a_p) \in R^{\ms A}$, 
	there is an $i \in T$ such that $\{a_1, \ldots, a_p\} \subseteq X_i$, and
\item for all $i, j_1, j_2 \in T$, if $i$ is on the path between $j_1$ and $j_2$ in
	$\mathcal T$, then $X_{j_1} \cap X_{j_2} \subseteq X_i$.
\end{itemize}

The sets $X_i$ are called \emph{bags}. The \emph{width} of a tree
decomposition is the size of its largest bag minus one, and the
\emph{treewidth} of a structure $\ms A$ is the minimum width of all tree
decompositions of~$\ms A$. 

Without loss of generality, we assume that each tree decomposition we
consider is \emph{nice}.  Nice tree decompositions are directed,
where each edge in $F$ has a direction away from the root, and have
the following properties:  Each node $i \in T$ has at most two children.
For leafs $i \in T$, we have $X_i = \emptyset$.
If $i$ has exactly one
child~$j$, then there is $a \in A$ such that either $X_i = X_j
\cup \{a\}$ or $X_i = X_j \setminus \{a\}$.  In the former case,
we say $i$ is an \emph{introduce} node, in the latter case we call
$i$ a \emph{forget} node of the tree decomposition.  Finally, if a node $i$
has two children $j_1$ and $j_2$, then we require $X_{i} = X_{j_1}
= X_{j_2}$ and call such nodes \emph{join} nodes.
If $i \to \cdots \to j$ is a directed path in $\mathcal T$ pointing away from the root,
we say $j$ appears \emph{below} $i$ in $\mathcal T$.

With every node $i \in T$ of a (nice) tree decomposition of a $\tau$-structure
$\ms A$ we associate a substructure $\ms A_i$ defined as follows:
Let $A_i \subseteq A$ be the set of objects in $X_i$ or in bags $X_j$ for nodes
$j$ below $i$ in the tree decomposition.
Then we let $\ms A_i := \ms A[A_i]$ be the substructure of $\ms A$
induced by $A_i$.

Computing the treewidth of a graph is NP-complete~\cite{ACP87}.
However, the algorithms in this paper rely on a given tree decomposition
of the input structure.  For graphs~$G$, there is a fixed-parameter tractable
algorithm~\cite{Bod96,DF99} with a running time of $2^{O(\tw(G)^3)} |G|$,
whose dependence on the treewidth might become a problem in
practical applications.  In a practical setting, heuristics seem
to work well and often nearly optimal tree decompositions can
be computed~\cite{BK10}.  Using \emph{Gaifman graphs},
one can also compute tree decompositions of arbitrary structures,
cf.,~\cite[Section~11.3]{FG06}.
In the following, we therefore just assume a tree decomposition is given
as part of the input.  For more information on treewidth, we refer
the reader to surveys such as~\cite{Bod93,Bod98}.

\subsection{MSO Logic}

MSO logic over a vocabulary $\tau$, denoted by $\MSO(\tau)$, is
simultaneously defined over all vocabularies~$\tau$ by induction.
Firstly, for every $p$-ary relation symbol $R \in \tau$ and any nullary symbols
$c_1, \ldots, c_p \in \tau$, $\MSO(\tau)$ contains the \emph{atomic} formula
$R(c_1, \ldots, c_p)$.  If $R$ is unary, we may abbreviate
$R(c)$ as $c \in R$.  Secondly:
\begin{itemize}
\item If $\varphi, \psi$ are in $\MSO(\tau)$, then $\neg \varphi$,
	$\varphi \vee \psi$, and $\varphi \wedge \psi$ are in $\MSO(\tau)$,
\item If $\varphi \in \MSO(\tau \cup \{c\})$ for some nullary symbol~$c$,
	then both, $\forall c \varphi$ and $\exists c \varphi$ are in
	$\MSO(\tau)$.  This is called \emph{first order} or \emph{object quantification}.
\item If $\varphi \in \MSO(\tau \cup \{R\})$ for a unary relation symbol $R$, then
	both, $\forall R \varphi$ and $\exists R \varphi$ are in $\MSO(\tau)$.
	The corresponding case is called \emph{second order} or \emph{set quantification}.
\end{itemize}

Note that we do not distinguish between ``basic'' symbols (contained in a
certain ``base'' vocabulary such as $\taugraph$), and symbols that are 
used as \emph{variables} subject to quantification.  Let $\tau$ be
a vocabulary and $\varphi \in \MSO(\tau)$ be a formula.  Let $\tau'
\subseteq \tau$ be the smallest vocabulary with $\varphi \in
\MSO(\tau')$.  Then we call the symbols in $\unaryrelations(\tau')
\cup \nullary(\tau')$ the \emph{free} symbols of $\varphi$.  Let
$\|\varphi\|$ be the size of a suitable encoding of~$\varphi$.

If $\varphi \in \{\forall c \psi, \forall R \psi, \psi_1 \wedge
\psi_2\}$ for some $c$, $R$, $\psi$, $\psi_1$, and $\psi_2$, we call $\varphi$
\emph{universal}.  Similarly, we call $\varphi$ \emph{existential}
if $\varphi \in \{\exists c \psi, \exists R \psi, \psi_1 \vee
\psi_2\}$.

If $\varphi$ does not contain set quantifiers, then
we say $\varphi$ is \emph{first order} and contained in $\FO(\tau)$.
Note that in particular all atomic formulas of $\MSO(\tau)$ are first
order.  The \emph{quantifier rank} $\qr(\varphi)$ of a formula $\varphi \in
\MSO(\tau)$ denotes the maximum number of nested quantifiers in $\varphi$,
counting both first order and second order quantifiers, and is defined
by induction over the structure of~$\varphi$ as
\begin{itemize}
\item $\qr(\varphi) = 0$ if $\varphi$ is an atomic formula,
\item $\qr(\varphi) = \qr(\neg \varphi)$, 
\item $\qr(\varphi) = \max\{\qr(\psi_1), \qr(\psi_2)\}$ if $\varphi \in \{\psi_1 \wedge \psi_2, \psi_1 \vee \psi_2\}$, and
\item $\qr(\varphi) = \qr(\psi) + 1$ if $\varphi \in \{\forall R \psi, \exists R \psi, \forall c \psi, \exists c \psi\}$.
\end{itemize}

Without loss of generality, we assume throughout the paper that
every formula is in negation normal form, i.e., the negation
symbol~$\neg$ only occurs in front of atomic formulas.  This can
be achieved by a simple rewriting of the formula.

For a fully interpreted $\tau$-structure $\ms A$ and a formula $\varphi \in \mso(\tau)$, we write
$\ms A \models \varphi$ if and only if $\varphi$ \emph{holds} in $\ms A$ or is \emph{true} in $\ms A$
in the classical sense, cf.~\cite{Tar44,EF99}.
We shall do not specify this further, since we will switch to a game-theoretic
characterization in the remainder of this paper, cf., Section~\ref{sec:mcg}.


In~\cite{ALS91}, Extended MSO was introduced.  Here, an MSO-formula over a relational
vocabulary is given together with an evaluation or optimization goal over the unary relation
symbols (set variables).  This principle was furthermore generalized
to semiring homomorphisms in~\cite{CM93}, where satisfying interpretations of
the free relation symbols are to be translated into an appropriate semiring.

In this paper, we shall consider MSO-definable linear optimization problems, also called
\linmso-definable optimization problems.  It is not hard to see that the methods in this
paper extend to other classes of MSO-definable problems, such as counting and enumeration problems.
See, e.g.,~\cite[Chapter~6]{Cou11} for an overview of MSO-definable problems and their algorithmic
applications.

\def\linmsoopt{
\min\Biggl\{\,\sum_{k=1}^l \alpha_k|U_k| \Biggm|
U_i \subseteq A \text{, $1 \le i \le l$, and } 
(\ms A, U_1, \ldots, U_l) \models \varphi \,\Biggr\}
}
\begin{definition}[\linmso-definable Optimization Problem]
Let $\tau$ be a relational vocabulary, $\bar R = \{R_1, \ldots, R_l\} \subseteq \tau$ be
a set of unary relation symbols, $\varphi \in \mso(\tau)$, and $\tau' = \tau \setminus \bar R$.
Let $\alpha_1, \ldots, \alpha_l \in \mathbf Z$ and $\min\,\emptyset := \infty$.

Then we call the problem of, given a $\tau'$-structure $\ms A$, computing
$$\linmsoopt$$ a \emph{\linmso-definable optimization problem}. 
\end{definition}

\begin{example}
\label{ex:mso-examples}
Consider the following formulas:
\begin{align*}
\vc(R) = {} & \forall x \forall y (\neg \adj(x,y) \vee x \in R \vee y \in R) \in \mso(\taugraph \cup \{R\}) \\
\ds(R) = {} & \forall x (x \in R \vee \exists y (y \in R \wedge \adj(x,y))) \in \mso(\taugraph \cup \{R\}) \\
\col = {} & \exists R_1 \exists R_2 \exists R_3 \Bigg( 
\forall x \bigg(\bigvee_{i=1}^3(x \in R_i) \wedge \bigwedge_{i\neq j}
	(\neg x \in R_i \vee \neg x \in R_j)\bigg) \wedge \\
	& \qquad \forall x \forall y\bigg(\neg \adj(x, y) \vee
	\bigwedge_{i=1}^3 (\neg x \in R_i \vee \neg y \in R_i)\bigg)\Bigg) \in \mso(\taugraph)
\end{align*}
Then, given a $\taugraph$-structure $\ms G$,
\begin{align*}
&\min\bigl\{\,|C| \bigm|  C \subseteq A \wedge
	(\ms G, C) \models \vc \,\bigr\}, \\
&\min\bigl\{\,|D| \bigm|  D \subseteq A \wedge (\ms G, D) \models \ds \,\}\text{, and} \\
&\min\bigl\{\,0 \bigm|  \ms G \models \col\,\bigr\}
\end{align*}
encode the well known graph problems \textsc{Minimum Vertex Cover}, \textsc{Minimum
Dominating Set}, and \textsc{3-Colorability}, respectively.

\end{example}

\section{Model Checking Games}
\label{sec:mcg}

The semantics of MSO in the classical sense (cf.~\cite{Tar44,EF99})
can be characterized using a two player pebble game, called the \emph{Hintikka game} or
\emph{model checking game}, cf.~\cite{Hin73,Gra07,Gra11}.

A pebble \emph{game} $\mathcal G = (P, M, P_0, P_1,
p_0)$ between two players, say Player~0 and Player~1, consists of
a finite set~$P$ of \emph{positions}, two disjoint sets $P_0, P_1 \subseteq
P$ assigning positions to the two players, an \emph{initial position}
$p_0 \in P$, and an acyclic binary relation $M \subseteq P
\times P$, which specifies the valid \emph{moves} in the game.  
We only allow moves from positions assigned to one of the two players,
i.e., we require $p \in P_0 \cup P_1$ for all $(p, p') \in M$.
On the other hand, we do allow that positions without outgoing moves
are assigned to players.
Let $|\mathcal G| := |P|$ be the \emph{size} of~$\mathcal G$.

For $p \in P$, we let $\next_{\mathcal G} (p) = \{\,p' \in P \mid (p, p') \in M\,\}$ be
the set of positions reachable from $p$ via a move in~$M$. 
For any position~$p \in \next_{\mathcal G}(p_0)$ we let
$\subgame_{\mathcal G}(p) = (P, M, P_0, P_1, p)$ be a \emph{subgame} of~$\mathcal G$,
which is issued from the new initial position~$p$.
The set of all 
subgames of~$\mathcal G$ is denoted by~$\subgames(\mathcal G)$.
If $\mathcal G$ is clear from the context, we usually omit the subscript and
write~$\next(p)$ and~$\subgame(p)$.

A \emph{play} of~$\mathcal G$ is a maximal sequence $(p_0, \ldots, p_l)$
of positions $p_0, \ldots, p_{l-1} \in P_0 \cup P_1$, such that
between any subsequent positions $p_i$ and $p_{i+1}$ there
is a valid move, i.e., $(p_i, p_{i+1}) \in M$ for $0 \le i \le l-1$.
Such a play is said to have~$l$ \emph{rounds} and to \emph{end} in
position~$p_l$. 

The rules of the game are that in the $i$th round of the play, where $1 \le i \le l -
1$, the player assigned to position $p_i$ has to place a valid
\emph{move}, i.e., has to choose the next position $p_{i+1} \in \next(p_i)$.
If no such position $p_{i+1}$ exists, or the position $p_i$ is not assigned
to either of the players, the play ends.
If the play ends in a position $p_l$ with $p_l
\in P_i$, where $i \in \{0,1\}$, then the other player, Player~$(1-i)$, \emph{wins} the play. 
If, however, the play ends in a position $p_l$ with $p_l \notin P_0 \cup P_1$, then there
is a draw and none of the players wins the play.
The goal of game is to force the other player into a position where
they cannot move.

We say that a player has a \emph{winning strategy on $\mathcal G$},
if and only if they can win every play of the game irrespective
of the choices of the other player.  For instance, Player~0 has a winning
strategy on $\mathcal G$ if and only if either
\begin{itemize}
\item $p_0 \in P_0$ and there is a move $(p_0, p_1) \in M$
	such that Player~0 has a winning strategy on $\subgame_{\mathcal G}(p_1)$; or
\item $p_0 \in P_1$ and Player~0 has a winning strategy
	on $\subgame_{\mathcal G}(p_1)$ for all moves
	$(p_0, p_1) \in M$.  Note that this includes
	the case that Player~1 cannot move at all.
\end{itemize}
A game $\mathcal G$ is said to be \emph{determined} or \emph{well-founded} if
either one of the players has a winning strategy on~$\mathcal G$, otherwise 
$\mathcal G$ is \emph{undetermined}.

We fix two special games $\false$ and $\true$ on which the first player
and the second player, respectively, have winning strategies. 
One can efficiently test whether one of the player has a winning strategy on
a game~$\mathcal G$, cf.,~\cite{Gra07,Gra11}. 
Algorithm~\ref{alg:eval} determines whether one of the players has a winning strategy on a
game~$\mathcal G$ and returns either $\false$ or $\true$ if this is the case.  If
none of the players has a winning strategy, the algorithm returns a corresponding
``proof'', a list of all the plays of~$\mathcal G$ that ended with a draw.

\begin{algorithm}[tb]
Algorithm $\eval(\mathcal G)$
\halign{# \hfil&#\hfil\cr
Input: & A game $\mathcal G = (P, M, P_0, P_1, p_0)$. \cr
}
\medskip \activatealgo \obeylines
[if] $\mathcal G \in \{\true, \false\}$ [then] [return] $\mathcal G$

Let $P' = \{p_0\}$, $M' = \emptyset$, $P_0' = P_0 \cap \{p_0\}$, and $P_1' = P_1 \cap \{p_0\}$.
[for] $p' \in \next(p_0)$ [do]
    \> Let $(P'', M', P_0'', P_1'', p_0'') = \eval(\subgame_{\mathcal G}(p'))$.
    \> Update $P' := P' \cup P''$ and $P_0' := P_0' \cup P_0''$, $P_1' := P_1' \cup P_1''$.
    \> Update $M' := M' \cup M'' \cup \{(p_0', p_0'')\}$ .
Let $\mathcal G' = (P', M', P_0', P_1', p_0)$ and compute $\subgames(\mathcal G')$.
[if] $p_0 \in P_0'$ [then]
    \> [if] $\subgames(\mathcal G') = \{\true\}$ [or] $\subgames(\mathcal G') = \emptyset$ [then] [return] $\true$
    \> [if] $\false \in \subgames(\mathcal G')$ [then] [return] $\false$
[if] $p_0 \in P_1'$ [then]
    \> [if] $\subgames(\mathcal G') = \{\false\}$ [or] $\subgames(\mathcal G') = \emptyset$ [then] [return] $\false$
    \> [if] $\true \in \subgames(\mathcal G')$ [then] [return] $\true$
[return] $\mathcal G'$
\caption{\label{alg:eval} Evaluating a game.}
\end{algorithm}

In the case of the model checking game,
we call the two players the \emph{falsifier} and the \emph{verifier}.
The verifier wants to prove that a 
formula is \emph{true} on a structure (or, the structure \emph{satisfies}
the formula), while the falsifier tries to show that it is \emph{false} (or,
the structure does not satisfy the formula).
The reader may therefore call $\true$ ``true'' and $\false$ ``false''.

\begin{definition}[Model Checking Game]
\label{def:model-checking-game}
The (classical) \emph{model checking game} 
$\mcg(\ms A, \varphi) = (P, M, P_0, P_1, p_0)$
over a fully
interpreted $\tau$-structure~$\ms A$ and a formula~$\varphi \in \mso(\tau)$
is defined by induction over the structure of~$\varphi$ as follows.
Let $p_0 = (\ms A[\bar c^{\ms A}], \varphi)$, where $\bar c = \nullary(\tau)$.
If $\varphi$ is an atomic or negated formula, then
$\mcg(\ms A, \varphi) = (\{p_0\}, \emptyset, P_0, P_1, p_0)$, where
\begin{itemize}
\item $p_0 \in P_0$ if and only if
    \begin{itemize}
	\item $\varphi = R(c_1, \ldots, c_p)$
	    and $(c_1^{\ms A}, \ldots, c_p^{\ms A}) \in R^{\ms A}$, or
	\item $\psi = \neg R(c_1, \ldots, c_p)$
	    and $(c_1^{\ms A}, \ldots, c_p^{\ms A}) \notin R^{\ms A}$.
	\end{itemize}
\item $p_0 \in P_1$ if and only if 
    \begin{itemize}
	\item $\varphi = R(c_1, \ldots, c_p)$ and
	    $(c_1^{\ms A}, \ldots, c_p^{\ms A}) \notin R^{\ms A}$, or
	\item $\psi = \neg R(c_1, \ldots, c_p)$ and
	    $(c_1^{\ms A}, \ldots, c_p^{\ms A}) \in R^{\ms A}$.
	\end{itemize}
\end{itemize}

If $\varphi \in \{\forall R \psi, \exists R \psi\}$ for some relation
symbol~$R$, let $\ms A_U = (\ms A, U)$ for $U \subseteq A$ 
be the $(\tau, R)$-expansion of~$\ms A$
with $R^{\ms A_U} = U$, and
let $\mcg(\ms A_U, \psi) = (P_U, M_U, P_{0,U}, P_{1,U}, p_U)$ be
the corresponding model checking game over~$\ms A_U$ and~$\psi$.
Then $\mcg(\ms A, \varphi) = (P, M, P_0, P_1, p_0)$, where
\begin{itemize}
    \item $P = \{p_0\} \cup \bigcup_{U \subseteq A} P_U$,
    \item $M = \bigcup_{U \subseteq A} (M_U \cup \{(p_0, p_U)\})$,
    \item $P_0 = P_0' \cup \bigcup_{U \subseteq A} P_{0,U}$,
	where $P_0' = \{p_0\}$ iff $\varphi = \forall R \psi$
	and $P_0' = \emptyset$ otherwise,
    \item $P_1 = P_1' \cup \bigcup_{U \subseteq A} P_{1,U}$,
	where $P_1' = \{p_1\}$ iff $\varphi = \exists R \psi$
	and $P_1' = \emptyset$ otherwise.
\end{itemize}

If $\varphi \in \{\forall c \psi, \exists c \psi\}$ for some nullary
symbol~$c$, let $\ms A_a = (\ms A, a)$ be the 
$(\tau, c)$-expansion of~$\ms A$ with $c^{\ms A_a} = a \in A$, and
let $\mcg(\ms A_a, \psi) = (P_a, M_a, P_{0,a}, P_{1,a}, p_a)$ be
the corresponding model checking game over~$\ms A_a$ and~$\psi$.
Then $\mcg(\ms A, \varphi) = (P, M, P_0, P_1, p_0)$, where
\begin{itemize}
    \item $P = \{p_0\} \cup \bigcup_{a \in A} P_a$,
    \item $M = \bigcup_{a \in A} (M_a \cup \{(p_0, p_a)\})$,
    \item $P_0 = P_0' \cup \bigcup_{a \in A} P_{0,a}$,
	where $P_0' = \{p_0\}$ iff $\varphi = \forall c \psi$
	and $P_0' = \emptyset$ otherwise,
    \item $P_1 = P_1' \cup \bigcup_{a \in A} P_{1,a}$,
	where $P_1' = \{p_1\}$ iff $\varphi = \exists c \psi$
	and $P_1' = \emptyset$ otherwise.
\end{itemize}

If $\varphi \in \{\psi_1 \wedge \psi_2, \psi_1 \vee \psi_2\}$,
let
$\mcg(\ms A, \psi) = (P_\psi, M_\psi, P_{0,\psi}, P_{1,\psi}, p_\psi)$ be
the model checking game over~$\ms A$ and~$\psi \in \{\psi_1, \psi_2\}$.
Then $\mcg(\ms A, \varphi) = (P, M, P_0, P_1, p_0)$, where
\begin{itemize}
    \item $P = \{p_0\} \cup \bigcup_{\psi \in \{\psi_1,\psi_2\} } P_\psi$,
    \item $M = \bigcup_{\psi \in \{\psi_1,\psi_2\}}  (M_\psi \cup \{(p_0, p_\psi)\})$,
    \item $P_0 = P_0' \cup \bigcup_{\psi \in \{\psi_1,\psi_2\}}  P_{0,\psi}$,
	where $P_0' = \{p_0\}$ iff $\varphi = \psi_1 \wedge \psi_2$
	and $P_0' = \emptyset$ otherwise,
    \item $P_1 = P_1' \cup \bigcup_{\psi \in \{\psi_1,\psi_2\}}  P_{1,\psi}$,
	where $P_1' = \{p_1\}$ iff $\varphi = \psi_1 \vee \psi_2$
	and $P_1' = \emptyset$ otherwise.
\end{itemize}
\end{definition}

Note that the falsifier is the \emph{universal} player and moves on universal formulas, while
the verifier is the \emph{existential} player and moves on existential formulas. 
Furthermore, if the structure~$\ms A$ is empty, then, by definition,
$\ms A \models \forall c \psi$ and $\ms A \not\models \exists c \psi$ for all $\psi$.
In the model checking game, this corresponds to the case that there
are no moves from the current position.  Consequently, the play ends and
the player assigned to this position looses.
On non-empty structures, each play ends in an atomic or negated
atomic formula.  The goal of the verifier is to make the play end in
a position~$(\ms A', \psi)$ with $\ms A' \models \psi$, and conversely
the goal of the falsifier is to force the play into an ending position
$(\ms A', \psi)$ with $\ms A' \not\models \psi$.  It is well-known that
the classical model checking game is well-founded~\cite{Hin73} and that
the verifier has a winning strategy on $\mcg(\ms A, \varphi)$ if and
only if $\ms A \models \varphi$, see, e.g.,~\cite{Gra07}.

\subsection{An Extension of the Classical Model Checking Game}

We shall now consider an extension of the model checking game
that has the following two central properties:
\begin{itemize}
\item It is defined for partially interpreted structures;  and
\item it is ``well-defined'' under taking the union of structures in the
	sense that if one of the players has a winning strategy on the
	game on $\ms A$ and $\varphi$, then the same player has a winning
	strategy in the game on $\ms A \cup \ms B$ and $\varphi$ for all
	structures~$\ms B$ compatible with~$\ms A$.
\end{itemize}

Before we give the formal definition of the new game, let us briefly
mention why we require these properties: Recall that we want to use
the model checking game $\mcg(\ms A, \varphi)$ to decide algorithmically
whether a $\tau$-structure~$\ms A$ holds on a formula~$\varphi \in
\mso(\tau)$.  If $\varphi$ contains set quantifiers, then there is a 
number of positions in $\mcg(\ms A, \varphi)$ that grows exponentially with
the size of~$A$.  In order to avoid exponential running time on structures
of bounded treewidth, a tree decomposition~$(\mathcal T, \mathcal X)$ of~$\ms A$, where $\mathcal
T = (T, F)$, is traversed bottom-up by a dynamic programming
algorithm.  
At a node~$i \in T$, we only consider the substructure $\ms A_i$ of $\ms A$.
Let $\ms A'$ be some expansion of $\ms A$.  Then $\ms A[A_i]$ is in general
not fully interpreted, which explains the first requirement.

For the second requirement, note that for each $i \in T$ there is
a $\tau$-structure~$\ms B_i$, such that $\ms A$ can be written as
$\ms A = \ms A_i \cup \ms B_i$.  The structure $\ms B_i$ is sometimes
called the ``future'' of~$\ms A_i$ in the literature.  Therefore, if one of the players
has a winning strategy in the game on $\ms A_i$ and $\varphi$,
we require that the same player has a winning strategy on $\ms A = \ms A_i \cup \ms B_i$ and
$\varphi$.

In order to make the inductive construction work, we additionally need to
distinguish the nodes in the ``current'' bag $X_i$ of the tree decomposition.
The game therefore additionally depends on a given set $X = X_i \subseteq A$.

\begin{definition}[Extended Model Checking Game]
\label{def:extended-model-checking-game}
\label{def:outcome}

The \emph{extended model checking game} 
$\emcg(\ms A, X, \varphi) = (P, M, P_0, P_1, p_0)$
over a $\tau$-structure~$\ms A$, a set
$X \subseteq A$, and a formula~$\varphi \in \mso(\tau)$
is defined by induction over the structure of~$\varphi$ as follows.
Let $p_0 = (\ms A[X \cup \bar c^{\ms A}], X, \varphi)$, where $\bar c = \nullary(\tau)$.
If $\varphi$ is an atomic or negated formula, then
$\emcg(\ms A, \varphi) = (\{p_0\}, \emptyset, P_0, P_1, p_0)$, where
\begin{itemize}
\item $p_0 \in P_0$ if and only if either
    \begin{itemize}
    	\item $\varphi = R(c_1, \ldots, c_p)$, such that
		$\{c_1, \ldots, c_p\} \subseteq \interpreted(\ms A)$,
		and $(c_1^{\ms A}, \ldots, c_p^{\ms A}) \in R^{\ms A}$, or
	\item $\varphi = \neg R(c_1, \ldots, c_p)$, such that
		$\{c_1, \ldots, c_p\} \subseteq \interpreted(\ms A)$,
	    	and $(c_1^{\ms A}, \ldots, c_p^{\ms A}) \notin R^{\ms A}$.
    \end{itemize}
\item $p_0 \in P_1$ if and only if either
    \begin{itemize}
	\item $\varphi = R(c_1, \ldots, c_p)$, such that
	    $\{c_1, \ldots, c_p\} \subseteq \interpreted(\ms A)$,
	    and $(c_1^{\ms A}, \ldots, c_p^{\ms A}) \notin R^{\ms A}$, or
	\item $\varphi = \neg R(c_1, \ldots, c_p)$,
	    such that $\{c_1, \ldots, c_p\} \subseteq \interpreted(\ms A)$,
	    and $(c_1^{\ms A}, \ldots, c_p^{\ms A}) \in R^{\ms A}$.
    \end{itemize}
\end{itemize}

If $\varphi \in \{\forall R \psi, \exists R \psi\}$ for some relation
symbol~$R$, or $\varphi \in \{\psi_1 \wedge \psi_2, \psi_1 \vee \psi_2\}$,
then
$\emcg(\ms A, X, \varphi)$ is defined analogously to~$\mcg(\ms A, \varphi)$.

If $\varphi \in \{\forall c \psi, \exists c \psi\}$ for some nullary
symbol~$c$, let $\ms A_u = (\ms A, u)$ be the 
$(\tau, c)$-expansion of~$\ms A$ with $c^{\ms A_u} = u \in A \cup \{\nil\}$, and
let $\emcg(\ms A_u, X, \psi) = (P_u, M_u, P_{0,u}, P_{1,u}, p_u)$ be
the corresponding extended model checking game over~$\ms A_u$ and~$\psi$.
Then $\emcg(\ms A, X, \varphi) = (P, M, P_0, P_1, p_0)$, where
\begin{itemize}
    \item $P = \{p_0\} \cup \bigcup_{u \in A \cup \{\nil\}} P_u$,
    \item $M = \bigcup_{u \in A \cup \{\nil\}} (M_u \cup \{(p_0, p_u)\})$,
    \item $P_0 = P_0' \cup \bigcup_{u \in A \cup \{\nil\}} P_{0,u}$,
	where $P_0' = \{p_0\}$ iff $\varphi = \forall c \psi$
	and $P_0' = \emptyset$ otherwise,
    \item $P_1 = P_1' \cup \bigcup_{u \in A \cup \{\nil\}} P_{1,u}$,
	where $P_1' = \{p_1\}$ iff $\varphi = \exists c \psi$
	and $P_1' = \emptyset$ otherwise.
\end{itemize}

\end{definition}

\bigskip
For the games we consider throughout this paper, one can derive from a position~$p \in P$ whether
$p \in P_0$ or $p \in P_1$ (cf., the definitions of~$\mcg$ and $\emcg$).
To avoid cluttered notation, we shall therefore usually omit the sets~$P_0$ and $P_1$ from the
tuple $(P, M, P_0, P_1, p_0)$ and identify games with the triple~$(P, M, p_0)$.
Figure~\ref{fig:unfolding} shows a simplified schematic of an extended model checking game and
the result after an application of the evaluation algorithm~$\eval$.

\begin{figure}[tbp]
\begin{center}
\includegraphics[scale=0.8]{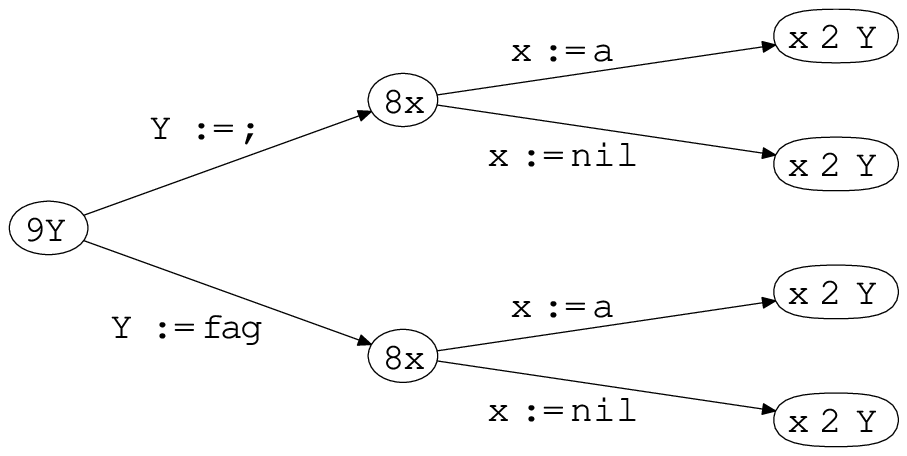}

\includegraphics[scale=0.8]{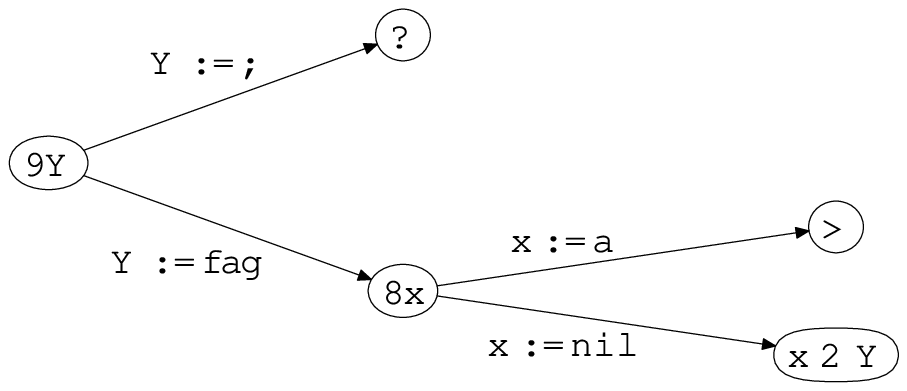}
\end{center}
\caption{Top: simplified schematic of $\emcg(\ms A, \emptyset, \varphi)$ for the structure $\ms A$ with universe $A = \{a\}$ 
and $\varphi = \exists Y \forall x (x \in Y)$.  Bottom: $\eval(\emcg(\ms A, \emptyset, \varphi))$.  The
lower branch witnesses a play that ends with a draw.%
\label{fig:unfolding}
\label{fig:eval}}
\end{figure}

If $\ms A$ is a fully interpreted structure, $\mcg(\ms A, \varphi)$ can be embedded into $\emcg(\ms A, X, \varphi)$
such that for each play of $\mcg(\ms A, \varphi)$ there is a corresponding, equivalent play of~$\emcg(\ms A, X, \varphi)$.
Algorithm~\ref{alg:convert} effectively computes this embedding (Lemma~\ref{lem:convert-me}).
Furthermore, if $\emcg(\ms A, X, \varphi)$ is determined, then so is
$\mcg(\ms A, \varphi)$ (Lemma \ref{lem:emcg-determines-mcg}).

\begin{lemma}
\label{lem:convert-me}
Let $\ms A$ be a fully interpreted $\tau$-structure, $X \subseteq A$, and $\varphi \in \mso(\tau)$.
Then, using Algorithm~\ref{alg:convert}, we have
$$
\mcg(\ms A, \varphi) = \convert(\emcg(\ms A, X, \varphi)).
$$
\end{lemma}
\begin{proof}
The proof is an induction over the structure of~$\varphi$.
For atomic or negated atomic formulas, the statement trivially holds by
definition of~$\mcg(\ms A, \varphi)$, since $\subgames(\emcg(\ms A, X, \varphi)) = \emptyset$.
Let $\mathcal G = \emcg(\ms A, X, \varphi) = (P, M, p_0)$.

Let $\varphi \in \{\forall R \psi, \exists R \psi\}$ or
$\varphi \in \{\psi_1 \wedge \psi_2, \psi_1 \vee \psi_2\}$ and $\psi \in \{\psi_1, \psi_2\}$
and consider $p = (\ms H, X, \psi) \in \next_{\mathcal G}(p_0)$.
We have $\subgame_{\mathcal G}(p) = \emcg(\ms A', X, \psi)$,
where either $\ms A' = (\ms A, U)$ is an $(\tau, R)$-expansion
of~$\ms A$ for some $U \subseteq A$, or $\ms A' = \ms A$, respectively.
Since $\ms A$ is fully interpreted, $\ms A'$ is fully interpreted, and we
obtain $\mcg(\ms A', \psi) = \convert(\emcg(\ms A', X, \psi))$ by the induction hypothesis.

If otherwise $\varphi \in \{\forall c \psi, \exists c \psi\}$, consider
$p = (\ms H, X, \psi) \in \next_{\mathcal G}(p_0)$.
By definition, $\subgame_{\mathcal G}(p) = \emcg(\ms A', X, \psi)$, where
$\ms A'$ is a $(\tau, c)$-expansion of~$\ms A$ with $c^{\ms A'} \in A \cup \{\nil\}$.
If all constant symbols are interpreted in $\ms H$, then
$c^{\ms A'} \neq \nil$, i.e., $\ms A'$ is fully interpreted.
By the induction hypothesis we get $\mcg(\ms A', \psi) = \convert(\emcg(\ms A', X, \psi))$.

Together, the statement follows.
\end{proof}

\begin{algorithm}[tb]
Algorithm $\convert(\mathcal G)$

\halign{# \hfil&#\hfil\cr
Input: & A game $\mathcal G = (P, M, p_0)$. \cr
}
\medskip \activatealgo \obeylines
[if] $\mathcal G \in \{\true, \false\}$ [then] [return] $\mathcal G$.
Let $p_0 = (\ms H, X, \varphi)$ and $\bar c = \nullary(\vocabulary(\ms H))$.
Let $p_0' = (\ms H[\bar c^{\ms H}], \varphi)$, $P' = \{p_0'\}$, and $M' = \emptyset$.
[for] $p_1 = (\ms H_1, X, \psi) \in \next(p_0)$ s.t.\ $\ms H_1$ is fully interpreted [do]
    \> Let $(P_1', M_1', p_1') = \convert(\subgame_{\mathcal G}(p_1))$.
    \> Update $P' := P' \cup P_1'$ and $M' := M' \cup M_1' \cup \{(p_0', p_1')\}$.
[return] $\mathcal (P', M', p_0')$

\caption{\label{alg:convert} Converting $\emcg$ to \mcg}
\end{algorithm}

We now prove that if an extended model game is determined, then the corresponding player can win the game without using
any further ``$\nil$-moves''.  This will be useful in the proof of Lemma~\ref{lem:emcg-determines-mcg}.

\begin{lemma}
\label{lem:partially-replaced}
Let $\ms A_1$ and $\ms A_2$ be $\tau$-structures with $A_1 = A_2$ and $c \in \nullary(\tau)$, such
that $c^{\ms A_1} = \nil$, $R^{\ms A_1} = R^{\ms A_2}$ for all $R \in \relations(\tau)$ and
$d^{\ms A_1} = d^{\ms A_2}$ for all $d \in \nullary(\tau) \setminus \{c\}$.
Let $\varphi \in \mso(\tau)$.

If $\eval(\emcg(\ms A_1, X, \varphi)) \in \{\true, \false\}$, then
$A_1 \neq \emptyset$ and
$$
\eval(\emcg(\ms A_1, X, \varphi)) = \eval(\emcg(\ms A_2, X, \varphi)).
$$
\end{lemma}

Before we give the formal proof, consider the following high-level argument:
Suppose that $\eval(\emcg(\ms A_1, X, \varphi)) = \true$.
Then there is at least one play of the game $\emcg(\ms A_1, X, \varphi)$ that is
won by the verifier.
Consider an arbitrary play $(p_0, \ldots, p_l)$ won by the verifier
and let $p_l = (\ms H, X, \psi)$.  Since $p_l$ is assigned to the falsifier,
all constant symbols occurring in $\psi$ are interpreted and hence different
from~$c$.  The verifier can therefore win the game without depending on formulas
where~$c$ occurs.

\begin{proof}
The proof is an induction over the structure of~$\varphi$.

Let $\eval(\emcg(\ms A_1, X, \varphi)) \in \{\true, \false\}$.  If $\varphi$
is an atomic or negated formula, say $\varphi = R(c_1, \ldots, c_p)$,
then $\{c_1, \ldots, c_p\} \subseteq \interpreted(\ms A_1)$.
Therefore, $A_1 \neq \emptyset$ and for all $1 \le i \le p$, we have
$c \neq c_i$ and $c_i^{\ms A_1} = c_i^{\ms A_2}$, which implies
$\eval(\emcg(\ms A_1, X, \varphi)) = \eval(\emcg(\ms A_2, X, \varphi))$.

If $\varphi \in \{\forall R \psi, \exists R \psi\}$ for a relation symbol~$R$,
let $U \subseteq A$ and $\ms A_1'$, $\ms A_2'$ be the $(\tau, R)$-expansions
of~$\ms A_1$ and $\ms A_2$, respectively, with
$R^{\ms A_1'} = R^{\ms A_2'} = U$.
Then
by the induction hypothesis
$\eval(\emcg(\ms A_1', X, \psi)) = \eval(\emcg(\ms A_2', X, \psi))$
if $\eval(\emcg(\ms A_1', X, \psi)) \in \{\true, \false\}$.

Similarly, if $\varphi \in \{\forall d \psi, \exists d \psi\}$ for a nullary symbol~$d$,
let $\ms A_1'$ and $\ms A_2'$ be $(\tau, d)$-expansions of~$\ms A_1$ and $\ms A_2$, respectively,
such that $d^{\ms A_1'} = d^{\ms A_2'}$.
Then by the induction hypothesis
$\eval(\emcg(\ms A_1', X, \psi)) = \eval(\emcg(\ms A_2',\allowbreak X, \psi))$,
if $\eval(\emcg(\ms A_1', X, \psi)) \in \{\true, \false\}$.

Finally, if $\varphi \in \{\psi_1 \wedge \psi_2, \psi_1 \vee \psi_2\}$,
then from $\eval(\emcg(\ms A_1, X, \psi)) \in \{\true, \false\}$, where
$\psi \in \{\psi_1, \psi_2\}$, we get
$\eval(\emcg(\ms A_1', X, \psi)) = \eval(\emcg(\ms A_2', X, \psi))$.

Together, the statement of the lemma follows.
\end{proof}

We can now prove that if some player has a winning strategy in the extended model checking
game, then the same player has a winning strategy in the classical model checking game.

\begin{lemma}
\label{lem:emcg-determines-mcg}
Let $\ms A$ be a fully interpreted $\tau$-structure, $X \subseteq A$, and $\varphi \in \mso(\tau)$.
If $\eval(\emcg(\ms A, X, \varphi)) \in \{\true, \false\}$, then
$$
\eval(\mcg(\ms A, \varphi)) = \eval(\emcg(\ms A, X, \varphi)).
$$
\end{lemma}
\begin{proof}
The proof is an induction over the structure of~$\varphi$.

Suppose $\eval(\emcg(\ms A, X, \varphi)) = \true$ (the case $\false$ is shown
analogously).  If $\varphi$ is an atomic or negated atomic formula, then the statement
clearly holds.
If $\varphi = \psi_1 \wedge \psi_2$, then for each $\psi \in \{\psi_1, \psi_2\}$ we
have $\eval(\emcg(\ms A, X, \psi)) = \true$.  This implies $\eval(\mcg(\ms A, \psi)) = \true$
by the induction hypothesis, and therefore~$\eval(\mcg(\ms A, \varphi)) = \true$.

Similarly, if $\varphi = \forall R \psi$ for a relation symbol~$R$,
then $\eval(\emcg(\ms A', X, \psi)) = \true$ for each $(\tau, R)$-expansion $\ms A'$ of $\ms A$,
each of which is fully interpreted.
We get $\eval(\mcg(\ms A', \psi)) = \true$ by the induction hypothesis.
Hence, $\eval(\mcg(\ms A, \varphi)) = \true$.

If $\varphi = \forall c \psi$ for a nullary symbol~$c$,
then $\eval(\emcg(\ms A', X, \psi)) = \true$ for each fully interpreted
$(\tau, c)$-expansion $\ms A'$ of $\ms A$.
This implies $\eval(\mcg(\ms A', \psi)) = \true$ by the induction hypothesis,
and therefore~$\eval(\mcg(\ms A, \varphi)) = \true$.

If $\varphi = \psi_1 \vee \psi_2$, then there is $\psi \in \{\psi_1, \psi_2\}$ 
with $\eval(\emcg(\ms A, X, \psi)) = \true$.  We get
$\eval(\mcg(\ms A, \psi)) = \true$
by the induction hypothesis, and therefore~$\eval(\mcg(\ms A, \varphi)) = \true$.

Similarly, if $\varphi = \exists R \psi$ for a relation symbol~$R$,
then there is a $(\tau, R)$-expansion $\ms A'$ of $\ms A$ with
$\eval(\emcg(\ms A', X, \psi)) = \true$.
Since $\ms A$ is fully interpreted, $\ms A'$ is fully interpreted.
Using the induction hypothesis, we have $\eval(\mcg(\ms A', \psi)) = \true$
and therefore~$\eval(\mcg(\ms A, \varphi)) = \true$.

Finally, if $\varphi = \exists c \psi$ for a nullary symbol~$c$, then
there is a $(\tau, c)$-expansion $\ms A'$ of $\ms A$ with $\eval(\emcg(\ms
A', X, \psi)) = \true$. 
By Lemma~\ref{lem:partially-replaced}, we can assume~$c^{\ms A} \neq \nil$.
Then $\ms A'$ is fully interpreted and we get $\eval(\mcg(\ms A', \psi)) = \true$
by the induction hypothesis.  Therefore
$\eval(\mcg(\ms A, \varphi)) = \true$.
\end{proof}

We can significantly strengthen this statement further: If $\emcg(\ms
A, X, \varphi)$ is determined, then $\emcg(\ms A \cup \ms B,
X, \varphi)$ is also determined for \emph{all}~$\ms B$ compatible with~$\ms
A$.  Note that the union $\ms A \cup \ms B$ arises on \emph{join} or
\emph{introduce} nodes~$i$ of the tree decomposition, where $X =
X_i$ is the current bag, cf., Figure~\ref{fig:join}.

Recall, for instance, the example \textsc{3-Colorability} from the introduction:  If 
a subgraph $\ms A'$ of a graph $\ms A$ is not three-colorable, then clearly
$\ms A$ is not three-colorable either.
The following lemma formalizes this observation.

Let us give a brief high-level explanation before we state the lemma and
give its proof.  Roughly speaking, if $\mathcal G = \emcg(\ms A, X, \varphi)$ is
determined, then moves to objects $b \in B \setminus A$ in
$\mathcal G' = \emcg(\ms A \cup \ms B, X, \varphi)$
are either ``irrelevant'' for a player's strategy
or already ``sufficiently'' captured by moves to $\nil$
(cf., Lemma~\ref{lem:partially-replaced}).
If therefore one of the players, say the falsifier, has a
winning strategy in $\mathcal G$, then in some sense this winning
strategy carries over to $\mathcal G'$.
In the case of \textsc{3-Colorability}, if $\ms A$ is not three-colorable, then 
the falsifier has a winning strategy on~$\emcg(\ms A, X, \col)$:
No matter which three sets the verifier chooses, either these sets are not a partition
or not independent sets.  In either case there are witnessing vertices
that the falsifier can choose.  Thus, no matter which subsets
the verifier chooses in $\mathcal G' = \emcg(\ms A \cup \ms B, X,
\col)$, the falsifier can then choose the same witnessing vertices
to win each play of~$\mathcal G'$.

\begin{lemma}[Introduce]
\label{lem:emcg-eval-introduce}
Let $\ms A$ and $\ms B$ be compatible $\tau$-structures with $B = A \uplus \{b\}$.
Let $X \subseteq A$ and $\varphi \in \mso(\tau)$.  Let $\mathcal G = \emcg(\ms A, X, \varphi)$ and
$\mathcal G' = \emcg(\ms B, X \cup \{b\}, \varphi)$.

\begin{enumerate}
\item If $\eval(\mathcal G) = \true$, then $\eval(\mathcal G') = \true$.
\item If $\eval(\mathcal G) = \false$, then $\eval(\mathcal G') = \false$.
\end{enumerate}
\end{lemma}
\begin{proof}
We prove the lemma by induction over the structure of~$\varphi$.  Let $\bar c = \nullary(\tau)$.
Let $\mathcal G =
(P, M, p_0)$ and $\mathcal G' = (P', M', p_0')$ with $p_0 = (\ms H, X,
\varphi)$ and $p_0' = (\ms H', X \cup \{b\}, \varphi)$, where $\ms H = \ms
A[X \cup \bar c^{\ms A}]$ and $\ms H' = \ms B[X \cup \{b\} \cup \bar c^{\ms B}]$.
Suppose $\eval(\mathcal G) = \true$ (the second case $\eval(\mathcal G) =
\false$ is proven analogously).

Let $\varphi = R(c_1, \ldots, c_p)$ or $\varphi = \neg
R(c_1, \ldots, c_p)$ for a relation symbol $R \in \tau$.
We have $\eval(\mathcal G) = \true$, and hence, by definition 
$c_i \in \interpreted(\ms A)$ for all $1 \le i \le p$.
Here, $c_i^{\ms H} = c_i^{\ms A} = c_i^{\ms B} = c_i^{\ms H'}$
for all $1 \le i \le p$, since $\ms A$ and $\ms B$ are compatible, and
therefore $R^{\ms H} = R^{\ms H'} \cap H^p$, since $H = H' \setminus \{b\}$.
Hence, $(c_1^{\ms H}, \ldots, c_p^{\ms H}) \in
R^{\ms H}$ if and only if $(c_1^{\ms H'}, \ldots, c_p^{\ms H'}) \in R^{\ms H'}$,
and thus $\eval(\mathcal G') = \true$.

Assume now that $\varphi = \psi_1 \wedge \psi_2$ or $\varphi = \psi_1
\vee \psi_2$.  By definition, for each $\psi \in \{\psi_1, \psi_2\}$
there is a subgame $\mathcal G_\psi = \emcg(\ms A, X, \psi) \in \subgames(\mathcal G)$
and a subgame $\mathcal G'_\psi = \emcg(\ms B, X \cup \{b\}, \psi) \in \subgames(\mathcal G')$.
By the induction hypothesis, $\eval(\mathcal G'_\psi) = \true$ if $\eval(\mathcal G_\psi) = \true$,
and hence $\eval(\mathcal G') = \true$ if $\eval(\mathcal G) = \true$.

If $\varphi = \forall R \psi$ or $\varphi = \exists R
\psi$, then for each $U \subseteq A$ there is a subgame
$\mathcal G_{U} = \emcg((\ms A, U), X, \psi)) \in \subgames(\mathcal G)$,
and for each $U' \subseteq B$ there is a subgame
$\mathcal G_{U'}' = \emcg((\ms B, U'), X \cup \{b\}, \psi)) \in \subgames(\mathcal G')$.

If $\varphi = \forall R \psi$, consider an arbitrary $U' \subseteq B$
and let $U = U' \setminus\{b\}$. 
We know, by definition of $\eval(\mathcal G)$, that $\eval(\mathcal G_U)
= \true$.  Furthermore, $(\ms A, U)$ and $(\ms B, U')$ are
compatible, and therefore, by the induction hypothesis, also $\eval(\mathcal G'_{U'})
= \true$.  Therefore, $\eval(\mathcal G'_{U'}) = \true$ for all
$U' \subseteq B$, and hence $\eval(\mathcal G') = \true$.

If otherwise $\varphi = \exists R \psi$, then there is some
$U \subseteq A$ such that $\eval(\mathcal G_U)
= \true$.   Since $(\ms A, U)$ and $(\ms B, U)$ are compatible, 
$\eval(\mathcal G'_U) = \true$ by the induction hypothesis.
Therefore, $\eval(\mathcal G') = \true$.

If $\varphi = \forall c \psi$, consider an arbitrary $(\tau, c)$-expansion~$\ms B'$
of~$\ms B$ and let $\ms A' := \ms B[A]$.
Note that if $c^{\ms B'} \neq b$, then $c^{\ms A'} = c^{\ms B'} \in A$, and if
$c^{\ms B'} = b$ or $c^{\ms B'} = \nil$, then $c^{\ms A'} = \nil$.
In either case, $\ms A'$ and $\ms B'$ are compatible.
We know, by definition of $\eval(\mathcal G)$,
that $\eval(\emcg(\ms A', X, \psi)) = \true$.  Hence, by the induction hypothesis,
$\eval(\emcg(\ms B', X \cup \{b\}, \psi)) = \true$. 
All in all, $\eval(\mathcal G') = \true$.

Assume now that $\varphi = \exists c \psi$.  Since
$\eval(\mathcal G) = \true$, we know that there is a $(\tau, c)$-expansion~$\ms A'$ of
$\ms A$, such that $\eval(\emcg(\ms A', X, \psi)) = \true$.
Let $\ms B'$ the $(\tau, c)$-expansion of $\ms B$ with
$c^{\ms B'} = c^{\ms A'}$.  Then $\ms A'$ and $\ms B'$ are
compatible, and using the induction hypothesis as above, we obtain
$\eval(\mathcal G') = \true$.
\end{proof}

\begin{figure}[tbp]
\hfill
\includegraphics[scale=0.8]{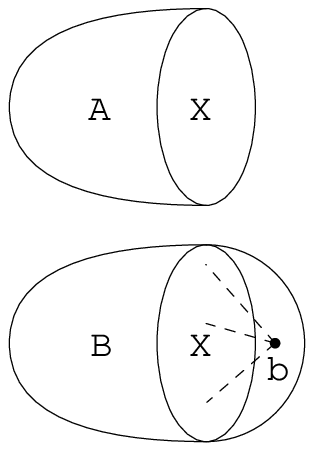}
\hfill
\includegraphics[scale=0.8]{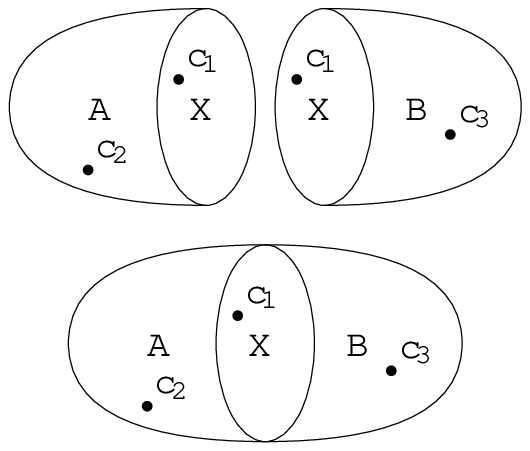}
\hfill
\mbox{}
\caption{Introduce (left):  If $\ms A$ and $\ms B$ are such that $\ms A = \ms B[A]$,
then winning strategies for $\emcg(\ms A, X, \varphi)$ carry over to $\emcg(\ms B, X \cup \{b\}, \varphi)$. 
Join/union (right): If $\ms A$ and $\ms B$ are compatible,
	then winning strategies for $\emcg(\ms A, X, \varphi)$
	carry over to $\emcg(\ms A \cup \ms B, X, \varphi)$. %
	\label{fig:join}
}
\end{figure}

\begin{corollary}
\label{cor:introduce-join}
Let $\ms A$ and $\ms B$ be compatible $\tau$-structures with $A \subseteq B$.
Let $X \subseteq A$ and $\varphi \in \mso(\tau)$.  Let $\mathcal G = \emcg(\ms A, X, \varphi)$ and
$\mathcal G' = \emcg(\ms B, X \cup (B \setminus A), \varphi)$.
\begin{enumerate}
\item If $\eval(\mathcal G) = \true$, then $\eval(\mathcal G') = \true$.
\item If $\eval(\mathcal G) = \false$, then $\eval(\mathcal G') = \false$.
\end{enumerate}
\end{corollary}
\begin{proof}
We use Lemma~\ref{lem:emcg-eval-introduce} and induction over $|B \setminus A|$.
Let $\eval(\mathcal G) \in \{\true, \false\}$.

If $B \setminus A = \emptyset$ and therefore $\ms A = \ms B$, the statement clearly holds.
Otherwise, consider $b \in B \setminus A$ and let $\ms A' = (\ms A \cup \ms B)[A \cup \{b\}]$.
From $\eval(\mathcal G) \in \{\true, \false\}$ we get
$\eval(\emcg(\ms A', X \cup \{b\}, \varphi) = \eval(\mathcal G)$ by
Lemma~\ref{lem:emcg-eval-introduce}.

We can now use the induction hypothesis on $\ms A'$, $\ms B$ and $X \cup \{b\}$, since
$\ms A'$ and $\ms B$ a compatible and $|B \setminus A'| < |B \setminus A|$,
and obtain $\eval(\mathcal G') = \eval(\mathcal G)$.
\end{proof}

The \emph{forget} operation at a node~$i$ of a tree decomposition does not change the underlying
structure~$\ms A_i$.  It is therefore not surprising that any winning strategies carry over.

\begin{lemma}[Forget]
\label{lem:emcg-eval-forget}
Let $\ms A$ be a $\tau$-structure, $X' \subseteq X \subseteq A$ and $\varphi \in \mso(\tau)$.
Let $\mathcal G = \emcg(\ms A, X, \varphi)$ and
$\mathcal G' = \emcg(\ms A, X', \varphi)$.
\begin{enumerate}
\item If $\eval(\mathcal G) = \true$, then $\eval(\mathcal G') = \true$.
\item If $\eval(\mathcal G) = \false$, then $\eval(\mathcal G') = \false$.
\end{enumerate}
\end{lemma}
\begin{proof}
Let $\mathcal G = (P, M, P_0, P_1, p_0)$ and $\mathcal G' = (P', M', P_0', P_1', p_0')$.
It is not hard to see that $\mathcal G$ and $\mathcal G'$ are almost identical,
the only difference being slightly differently labeled positions:
By definition, $p_0 = (\ms H, X, \varphi)$ and $p_0' = (\ms H', X', \varphi)$,
where $\ms H = \ms A[X \cup c^{\ms A}]$ and $\ms H' = \ms A[X' \cup c^{\ms A}]$.
In particular, $p_0 \in P_i$ if and only if $p_0' \in P_i'$, where $i \in \{1,2\}$.
By induction over the structure of $\varphi$, the claim then easily follows.
\end{proof}

Finally we show that the same holds for \emph{join} nodes of a tree decomposition.
Note that the corresponding operation on structures is the union.

\begin{lemma}[Join/Union]
\label{lem:emcg-eval-union}
Let $\ms A, \ms B$ be compatible $\tau$-structures, $X = A \cap B$, and
$\varphi \in \mso(\tau)$.  Let 
$\mathcal G = \emcg(\ms A, X, \varphi)$
and
$\mathcal G' = \emcg(\ms A \cup \ms B, X, \varphi).$
\begin{enumerate}
\item If $\eval(\mathcal G) = \true$, then $\eval(\mathcal G') = \true$.
\item If $\eval(\mathcal G) = \false$,  then $\eval(\mathcal G') = \false$.
\end{enumerate}
\end{lemma}
\begin{proof}
Let $\eval(\emcg(\ms A, X, \varphi)) \in \{\true, \false\}$.
By Corollary~\ref{cor:introduce-join}, 
$\eval(\emcg(\ms A \cup \ms B, X \cup (B \setminus A), \varphi))
= \eval(\emcg(\ms A, X, \varphi))$.
The claim then immediately follows by Lemma~\ref{lem:emcg-eval-forget}.
\end{proof}

\section{Reducing the Size of Games}

In this section we show that for every game
$\mathcal G = (P, M, P_0) = \emcg(\ms A, X, \varphi)$ one can construct a 
game $\mathcal G' = (P', M', p_0)$ such that $\eval(\mathcal G) = \eval(\mathcal G')$
if $\eval(\mathcal G) \in \{\true, \false\}$, 
but $P' \subseteq P$ and $M' \subseteq M'$ are typically much smaller than $P$ and $M$.
This will be crucial for obtaining the desired running times of our algorithm.
We first define a suitable notion of \emph{equivalence} between games.

\begin{definition}[Equivalent Games]
\label{def:equivalent-games}
We say that two positions $p_1, p_2$ are \emph{equivalent}, denoted
by $p_1 \cong p_2$ iff
\begin{itemize}
\item $p_1 = (\ms H_1, X, \varphi)$ and $p_2 = (\ms H_2, X, \varphi)$ for some formula
    $\varphi$ and set $X \subseteq H_1 \cap H_2$,
\item there is an isomorphism $h\colon H_1 \to H_2$ between $\ms H_1$ and $\ms H_2$, such that
    $h(a) = a$ for all $a \in X$.
\end{itemize}

We say that two games $\mathcal G_1 = (P_1, M_1, p_1)$ and $\mathcal G_2 = (P_2, M_2, p_2)$ 
are \emph{equivalent}, denoted by $\mathcal G_1 \cong \mathcal G_2$,
if $p_1 \cong p_2$ and there is a bijection $\pi\colon
\subgames(\mathcal G_1) \to \subgames(\mathcal G_2)$,
such that $\mathcal G' \cong \pi(\mathcal G')$ for all $\mathcal G' \in \subgames(\mathcal G_1)$.
\end{definition}

\begin{algorithm}[tb]
Algorithm $\reduce(\mathcal G)$
\halign{# \hfil&#\hfil\cr
Input: & A game $\mathcal G = (P, M, p_0)$ with $p_0 = (\ms H, X, \varphi)$. \cr
}
\medskip \activatealgo \obeylines
[if] $\mathcal G \in \{\true, \false\}$ [then] [return] $\mathcal G$
[if] $\varphi$ is an atomic or negated atomic formula [then] [return] $\eval(\mathcal G)$
Let $P' := \{p_0\}$ and $M' := \emptyset$.
[for] $p \in \next(p_0)$ [do]
    \> Let $\mathcal G' = (P_1', M_1', p') := \reduce(\subgame_{\mathcal G}(p))$.
    \> [if] $\varphi$ is universal [and] $\mathcal G' = \false$ [then] [return] $\false$
    \> [if] $\varphi$ is existential [and] $\mathcal G' = \true$ [then] [return] $\true$
    \> [if] $\mathcal G' \notin \{\true, \false\}$ [and] $\mathcal G' \not\cong \mathcal G''$ %
	for all $\mathcal G'' \in \subgames((P', M', p_0))$ [then]
	\>\> Update $P' := P' \cup P_1'$ and $M' := M' \cup M_1' \cup \{(p_0, p')\}$.
[if] $P' = \{p_0\}$ [then] [return] $\eval((P', M', p_0))$.
[return] $(P', M', p_0)$

\caption{\label{alg:reduce} Reducing a game.}
\end{algorithm}

We now define a \emph{reduce operation} that significantly shrinks the
size of a game~$\mathcal G$ (see Algorithm~\ref{alg:reduce}).
Firstly, subgames won by the opponent player are removed.  If, for instance, the formula is universal, then
the falsifier can safely ignore subgames that evaluate as
$\true$, i.e., for which the verifier has a winning strategy. 
For example, it is easy to see that we can remove the two subgames $\true$ and $\false$
in Figure~\ref{fig:eval}.

Secondly, we only need to keep one representation per equivalence class
under~$\cong$ for all undetermined games.  Here, we use the fact that
$\eval(\mathcal G_1) \cong \eval(\mathcal G_2)$ for any $\mathcal G_1, \mathcal
G_2$ with $\mathcal G_1 \cong \mathcal G_2$.  We will not explicitly prove
this claim.  If, however, 
$\mathcal G_1 = \emcg(\ms A_1, X, \varphi)$ and $\mathcal G_2 = \emcg(\ms A_2, X, \varphi)$
for some $\tau$-structures $\ms A_1$ and $\ms A_2$, for $X \subseteq A_1 \cap A_2$ and $\varphi \in \mso(\tau)$,
then the bijection $\pi$ induced by the definition of $\cong$ yields a bisimulation
between $\emcg(\ms A_1, X, \varphi)$ and $\emcg(\ms A_2, X, \varphi)$.
In particular, if both $\mathcal G_1$ and $\mathcal G_2$ are subgames of the same
game~$\ms G$, then it suffices to keep either subgame
as ``witness'' for possible winning positions for the respective player in the
model checking game.  Thus, removing equivalent subgames from a game~$\mathcal G$
can be seen as a variant of taking the \emph{bisimulation quotient} (cf., \cite[Chapter 7]{BK08}) of $\mathcal G$.

See Figures~\ref{fig:vc-unfolding} and~\ref{fig:ds-unfolding} in Section~\ref{sec:concrete-problems} for two
examples.

\begin{lemma}
\label{lem:eval-equals-reduce}
Let $\ms A$ be a $\tau$-structure, $X \subseteq A$, and $\varphi \in \mso(\tau)$.
Let $\mathcal G = \emcg(\ms A, X, \varphi)$.  Then
\begin{itemize}
\item $\eval(\mathcal G) = \true$, if and only if $\reduce(\mathcal G) = \true$, and
\item $\eval(\mathcal G) = \false$, if and only if $\reduce(\mathcal G) = \false$.
\end{itemize}
\end{lemma}
\begin{proof}
Let $\mathcal G = (P, M, p_0)$, where $p_0 = (\ms H, X, \varphi)$.  Without loss of generality, we
assume that $\mathcal G \notin \{\true, \false\}$.
We only show the first case (\true), the second statement is proven analogously.
The proof is an induction over the structure of~$\varphi$.
If $\varphi$ is an atomic or negated atomic formula or $P = \{p_0\}$, then the statement holds
by definition of~$\reduce(\mathcal G)$.
For the induction step, assume $\varphi$ is not an atomic or negated formula, and $\next(p_0) \neq \emptyset$. 

Let $\mathcal G_{p} = \subgame_{\mathcal G}(p)$ for all $p \in \next(p_0)$ and
let $\eval(\mathcal G) = \true$.
If $\varphi$ is existential, then there is $p \in \next(p_0)$ with $\eval(\mathcal G_{p}) = \true$.
By the induction hypothesis, $\reduce(\mathcal G_{p}) = \eval(\mathcal G_{p}) = \true$,
and therefore $\reduce(\mathcal G) = \true$.
Similarly, if $\varphi$ is universal, then $\eval(\mathcal G_{p}) = \true$ for all $p \in \next(p_0)$.
By the induction hypothesis, $\reduce(\mathcal G_{p}) = \true$ for each $p \in \next(p_0)$.
Hence, we have $P' = \{p_0\}$ after the for-loop.  Since $\varphi$ is universal, the call to
$\eval((P', M', p_0))$ returns $\true$ by definition, and therefore $\reduce(\mathcal G) = \true$.

Conversely, let $\reduce(\mathcal G) = \true$.
If $\varphi$ is existential, then there must be some $p \in \next(p_0)$ with $\reduce(\mathcal G_{p}) = \true$.
Assume for a contradiction that $\reduce(\mathcal G_{p}) = \false$ for all
$p \in \next(p_0)$.  Then $P' = \{p_0\}$ after the for-loop, which implies $\eval((P', M', p_0)) = \false$, a contradiction.
Let therefore $p$ be such a position with $\reduce(\mathcal G_{p}) = \true$.
Then, by the induction hypothesis, $\eval(\mathcal G_{p}) = \true$ for this~$p$, and therefore
also $\eval(\mathcal G) = \true$.
If $\varphi$ is universal, then we know $P' = \{p_0\}$ after the for-loop, as this is
the only possibility how~$\reduce(\mathcal G)$ can return~$\true$.
Therefore, $\reduce(\mathcal G_{p}) = \true$ for all $p \in \next(p_0)$, and
hence $\eval(\mathcal G) = \true$ by the induction hypothesis and definition of~$\eval(\mathcal G)$.
\end{proof}

Now we prove an upper bound for the size of a reduced game.  Since this
is a general upper bound for arbitrary formulas and structures, we
cannot expect better bounds than the known lower bounds (unless $\rm
P=NP$)~\cite{FG04}.

\begin{definition}[Equivalent Structures]
Let $\tau$ be a vocabulary and $\varphi \in \mso(\tau)$.
Let $\ms A_1, \ms A_2$ be two $\tau$-structures and $X \subseteq A_1 \cap A_2$.

We call $\ms A_1$ and $\ms A_2$ \emph{equivalent with respect to $\varphi$ and $X$}, denoted
by $\ms A_1 \cong_{X, \varphi} \ms A_2$, if
$\reduce(\emcg(\ms A_1, X, \varphi)) \cong \reduce(\emcg(\ms A_2, X, \varphi))$.

For an arbitrary set $X$ of objects, we let
$$\STR(\tau, X) = \{\,\ms A \in \STR(\tau) \mid \text{$X \subseteq A$}\,\}$$
be the set of all $\tau$-structures that contain $X$, and
$\STR(\tau, X)/{\cong_{X, \varphi}}$ the set of equivalence classes of $\STR(\tau, X)$ under $\cong_{X, \varphi}$.
We let
$$N_{X, \varphi} := |\STR(\tau, X)/{\cong_{X, \varphi}}|.$$
\end{definition}

\begin{lemma}
\label{lem:number}
Let $\tau$ be a vocabulary, $\varphi \in \mso(\tau)$, and $X$ be a set of objects.
Then
$$
N_{X, \varphi} \leq \exp^{\qr(\varphi)+1}((|X|+1)^{O(\|\varphi\|)}),
$$
where
$\|\varphi\|$ is the length of an encoding of~$\varphi$.
\end{lemma}
\begin{proof}
Without loss of generality, we assume $\tau$ is minimal such that $\varphi \in \mso(\tau)$ and
therefore $\|\varphi\| \ge \max\{|\tau|, \arity(\tau)\}$.
We prove the claim by induction over the structure of~$\varphi$.

If $\varphi$ is an atomic or negated atomic formula, let $\bar c = \nullary(\tau)$, and
$\ms A \in \STR(\tau, X)$.  Let $\mathcal G_{\ms A} = \reduce(\emcg(\ms A, X, \varphi))$.
Then either $\mathcal G_{\ms A} \in \{\true, \false\}$,
or $\mathcal G_{\ms A} = (P, M, p_0)$, where $p_0 = (\ms H, X, \varphi)$
and $\ms H = \ms A[X \cup c^{\ms A}]$.  Hence, $N_{X, \varphi}$
depends on the number of non-isomorphic structures on at most $n := |X| +
|\bar c^{\ms A}| \le |X| + |\nullary(\tau)|$ objects.  For a fixed
relation symbol $R \in \tau$, there are $2^{n^{\arity(R)}}$ ways to
choose the interpretation~$R^{\ms H}$.  The total number of non-isomorphic $\tau$-structures
over at most $n$ objects is therefore bounded by $N_{X, \varphi} \leq
\exp^{\qr(\varphi)+1}((|X|+1)^{O(\|\varphi\|)})$.

If $\varphi = \psi_1 \wedge \psi_2$ or $\varphi =
\psi_1 \vee \psi_2$, then $\qr(\varphi) = \max\{\qr(\psi_1),
\qr(\psi_2)\}$ and $\|\psi_1\| + \|\psi_2\| \le \|\varphi\|$.
Furthermore, by the induction hypothesis we get $N_{X, \psi_i} \leq
\exp^{\qr(\varphi)+1}((|X|+1)^{O(\|\psi_i\|)})$.
We conclude that $N_{X, \varphi} = O(N_{X, \psi_1} \cdot N_{X, \psi_2}) \leq
\exp^{\qr(\varphi)+1}((|X|+1)^{O(\|\varphi\|)})$.

If $\varphi \in \{\forall c \psi, \exists c \psi, \forall
R \psi, \exists R \psi\}$, then $\qr(\psi) =
\qr(\varphi) - 1$, $\|\psi\| < \|\varphi\|$, and, by the induction
hypothesis, $N_{X, \psi} = \exp^{\qr(\psi)+1}((|X|+1)^{O(\|\psi\|)})$.
Since $\reduce()$ ignores equivalent subgames, 
the total number $N_{X, \varphi}$ is upper-bounded by $2^{N_{X, \psi}}
\leq \exp^{\qr(\varphi)+1}((|X|+1)^{O(\|\varphi\|)})$.
\end{proof}

\begin{lemma}
\label{lem:size}
Let $\ms A$ be a $\tau$-structure, $X \subseteq A$ and $\varphi \in \mso(\tau)$.
Then
$$
|\reduce(\emcg(\ms A, X, \varphi))|  \le 
\exp^{\qr(\varphi)+1}((|X|+1)^{O(\|\varphi\|)}),$$
where
$\|\varphi\|$ is the length of an encoding of~$\varphi$.
\end{lemma}
\begin{proof}
We use induction over the structure of~$\varphi$.
If $\varphi$ is an atomic or negated atomic formula, then $\mathcal G = \emcg(\ms A, X, \varphi)$
contains only a single position and $\reduce(\mathcal G) \in \{\true, \false, \mathcal G\}$.

If $\varphi = \psi_1 \wedge \psi_2$ or $\varphi = \psi_1 \vee \psi_2$, let, for $i \in \{1,2\}$,
be $\mathcal G_{\psi_i} = \emcg(\ms A, X, \psi_i)$. 
By the induction hypothesis,
$|\reduce(\mathcal G_{\psi_i})| = \exp^{\qr(\varphi)+1}((|X|+1)^{O(\|\psi_i\|)})$
where $\qr(\psi_i) \le \qr(\varphi)$ and $\|\psi_1\| + \|\psi_2\| \le \|\varphi\|$, 
and therefore,
\begin{multline*}
|\reduce(\mathcal G)| \le 1 + |\reduce(\mathcal G_{\psi_1})| +
|\reduce(\mathcal G_{\psi_2})| \\
\le \exp^{\qr(\varphi)+1}((|X|+1)^{O(\|\varphi\|)}).
\end{multline*}

If otherwise  $\varphi \in \{\forall c \psi, \exists c \psi, \forall
R \psi, \exists R \psi\}$, then $\qr(\psi) =
\qr(\varphi) - 1$ and $\|\psi\| < \|\varphi\|$.
Since equivalent subgames are ignored,
\begin{multline*}
|\reduce(\mathcal G)| \le 1 + N_{X, \psi} \cdot \exp^{\qr(\psi)+1}((|X|+1)^{O(\|\psi\|)}) \\
\le \exp^{\qr(\varphi)+1}((|X|+1)^{O(\|\varphi\|)}).
\end{multline*}
\end{proof}

\section{Combining and Extending Games}

In this section, we show how model checking games on structures can be computed inductively.
We will introduce two algorithms:  Algorithm~\ref{alg:combine} will be used when
structures are \emph{combined}, i.e., taking the union of two compatible structures.
This happens at \emph{join} and \emph{introduce} nodes of the tree decomposition.
Algorithm~\ref{alg:forget} will be used when objects are removed from the set~$X$, 
which happens at \emph{forget} nodes of the tree decomposition.
We first will study the case of combining games.  The next lemma
is required for technical reasons.

\begin{algorithm}[tb]
Algorithm $\combine(\mathcal G_1, \mathcal G_2)$
\halign{# \hfil&#\hfil\cr
Input: 	& Two games $\mathcal G_i = (P_i, M_i, p_i)$ with %
	$p_i = (\ms H_i, X_i, \varphi)$, \cr
	& where $\ms H_1$ and $\ms H_2$ are compatible $\tau$-structures, \cr
	& $X_i \subseteq H_i$, and $\varphi \in \mso(\tau)$. \cr
}
\medskip \activatealgo \obeylines
Let $p_0 := (\ms H_1 \cup \ms H_2, X_1 \cup X_2, \varphi)$, $P := \{p_0\}$ and $M := \emptyset$.
[for] each $(p_1', p_2') \in \next(p_1) \times \next(p_2)$ [do]
    \> Let $p_1' = (\ms H_1', X_1, \psi_1)$ and $p_2' = (\ms H_2', X_2, \psi_2)$.
    \> [if] $\psi_1 = \psi_2$ [and] $\ms H_1'$ and $\ms H_2'$ are compatible [then]
	\>\> Let $(P', M', p_0') = \combine(\subgame_{\mathcal G_1}(p_1'), \subgame_{\mathcal G_2}(p_2'))$.
	\>\> Update $P := P \cup P'$ and $M := M \cup M' \cup \{(p_0, p_0')\}$.
[return] $\reduce((P, M, p_0))$ 

\caption{\label{alg:combine} Combining two games.}
\end{algorithm}

\begin{lemma}
\label{lem:matches-compatible}
Let $\ms A_1$ and $\ms A_2$ be compatible $\tau$-structures, $\varphi \in \mso(\tau)$ and
let $X_1 \subseteq A_1$ and $X_2 \subseteq A_2$ with $A_1 \cap A_2 = X_1 \cap X_2$.
Let, for $i \in \{1,2\}$, $\mathcal R_i = \reduce(\emcg(\ms A_i, X_i, \varphi)) \notin 
\{\true,\false\}$ and $\mathcal G_i = (P_i, M_i, p_i) \cong \mathcal R_i$,
where $p_i = (\ms H_i, X_i, \varphi)$.
Then $\ms H_1$ and $\ms H_2$ are compatible.
\end{lemma}
\begin{proof}
Let $\bar c = \nullary(\tau)$.
Since $\mathcal G_i \cong \mathcal R_i$, we have, by Definition~\ref{def:equivalent-games}, 
$\ms H_i \cong \ms A_i[X_i \cup \bar c^{\ms A_i}]$ for an isomorphism
$h_i$ with $h_i(a) = a$ for all $a \in X_i$.

By definition, $\bar c^{\ms A_i} = \{\,c^{\ms A_i} \mid c \in \bar c
\cap \interpreted(\ms A_i)\,\}$, and therefore $c \in \interpreted(\ms
A_i)$ if and only if $c \in \interpreted(\ms H_i)$.

If $c^{\ms H_i} \in H_1 \cap H_2$, then in particular $c^{\ms H_i} \in A_1 \cap A_2 \subseteq X_i$.
Hence, $c^{\ms H_i} = h_i(c^{\ms H_i}) = c^{\ms A_i}$.
Since $\ms A_1$ and $\ms A_2$ are compatible, $c^{\ms A_1} = c^{\ms A_2}$ for all $c^{\ms A_i} \in A_1 \cap A_2$,
and therefore $c^{\ms H_1} = c^{\ms H_2}$ for all $c^{\ms H_i} \in H_1 \cap H_2$.

Accordingly, $\ms H_1$ and $\ms H_2$ are compatible.
\end{proof}

We now prove that for a structure $\ms A$ with $\ms A = \ms A_1
\cup \ms A_2$ the reduced model checking game $\reduce(\emcg(\ms
A, X, \varphi))$ can, up to equivalence, be computed from $\mathcal
R_1 = \reduce(\emcg(\ms A_1, X, \varphi))$ and $\mathcal R_2 =
\reduce(\emcg(\ms A_2, X, \varphi))$.  Here, $\combine(\mathcal
R_1, \mathcal R_2)$ essentially computes the Cartesian product of
plays in the games over $\ms A_1$ and $\ms A_2$, respectively.
This is possible because each set $U \subseteq A$ can be split into
$U \cap A_1$ and $U \cap A_2$, such that
$(\ms A_1, U \cap A_1) \cup (\ms A_2, U \cap A_2) = (\ms A, U)$.
Similarly, each interpretation of a nullary symbol is either $\nil$, or
contained in $A_1 \cap A_2$, in $A_1 \setminus A_2$, or in
$A_2 \setminus A_1$ (cf., Figure~\ref{fig:join}).
These cases can be reconstructed from the respective subgames on $\ms A_1$ and
$\ms A_2$.

\begin{lemma}
\label{lem:combine}
Let $\ms A_1$ and $\ms A_2$ be compatible $\tau$-structures, $\varphi \in \mso(\tau)$ and
let $X_1 \subseteq A_1$ and $X_2 \subseteq A_2$ with $A_1 \cap A_2 = X_1 \cap X_2$.
Let, for $i \in \{1,2\}$, $\mathcal R_i = \reduce(\emcg(\ms A_i, X_i, \varphi)) \notin 
\{\true,\false\}$ and $\mathcal G_i \cong \mathcal R_i$.
Then
$$
\reduce(\emcg(\ms A_1 \cup \ms A_2, X_1 \cup X_2, \varphi)) \cong \combine(\mathcal G_1, \mathcal G_2).
$$
\end{lemma}
\begin{proof}
The proof is an induction over the structure of~$\varphi$.  Let $\ms A = \ms A_1 \cup \ms
A_2$, $X = X_1 \cup X_2$, and $\bar c = \nullary(\tau)$.  Let $\mathcal
R = (P_{\mathcal R}, M_{\mathcal R}, p_{\mathcal R}) =
\reduce(\emcg(\ms A, X, \varphi))$ and $\mathcal G = (P_{\mathcal
G}, M_{\mathcal G}, p_{\mathcal G}) = \combine(\mathcal G_1,
\mathcal G_2)$.
Let, for $i \in \{1,2\}$, $\mathcal G_i = (P_{\mathcal G_i}, M_{\mathcal G_i}, p_{\mathcal G_i})$ and
$p_{\mathcal G_i} = (\ms H_i, X_i, \varphi)$.

By Lemma~\ref{lem:matches-compatible}, $\ms H_1$ and $\ms
H_2$ are compatible. 
Furthermore, $\ms A_i[X_i \cup \bar c^{\ms A_i}] \cong \ms H_i$, and
thus $\ms A[X \cup \bar c^{\ms A}] = \ms A_1[X_1
\cup \bar c^{\ms A_1}] \cup \ms A_2[X_2 \cup \bar c^{\ms A_2}] \cong
\ms H_1 \cup \ms H_2$.

If $\mathcal R \notin \{\true, \false\}$, then $p_{\mathcal R} =
(\ms A[X \cup \bar c^{\ms A}], X, \varphi)$.  Therefore, $$ p_{\mathcal
R} = (\ms A[X \cup \bar c^{\ms A}], X, \varphi) \cong (\ms H_1 \cup
\ms H_2, X_1 \cup X_2, \varphi) = p_{\mathcal G}.  $$

Let $\varphi$ be an atomic or negated atomic formula.  If $\mathcal
R \notin \{\true, \false\}$ the lemma already holds with above
considerations.  Therefore consider the case $\mathcal R \in \{\true,
\false\}$, say~$\mathcal R = \true$.  Then $\eval(\emcg(\ms A, X, \varphi)) =
\mathcal R = \true$ by Lemma~\ref{lem:eval-equals-reduce}.
Therefore, $\mathcal R = \true$ if and only if
the verifier wins the play $(p_0)$, where
$p_0$ is the initial position of $\emcg(\ms A[X \cup \bar c^{\ms A}], X, \varphi)$.
The claim then follows, since
$p_0 = (\ms A[X \cup \bar c^{\ms A}], X, \varphi) \cong p_{\mathcal G}$,
where in particular
$\ms A[X \cup \bar c^{\ms A}] \cong \ms H_1 \cup \ms H_2$ and $X = X_1 \cup X_2$.

For the induction step, we distinguish the following cases.

\paragraph{Case $\varphi = \psi_1 \wedge \psi_2$ or $\varphi
= \psi_1 \vee \psi_2$}~

Let, for $\psi \in \{\psi_1, \psi_2\}$, $\mathcal
R_{\psi} = \reduce(\emcg(\ms A, X, \psi))$ and,
for each $i \in \{1,2\}$, be
$\mathcal R_{i,\psi} = \reduce(\emcg(\ms A_i, X_i, \psi))$.

Consider $\psi \in \{\psi_1, \psi_2\}$ with $\mathcal R_{\psi}
\notin \{\true, \false\}$ and suppose there was $i \in \{1,2\}$,
say $i = 1$, with $\mathcal R_{1,\psi} \in \{\true, \false\}$.
Let
$\mathcal U_{1,\psi} = \emcg(\ms A_1, X_1, \psi)$ and
$\mathcal U_{\psi} = \emcg(\ms A, X_1 \cup A_2, \psi)$.  By Lemma~\ref{lem:eval-equals-reduce}, $\eval(\mathcal
U_{1,\psi}) \in \{\true, \false\}$, and therefore by
Corollary~\ref{cor:introduce-join}, $\eval(\mathcal U_{\psi}) \in
\{\true, \false\}$.  Since $X_1 \cup X_2 \subseteq X_1 \cup A_2$,
also $\eval(\ms A, X_1 \cup X_2, \psi)
\in \{\true, \false\}$.  This contradicts $\mathcal R_\psi \notin
\{\true, \false\}$ via Lemma~\ref{lem:eval-equals-reduce}.

Therefore, we have $\mathcal R_{i,\psi} \notin \{\true, \false\}$
for each $i \in \{1,2\}$, which implies $\mathcal R_i \notin \{\true, \false\}$.
Since $\mathcal G_i \cong \mathcal R_i$ for $i \in \{1,2\}$, there
is $\mathcal G_{i,\psi} \in \subgames(\mathcal G_i)$ with $\mathcal G_{i,\psi} \cong \mathcal R_{i,\psi}$.
The algorithm $\combine(\mathcal G_1, \mathcal
G_2)$ will eventually call $\combine(\mathcal G_{1,\psi}, \mathcal G_{2,\psi})$.
Then, by the induction hypothesis, 
$\subgames(\combine(\mathcal G_1, \mathcal G_2))$ contains the required subgame
$\combine(\mathcal G_{1,\psi}, \mathcal G_{2,\psi}) \cong \mathcal R_{\psi}$.

Conversely, let $\psi \in \{\psi_1, \psi_2\}$ and
$(\mathcal G_{1,\psi}, \mathcal G_{2,\psi}) \in \subgames(\mathcal G_1) \times \subgames(\mathcal G_2)$
such that $\combine(\mathcal G_1, \mathcal G_2)$ recursively calls
$\combine(\mathcal G_{1,\psi}, \mathcal G_{2,\psi})$.
From $\mathcal G_i \cong \mathcal R_i$ we get $\mathcal G_{i,\psi} \cong \mathcal R_{i,\psi}$.
Then $\combine(\mathcal G_{1,\psi}, \mathcal G_{2,\psi}) \cong \mathcal R_\psi$
 by the induction hypothesis.

Together, the statement of the lemma follows.

\paragraph{Case $\varphi = \forall R \psi$ or $\varphi =
\exists R \psi$}~

Consider an arbitrary $U \subseteq A$ and let $\mathcal R' = \reduce(\emcg(\ms A', X, \psi))$, where
$\ms A' = (\ms A, U)$ with $R^{\ms A'} = U$.
For $i \in \{1,2\}$, let $U_i = U \cap A_i$ and $\mathcal R_i' = \reduce(\emcg(\ms A_i', X_i, \psi))$, where
$\ms A_i' = (\ms A_i, U_i)$.
If $\mathcal R' \notin \{\true, \false\}$, then $\mathcal R_i'
\notin \{\true, \false\}$ for each $i \in \{1,2\}$ by using a combination
of Lemma~\ref{lem:eval-equals-reduce} and
Corollary~\ref{cor:introduce-join}. 
Therefore, $\mathcal R_i \notin
\{\true, \false\}$.  Since $\mathcal G_i \cong \mathcal R_i$,
there is $\mathcal G_i' = (P_i', M_i', p_i') \in \subgames(\mathcal G_i)$ with $\mathcal
G_i' \cong \mathcal R_i'$.
Let $p_i' = (\ms H_i', X_i, \psi)$.
Since $\ms A_1'$ and $\ms A_2'$ are compatible
and $\mathcal G_i' \cong \mathcal R_i'$, we by
Lemma~\ref{lem:matches-compatible} have that $\ms H_1'$ and $\ms
H_2'$ are compatible. 
Therefore, the algorithm eventually recursively calls $\combine(\mathcal G_1', \mathcal G_2')$.
By the induction hypothesis, $\combine(\mathcal G_1', \mathcal G_2') \cong \mathcal R'$.

Conversely, assume the algorithm recursively calls
$\combine(\mathcal G_1', \mathcal G_2')$,
where $\mathcal G_i' = (P_i', M_i', p_i') \in \subgames(\mathcal G_i)$ for
each $i \in \{1,2\}$.
From $\mathcal G_i \cong \mathcal R_i$ we get $\mathcal G_i' \cong 
\reduce(\emcg(\ms A_i', X_i, \psi))$, where $\ms A_i' = (\ms A_i, U_i)$ for
some $U_i \subseteq A_i$.
Let $p_i' = (\ms H_i', X_i, \psi)$.
Since $\ms H_1$ and $\ms H_2$ are compatible
and $A_1 \cap A_2 \subseteq X_i \subseteq H_i$, also
$\ms A_1'$ and $\ms A_2'$ are compatible.
Therefore, the induction hypothesis implies
$\combine(\mathcal G_1', \mathcal G_2') \cong \reduce(\emcg(\ms A', X, \psi))$,
where $\ms A' = (\ms A, U_1 \cup U_2)$ with $R^{\ms A'} = U_1 \cup U_2$.

Together, the statement of the lemma follows.

\paragraph{Case $\varphi = \forall c \psi$ or $\varphi = \exists c \psi$}~

Consider a $(\tau, c)$-expansion $\ms A'$ of $\ms A$
and let $\mathcal R' = \reduce(\emcg(\ms A', X, \psi))$.
Let, for $i \in \{1,2\}$, $\ms A_i' = \ms A'[A_i]$ be the $(\tau, c)$-expansion of
$\ms A_i$ with $c^{\ms A_i'} = c^{\ms A'}$ if $c^{\ms A'} \in A_i$, and
$c^{\ms A_i'} = \nil$ otherwise.
Let $\mathcal R_i' = \reduce(\emcg(\ms A_i', X_i, \psi))$.
If $\mathcal R' \notin \{\true, \false\}$, then
$\mathcal R_i' \notin \{\true, \false\}$ by a combination
of Lemma~\ref{lem:eval-equals-reduce} and
Corollary~\ref{cor:introduce-join}. 
Therefore, $\mathcal R_i \notin \{\true, \false\}$. 
Since $\mathcal G_i \cong \mathcal R_i$,
there is $\mathcal G_i' = (P_i', M_i', p_i') \in \subgames(\mathcal G_i)$ with
$\mathcal G_i' \cong \mathcal R_i'$.
Let $p_i' = (\ms H_i', X_i, \psi)$.
Since $\ms A_1'$ and $\ms A_2'$ are compatible and $\mathcal
G_i' \cong \mathcal R_i'$, Lemma~\ref{lem:matches-compatible} implies
that $\ms H_1'$ and $\ms H_2'$ are compatible.
The algorithm therefore eventually calls
$\combine(\mathcal G_1', \mathcal G_2')$.
By the induction hypothesis, $\combine(\mathcal G_1', \mathcal G_2') \cong \mathcal R'$.

Conversely, assume the algorithm recursively calls
$\combine(\mathcal G_1', \mathcal G_2')$,
where $\mathcal G_i' = (P_i', M_i', p_i') \in \subgames(\mathcal G_i)$ for
each $i \in \{1,2\}$.
From $\mathcal G_i \cong \mathcal R_i$ we get $\mathcal G_i' \cong \reduce(\ms A_i', X_i, \psi)$
for some $(\tau, c)$-expansion~$\ms A_i'$ of~$\ms A_i$.
Since $\ms H_1$ and $\ms H_2$ are compatible
and $A_1 \cap A_2 \subseteq X_i \subseteq H_i$, also
$\ms A_1'$ and $\ms A_2'$ are compatible.
By the induction hypothesis,
$\combine(\mathcal G_1', \mathcal G_2') \cong \reduce(\emcg(\ms A', X, \psi))$,
where $\ms A' = \ms A_1' \cup \ms A_2'$.

Together, the statement of the lemma follows.
\end{proof}

\begin{algorithm}[tb]
Algorithm $\forget(\mathcal G, x)$
\halign{# \hfil&#\hfil\cr
Input: 	& A game $\mathcal G = (P, M, p_0)$ with $p_0 = (\ms H, X, \varphi)$ and  $x \in X$ \cr
}
\medskip \activatealgo \obeylines
[if] there is $c \in \interpreted(\ms H)$ with $c^{\ms H} = x$ [then]
    \> let $p_0' = (\ms H, X \setminus \{x\}, \varphi)$
[else] let $p_0' = (\ms H[H \setminus \{x\}], X \setminus \{x\}, \varphi)$.
Let $P' = \{p_0'\}$ and $M' = \emptyset$.
[for] each $\mathcal G' \in \subgames(\mathcal G)$ [do]
    \> Let $(P'', M'', p_0'') = \forget(\mathcal G')$.
    \> Set $P' := P' \cup P''$ and $M' := M' \cup M''$.
[return] $\reduce((P', M', p_0'))$
\caption{\label{alg:forget} Forgetting an object.}
\end{algorithm}

\begin{lemma}
\label{lem:forget}
Let $\ms A$ be a $\tau$-structure, $X \subseteq A$ and $x \in X$.
Let $\varphi \in \mso(\tau)$ and
$\mathcal G \cong \reduce(\emcg(\ms A, X, \varphi)) \notin \{\true,\false\}$.
Then
$$
\reduce(\emcg(\ms A, X \setminus \{x\}, \varphi)) \cong \forget(\mathcal G, x).
$$
\end{lemma}
\begin{proof}
We use induction over the structure of~$\varphi$.
Let $\bar c = \nullary(\tau)$, $X' = X \setminus \{x\}$, $\mathcal
R' = (P_{\mathcal R'}, M_{\mathcal R'}, p_{\mathcal R'}) =
\reduce(\ms A, X \setminus \{x\}, \varphi)$.
Let $\mathcal G = (P_{\mathcal G}, F_{\mathcal G}, p_{\mathcal G})$ with $p_{\mathcal G} = (\ms H, X, \varphi)$.
Here, $\ms H \cong \ms A[X \cup c^{\ms A}]$, since $\mathcal G \cong 
\reduce(\emcg(\ms A, X, \varphi))$.

If $\varphi$ is an atomic or negated atomic formula and $\mathcal R' \notin \{\true, \false\}$,
the statement holds since $p_{\mathcal R'} = (\ms A[X' \cup c^{\ms A}], X', \varphi)$
by definition.

If otherwise $\varphi$ is an atomic or negated atomic formula and
$\mathcal R' \in \{\true, \false\}$, 
let $\ms H' = \ms H[H \setminus \{x\}]$ if $c^{\ms H} \neq x$ for all $c \in \interpreted(\ms H)$,
and $\ms H' = \ms H$ otherwise.
If $\mathcal R' \in \{\true, \false\}$, then
$\eval(\emcg(\ms A, X', \varphi)) = \mathcal R'$ by Lemma~\ref{lem:eval-equals-reduce}.
Since $\ms H' \cong \ms A[X' \cup c^{\ms A}]$, we
have $\eval(\emcg(\ms A, X', \varphi)) = \eval(\emcg(\ms H', X', \varphi'))$
and $$
	\eval(\emcg(\ms H', X', \varphi')) = \reduce(\emcg(\ms H', X', \varphi)) = \forget(\mathcal G, x).
$$

For the induction step, let $\mathcal G' \in \subgames(\mathcal G)$ be an arbitrary subgame of~$\mathcal G$.
Since $\mathcal G \cong \reduce(\emcg(\ms A, X, \varphi))$, we know that
$\mathcal G' \cong \reduce(\emcg(\ms A', X, \psi))$ for some expansion $\ms A'$ of $\ms A$
and subformula $\psi$ of $\varphi$.  By the induction hypothesis,
$$
\forget(\mathcal G', x) \cong \reduce(\emcg(\ms A', X \setminus \{x\}, \psi)).
$$

Conversely, if $\mathcal R'' = \reduce(\emcg(\ms A', X \setminus \{x\}, \psi))$ is a subgame of~$\mathcal R'$, 
then
$\mathcal R'' \notin \{\true, \false\}$.  This implies
$\reduce(\emcg(\ms A', X, \psi)) \notin \{\true, \false\}$ by
Lemmas~\ref{lem:emcg-eval-forget} and~\ref{lem:eval-equals-reduce}.
Therefore, there is $\mathcal G' \in \subgames(\mathcal G)$ with $\mathcal G' \cong \reduce(\emcg(\ms A', X, \psi))$.
By the induction hypothesis, $\mathcal R'' \cong \forget(\mathcal G', x)$.

Together, the statement of the lemma follows.
\end{proof}

Finally, we come back to Algorithm~\ref{alg:convert} and show that
its correctness translates to reduced games.

\begin{lemma}
\label{lem:convert}
Let $\ms A$ be a fully interpreted $\tau$-structure, $X \subseteq A$, and $\varphi \in \mso(\tau)$.
Let $\mathcal G \cong \reduce(\emcg(\ms A, X, \varphi))$.
Then
$$
\eval(\mcg(\ms A, \varphi)) = \eval(\convert(\mathcal G)).
$$
\end{lemma}
\begin{proof}
We prove the statement by induction over the structure of~$\varphi$.
Recall that $\mathcal M = \mcg(\ms A, X, \varphi)$ is determined and
hence $\eval(\mathcal M) \in \{\true, \false\}$.

If $\mathcal G \in \{\true, \false\}$, then $\mathcal G =
\convert(\mathcal G)$.  We get $\mathcal G = \eval(\emcg(\ms A, X,
\varphi))$ from Lemma~\ref{lem:eval-equals-reduce} and therefore, using
Lemmma~\ref{lem:emcg-determines-mcg} for the first equality,
$$
\eval(\mcg(\ms A, \varphi)) = \eval(\emcg(\ms A, X, \varphi)) = \mathcal G = \eval(\mathcal G) = \eval(\convert(\mathcal G)).
$$

Let therefore $\mathcal G = (P, M, p_0) \notin \{\true, \false\}$ with
$p_0 = (\ms H, X, \varphi)$ and suppose $\eval(\mcg(\ms A, \varphi)) =
\true$ (the case $\false$ is shown analogously).  For atomic or negated
atomic formulas, the statement holds since, by definition, $\mathcal G
\cong \reduce(\emcg(\ms A, X, \varphi)) = \emcg(\ms A, X, \varphi)$,
and hence $\mcg(\ms A, \varphi) = \convert(\mathcal G)$
by Lemma~\ref{lem:convert-me}.

If $\varphi \in \{\forall R \psi, \exists R \psi\}$, say $\varphi = \forall R \psi$,
consider $U \subseteq A$ and let $\ms A' = (\ms A, U)$
with $R^{\ms A'} = U$.
If there is $\mathcal G' \in \subgames(\mathcal G)$ with
$\mathcal G' \cong \reduce(\emcg(\ms A', X, \psi))$,
then $\eval(\mcg(\ms A', \varphi)) = \eval(\convert(\mathcal G'))$ by the induction hypothesis.
If otherwise there is no such $\mathcal G'$ in $\subgames(\mathcal G)$,
then $\reduce(\emcg(\ms A', X, \psi)) = \true$ by definition of~$\reduce()$, since
$\mathcal G \notin \{\true, \false\}$.
By Lemmas~\ref{lem:eval-equals-reduce} and~\ref{lem:convert-me}, we then conclude
$\eval(\mcg(\ms A', \psi)) = \true$.
Together, the lemma follows.

Similarly, if $\varphi \in \{\psi_1 \wedge \psi_2, \psi_1 \vee
\psi_2\}$, then for $\psi \in \{\psi_1, \psi_2\}$ either there is
$\mathcal G' \in \subgames(\mathcal G)$, such that $\mathcal G' \cong
\reduce(\ms A, X, \psi)$, or there is no such $\mathcal G'$ contained
in~$\subgames(\mathcal G)$.  In the former case we again obtain
$\eval(\mcg, \psi) = \eval(\convert(\mathcal G'))$ by the induction
hypothesis, and in the latter case we can again argue that $\mathcal G'
\cong \reduce(\emcg(\ms A, X, \psi)) \in \{\true, \false\}$.

Finally, let $\varphi \in \{\forall c \psi, \exists c \psi\}$.
For any $a \in A$ and $\ms A' = (\ms A, a)$, where $c^{\ms A'} = a$,
we argue analogously to the previous cases that either there is $\mathcal G' \in
\subgames(\mathcal G)$, such that $\mathcal G' \cong \reduce(\ms A', X,
\psi)$, or there is no such $\mathcal G'$ contained in~$\subgames(\mathcal
G)$, which implies $\mathcal G' \cong \reduce(\emcg(\ms A', X, \psi))
\in \{\true, \false\}$.

Hence, consider the $(\tau, c)$-expansion $\ms A'$ of~$\ms A$ with $c^{\ms A'} = \nil$.
If there is $\mathcal G' = (P', M', p_0') \in \subgames(\mathcal G)$ with 
$\mathcal G' \cong \reduce(\emcg(\ms A', X, \psi))$, then $\mathcal G' \notin \{\true, \false\}$.
In particular, $\mathcal G' = (\ms H', X, \psi)$, where $\ms H'$ is not fully interpreted.
Therefore, $\convert(\mathcal G)$ removes the subgame~$\mathcal G'$ from~$\mathcal G$.
In either case, $\convert(\mathcal G)$ does only contain subgames where~$c$ has been
interpreted as an object in~$A$, as considered above.
Together, the statement of the lemma then follows.
\end{proof}

\section{Courcelle's Theorem}
\label{sec:algorithm}

We can now reprove Courcelle's Theorem for \linmso-definable
optimization problems. 
Throughout this section, we shall abbreviate
$\redgame(\ms A, X, \varphi) := \reduce(\emcg(\ms A, X, \varphi))$.

\begin{theorem}
\label{thm:main}
Fix a relational vocabulary $\tau$, a set
$\bar R = \{R_1, \ldots, R_l\} \subseteq \unaryrelations(\tau)$ of unary
relation symbols, and $\tau' = \tau \setminus \bar R$.
Let $\varphi \in \mso(\tau)$, and $w, \alpha_1, \ldots, \alpha_l \in \mathbf Z$ be
constants.
Given a $\tau'$-structure $\ms A$ together with a tree decomposition
$(\mathcal T, \mathcal X)$ of $\ms A$ having width at most~$w$, where
$\mathcal T = (T, F)$ and $\mathcal X = (X_i)_{i \in T}$, one
can compute 
$$\linmsoopt$$
in time $O(|T|)$.
\end{theorem}

The remainder of this section is devoted to the proof of this theorem.
We give an algorithm that essentially works as follows:
In a first phase, the algorithm uses dynamic programming on the
tree decomposition (based on Lemmas~\ref{lem:alg-leaf}--\ref{lem:alg-join})
to compute the reduced extended model checking games
$\mathcal G \cong \reduce(\ms A_i', \emptyset, \varphi)$ and the
values $\sum_{k=1}^l \alpha_k |U_k|$ for \emph{all} structures $\ms A_i' = (\ms A, U_1, \ldots, U_l)$
where $U_i \subseteq A$ for $1 \le i \le l$.
Note that by the previous sections the algorithm does not need to
distinguish between equivalent games.
In a second phase, the algorithm tests whether the verifier has
a winning strategy on $\convert(G)$, or, in other words (Lemma~\ref{lem:convert}),
whether $(\ms A, U_1, \ldots, U_l) \models \varphi$.
The algorithm then collects the values
$\sum_{k=1}^l \alpha_k |U_k|$ for all $\ms A_i' = (\ms A,
U_1, \ldots, U_l)$ with $\ms A_i' \models \varphi$ and outputs the optimal one.
Since most of the games considered are equivalent (Lemma~\ref{lem:number}),
we can obtain the desired run time bounds.

Without loss of generality, we assume $X_{\root(\mathcal T)} = \emptyset$.
Recall that for each $i \in T$, $\ms A_i$ is the substructure of
$\ms A$ induced by those objects that appear at or below $i$ in
the tree decomposition.
Let, for $i \in T$, 
$$
\mathcal {AR}_i = \Pow(A_i)
	\times \cdots \times \Pow(A_i) = \Pow(A_i)^l
$$
be the set of possible interpretations of the free relation symbols
$(R_1, \ldots, R_l)$ in $\ms A_i$,
$$
\mathcal{EXP}_i = \{\,(\ms A_i, U_1, \ldots, U_l) \mid (U_1, \ldots, U_l) \in
	\mathcal{AR}_i\,\}
$$
be the set of their corresponding $\tau$-expansions of $\ms A_i$, 
where for each $1 \le j \le l$ the symbol $R_j$ is interpreted as $U_j$,
and
$$
\mathcal {RED}_i = \{\, \redgame(\ms A_i', X_i, \varphi) \mid 
\ms A_i' \in \mathcal{EXP}_i\,\}
$$
be the corresponding extended model checking games in their reduced form.  We let
$(U_1, \ldots, U_l) \cap X_i := (U_1 \cap X_i, \ldots, U_l \cap X_i)$ and
$$
\mathcal{AR}_i \cap X_i = \{\,(U_1, \ldots, U_l) \cap X_i \mid
	(U_1, \ldots, U_l) \in \mathcal{AR}_i\,\}
$$
be the restriction of $\mathcal{AR}_i$ to $X_i$, and let,
for $\bar U = (U_1, \ldots, U_l) \in \mathcal{AR}_i \cap X_i$, 
$$
\mathcal{EXP}_i(\bar U) = \{\,
\ms A_i' \in \mathcal{EXP}_i \mid
R_j^{\ms A_i'} \cap X_i = U_j \text{ for $1 \le j \le l$}
\,\}
$$
be the set of $\tau$-expansions of~$\ms A$ that ``match'' $\bar U$ on $X_i$.
Let
$$
\mathcal{RED}_i(\bar U) =
\{\,\redgame(\ms A_i', X_i, \varphi)
\mid \ms A_i' \in \mathcal{EXP}_i(\bar U)
\,\}
$$
be the corresponding games, and, for arbitrary games $\mathcal R$,
$$
\mathcal{EXP}_i(\bar U, \mathcal R) = \{\,\ms A_i' \in \mathcal{EXP}_i(\bar U) \mid
\mathcal R \cong \redgame(\ms A_i', X_i, \varphi)\,\}
$$
and
$$
\mathcal{RED}_i(\bar U, \mathcal R) = 
\{\,\mathcal R' \in \mathcal{RED}_i(\bar U) \mid \mathcal R \cong \mathcal R'\,\}.
$$
Finally, we let, for $\bar U = (U_1, \ldots, U_l) \in \mathcal{AR}_i$,
$$
\ms A_i(\bar U) = (\ms A_i, U_1, \ldots, U_l)[X_i],
$$
where $R_i^{\ms A_i(\bar U)} = U_i \cap X_i$ for each $1 \le i \le l$, and
$$
\mathcal R(\bar U) = \redgame(\ms A_i(\bar U), X_i, \varphi).
$$

\subsection{The Algorithm}

We use dynamic programming on the tree decomposition as follows.
As usual, we associate with each node $i \in T$ of the tree decomposition a \emph{table} $S_i$ that
contains feasible, \emph{partial} solutions and their corresponding \emph{value} $\val_i$ under the optimization function.

Formally, we let $S_i\colon \mathcal{AR}_i \cap X_i \to \Pow(\mathcal{RED}_i \setminus \{\false\})$
map tuples $\bar U \in \mathcal{AR}_i \cap X_i$ to sets of \emph{feasible} games over $\ms A_i$, i.e.,
games $\mathcal R$ with $\mathcal R \neq \false$, and let $\val_i\colon
\mathcal{RED}_i \to \mathbf Z_\infty$ be the corresponding values,
where $\mathbf Z_\infty =  \mathbf Z \cup \{\infty\}$.

Initially, we let $S_i(\bar U) := \emptyset$ for all $\bar U \in \mathcal{AR}_i \cap X_i$ and
$\val_i(\mathcal R) := \infty$ for all $\mathcal R \in \mathcal{RED}_i$.

\paragraph{Phase 1}
The algorithm traverses the tree decomposition bottom-up. 
Recall that each node $i \in T$ is either a leaf, or of one of the three types
\emph{introduce}, \emph{forget}, or \emph{join}.  The algorithm
distinguishes these four cases as follows.

\begin{description}
\item[leaf] Let $X_i = \{x\}$.  For all $\bar U = (U_1, \ldots, U_l)
\in \mathcal{AR}_i \cap X_i$ the algorithm considers 
$\mathcal R(\bar U) = \redgame(\ms A_i(\bar U), X_i, \varphi)$.
If $\mathcal R(\bar U) \neq \false$, then the algorithm
sets
$$
S_i(\bar U) := \{\mathcal R(\bar U)\} \qquad \text{and} \qquad
\val_i(\mathcal R(\bar U)) := 0.
$$

\item[introduce] Let $j$ be the unique child of $i$ and
$X_i = X_j \cup \{x\}$ for $x \notin A_j$.

For each $\bar U_j = (U_{j,1}, \ldots, U_{j,l}) \in \mathcal{AR}_j \cap X_j$, and
each $\bar U_i = (U_{i,1}, \ldots, U_{i,l}) \in \mathcal{AR}_i \cap X_i$ such that
$(U_{j,1}, \ldots, U_{j,l}) = (U_{i,1} \cap X_j, \ldots, U_{i,l} \cap X_j)$,
the algorithm considers each $\mathcal R_j \in S_j(\bar U_j)$.

Let
$$
\mathcal R_i = \begin{cases}
\true & \text{if } \mathcal R_j = \true \text{ and}\\
\combine(\mathcal R_j, \mathcal R(\bar U_i)) & \text{otherwise.} \\
\end{cases}
$$
If there is $\mathcal R_i' \in S_i(\bar U_i)$ with $\mathcal R_i' \cong \mathcal R_i$,
then let $\mathcal R_i := \mathcal R_i'$ instead.

If $\mathcal R_i \neq \false$, the algorithm sets
$$
S_i(\bar U) := S_i(\bar U) \cup \{\mathcal R_i\}
\qquad \text{and} \qquad
\val_i(\mathcal R_i) := \min\{\val_i(\mathcal R_i), \val_j(\mathcal R_j)\}.
$$

\item[forget] Let $j$ be the unique child of $i$ and
$X_i \cup \{x\} = X_j$ for $x \notin A_i$.

For each $\bar U_j = (U_{j,1}, \ldots, U_{j,l}) \in \mathcal{AR}_j \cap X_j$
the algorithm considers each $\mathcal R_j \in S_j(\bar U_j)$.
Let $\bar U_i = (U_{i,1}, \ldots, U_{i,l}) =
(U_{j,1} \cap X_i, \ldots, U_{j,l} \cap X_i)$ and
$$
\mathcal R_i = \begin{cases}
\true & \text{if } \mathcal R_j = \true \text{ and}\\
\forget(\mathcal R_j, x) & \text{otherwise.} \\
\end{cases}
$$
If there is $\mathcal R_i' \in S_i(\bar U_i)$ with $\mathcal R_i' \cong \mathcal R_i$,
then let $\mathcal R_i := \mathcal R_i'$ instead.
If now $\mathcal R_i \neq \false$,
the algorithm sets $S_i(\bar U_i) := S_i(\bar U_i) \cup \{\mathcal R_i\}$
and
$$
\val_i(\mathcal R_i) := \min\left\{\val_i(\mathcal R_i),
\val_j(\mathcal R_j) + \sum_{k=1}^l \alpha_k (x \in U_{i, k})\right\}.
$$
where $(x \in U_{j,k}) \in \{0,1\}$ as defined in Section~\ref{sec:prel}.

\item[join] Let $j_1, j_2$ be the children of $i$.  Then $X_i = X_{j_1} = X_{j_2}$.

For each $\bar U = (U_1, \ldots, U_l) \in \mathcal{AR}_i \cap X_i$
the algorithm considers each pair
$(\mathcal R_{j_1}, \mathcal R_{j_2}) \in S_{j_1}(\bar U) \times S_{j_2}(\bar U)$.
Let 
$$
\mathcal R_i = \begin{cases}
\true & \text{if } \mathcal R_{j_1} = \true \text{ or } \mathcal R_{j_2} = \true \text{ and} \\
\combine(\mathcal R_{j_1}, \mathcal R_{j_2}) & \text{otherwise.} \\
\end{cases}
$$
If there is $\mathcal R_i' \in S_i(\bar U_i)$ with $\mathcal R_i' \cong \mathcal R_i$,
then let $\mathcal R_i := \mathcal R_i'$ instead.
If now $\mathcal R_i \neq \false$,
the algorithm sets $S_i(\bar U_i) := S_i(\bar U_i) \cup \{\mathcal R_i\}$
and
$$
\val_i(\mathcal R_i) := \min\left\{\val_i(\mathcal R_i), \val_{j_1}(\mathcal R_{j_1})
	+ \val_{j_2}(\mathcal R_{j_2})\right\}.
$$
\end{description}

\paragraph{Phase~2}
Let $r = \root(\mathcal T)$ and
$$\bar U_r = (\emptyset, \ldots, \emptyset) \in \mathcal{AR}_r \cap X_r = \mathcal{AR}_r \cap \emptyset.$$
The algorithm starts with $\OPT := \infty$ and
considers each $\mathcal R_r \in S_r(\bar U_r)$.
If $\eval(\convert(\mathcal R_r)) = \true$, then the algorithm updates
$$\OPT := \min\{\OPT, \val_r(\mathcal R_r)\}.$$
Finally, the algorithm outputs $\OPT$.

\subsection{Proofs}

In order to show that the algorithm is correct and computes the optimal solution, we use induction
over the structure of the tree decomposition to show the following
invariant.

\begin{invariant}
\label{inv:alg}
After the algorithm has processed a node $i \in T$ in Phase~1,
for each $\bar U = (U_1, \ldots, U_l) \in \mathcal{AR}_i \cap X_i$
we have that
\begin{itemize}[(III)]
\item[(I)] for each $\ms A_i' \in \mathcal{EXP}_i(\bar U)$
	with $\mathcal R = \redgame(\ms A_i', X_i,
	\varphi) \neq \false$ there is exactly one $\mathcal R' \in S_i(\bar U)$
	with $\mathcal R' \cong \mathcal R$,
\item[(II)] for each game $\mathcal R \in S_i(\bar U)$
	we have $\mathcal R \neq \false$ and
	$\mathcal{RED}_i(\bar U, \mathcal R) \neq \emptyset$, and
\item[(III)] for each $\mathcal R \in S_i(\bar U)$ we have
	\begin{align*}
		\val_i(\mathcal R) = \min\Biggl\{
			\,\sum_{k=1}^{l} \alpha_k |R_k^{\ms A_i'} \setminus X_i|
 	\Biggm| {} & \ms A_i' \in \mathcal{EXP}_i(\bar U, \mathcal R) \wedge {} \\
		& \reduce(\ms A_i', X_i, \varphi) \neq \false
	\,\Biggr\}.
	\end{align*}
\end{itemize}
\end{invariant}

Here, (I) guarantees that $S_i$ is \emph{complete}, i.e., $S_i(\bar U)$
contains games for all feasible partial solutions, (II) guarantees that
all games in $S_i(\bar U)$ do, in fact, correspond to a reduced game over some
$\tau$-expansion of $\ms A_i$, and (III) guarantees that we also compute the correct
solution, i.e., $\val_i(\mathcal R)$ is optimal for $\mathcal{RED}_i(\bar
U, \mathcal R)$.  Note that the ``exactly one'' in (I) is required for
the claimed running time, but not for the correctness of the solution.

\begin{lemma}
\label{lem:alg-leaf}
Invariant~\ref{inv:alg} holds for leafs of the tree decomposition.
\end{lemma}
\begin{proof}
Let $i \in T$ be a leaf and $\bar U = (U_1, \ldots, U_l) \in
\mathcal{AR}_i \cap X_i = \mathcal{AR}_i$.  Since~$i$ is a leaf, we have
$$
\mathcal{RED}_i(\bar U) =
\{\,\ms A_i(\bar U) \mid \bar U \in \mathcal{AR}_i \cap X_i\,\},
$$
such that (I) and (II) clearly hold.  Furthermore, 
$\mathcal R_j^{\ms A_i(\bar U)} \setminus X_i = \emptyset$ for all $1 \le j \le l$,
since $A_i \setminus X_i = \emptyset$,
and therefore $\val_i(\mathcal R) = 0$ for all $\mathcal R \in \mathcal{RED}_i$.
\end{proof}

\begin{lemma}
\label{lem:alg-introduce}
Let $i \in T$ be an introduce node of the tree decomposition and $j \in T$
be the unique child of $i$.
If Invariant~\ref{inv:alg} holds for~$j$ before the algorithm processes~$i$,
then it also holds for~$i$.
\end{lemma}

\begin{proof}
Let $X_i = X_j \cup \{x\}$, where $x \notin A_j$.
Let $\bar U_i = (U_{i,1}, \ldots, U_{i,l}) \in \mathcal{AR}_i \cap X_i$ and
$\bar U_j = (U_{j,1}, \ldots, U_{j,l}) \in \mathcal{AR}_j \cap X_j$ with
$(U_{j,1}, \ldots, U_{j,l}) = (U_{i,1} \cap X_j, \ldots, U_{i,l} \cap X_j)$.

Consider $\ms A_i' \in \mathcal{EXP}_i(\bar U_i)$ and let $\ms A_j' = \ms A_i'[A_j]$.
If $\mathcal R_i = \redgame(\ms A_i', X_i, \varphi) \neq \false$, then
also $\mathcal R_j = \redgame(\ms A_j', X_j, \varphi) \neq \false$ by
Lemma~\ref{lem:emcg-eval-introduce} and Lemma~\ref{lem:eval-equals-reduce}.
By Invariant~\ref{inv:alg}, $S_j(\bar U_j)$ therefore contains
exactly one game $\mathcal R_j'$ with $\mathcal R_j' \cong \mathcal R_j$.
If $\mathcal R_j' = \true$, then $\mathcal R_i = \true$ by
Lemma~\ref{lem:emcg-eval-introduce} and Lemma~\ref{lem:eval-equals-reduce}.
Otherwise, the algorithm computes
$\mathcal R_i' = \combine(\mathcal R_j', \mathcal R(\bar U_i))$.
By Lemma~\ref{lem:combine}, $\mathcal R_i' \cong \mathcal R_i$, which implies part~(I)
of the invariant.

Conversely, consider $\mathcal R_i \in S_i(\bar U_i)$. 
Then either $\mathcal R_i = \true$ and there is $\mathcal R_j \in S_j(\bar U_j)$
with $\mathcal R_j = \true$, or there is $\mathcal R_j \in S_j(\bar U_j)$ with
$\mathcal R_i \cong \combine(\mathcal R_j, \mathcal R_i(\bar U_i))$.
By the invariant for~$j$, $\mathcal{RED}_j(\mathcal R_j) \neq \emptyset$.
From this we get there is $\mathcal R_j'  \in \mathcal{RED}_j(\bar U_j, \mathcal R_j)$ such that
$\mathcal R_j' \cong \mathcal R_j$ and $\mathcal R_j' = \redgame(\ms A_j', X_j, \varphi)$
for some $\ms A_j' \in \mathcal{EXP}_j(\bar U_j)$.
Let $\ms A_i' \in \mathcal{EXP}_i(\bar U_i)$, chosen
in a way such that $(R_1^{\ms A_j'}, \ldots, R_l^{\ms A_j'}) = (R_1^{\ms A_i'} \cap A_j, \ldots, R_l^{\ms A_i'} \cap A_j)$.

If $\mathcal R_j = \true$, then, by Lemma~\ref{lem:emcg-eval-introduce} and
Lemma~\ref{lem:eval-equals-reduce}, $\redgame(\ms A_i', X_i, \varphi) = \true \in S_i(\bar U_i)$.
Otherwise, $\redgame(\ms A_i, X_i, \varphi) \cong \combine(\mathcal R_j,
\mathcal R(\bar U_i))$ by Lemma~\ref{lem:combine}.  Either case implies~(II).

Finally, let $\mathcal R_i \in S_i(\bar U_i)$ and 
$\ms O_i \in \mathcal{EXP}_i(\bar U_i, \mathcal R_i)$ with
$\reduce(\ms O_i, X_i, \varphi) \neq \false$ and
\begin{multline*}
\sum_{k=1}^{l} \alpha_k |R_k^{\ms O_i} \setminus X_i|
= \min\Biggl\{\,\sum_{k=1}^{l} \alpha_k |R_k^{\ms A_i'} \setminus X_i|
 \Biggm| \\
 \ms A_i' \in \mathcal{EXP}_i(\bar U_i, \mathcal R_i)
 \wedge \reduce(\ms A_i', X_i, \varphi) \neq \false
 \,\Biggr\}.
\end{multline*}
Let $\ms O_j = \ms O_i[A_j]$.  By Lemmas \ref{lem:emcg-eval-introduce}
and~\ref{lem:eval-equals-reduce},
$\mathcal R_j = \redgame(\ms O_j, X_j, \varphi) \neq \false$.
Therefore, either $\mathcal R_j = \redgame(\ms O_j, X_j, \varphi) = \mathcal R_i = \true$,
or otherwise 
$\combine(\mathcal R_j, \mathcal R(\bar U_i)) \cong \redgame(\ms O_j \cup \ms A_i(\bar U_i),
 X_i, \varphi) \cong \mathcal R_i$ by Lemma~\ref{lem:combine}.

We need that $\ms O_j$ is optimal for $\mathcal{EXP}_j(\bar U_j, \mathcal R_j)$.
To this end, assume there was $\ms A_j' \in \mathcal{EXP}_j(\bar U_j, \mathcal R_j)$
with $\mathcal R_j' = \redgame(\ms A_j', X_j, \varphi)$, such that either $\mathcal R_j' = \true$ or
$\mathcal R_i \cong \combine(\mathcal R_j', \ms A_i(\bar U_i))$, and furthermore
$$
\sum_{k=1}^{l} \alpha_k |R_k^{\ms O_j} \setminus X_j|
> 
\sum_{k=1}^{l} \alpha_k |R_k^{\ms A_j'} \setminus X_j|.
$$
Since, $\mathcal R_j' \cong \mathcal R_j$, we have, by Lemma~\ref{lem:combine},
$$
\mathcal R_i' \cong \begin{cases}
\true & \text{if } \mathcal R_j = \true \text{ and}\\
\combine(\mathcal R_j, \mathcal R(\bar U_i)) & \text{otherwise}, \\
\end{cases}
$$
where $\mathcal R_i' = \redgame(\ms A_i', X_i, \varphi)$ and
$\ms A_i' = \ms A_j' \cup \ms A_i(\bar U_i)$.
Therefore, 
$$
\sum_{k=1}^{l} \alpha_k |R_k^{\ms O_i} \setminus X_i|
> 
\sum_{k=1}^{l} \alpha_k |R_k^{\ms A_i} \setminus X_i|,
$$
a contradiction to the minimality of $\ms O_i$.
We conclude that $\ms O_j$ is optimal for $\mathcal{EXP}_j(\bar U_j, \mathcal R_j)$.
From this we get that
$$
\val_j(\mathcal R_j) = \sum_{k=1}^{l} \alpha_k |R_k^{\ms O_j} \setminus X_j|
$$
by the invariant for~$j$, which implies~(III).
\end{proof}


\begin{lemma}
\label{lem:alg-forget}
Let $i \in T$ be a forget node of the tree decomposition and $j \in T$
be the unique child of $i$.
If Invariant~\ref{inv:alg} holds for~$j$ before the algorithm processes~$i$,
then it also holds for~$i$.
\end{lemma}
\begin{proof}
Let $j$ be the unique child of $i$ and $X_i \cup \{x\} = X_j$ for $x \notin X_i$.
Note that $\ms A_i = \ms A_j$.
Let $\bar U_i = (U_{i,1}, \ldots, U_{i,l}) \in \mathcal{AR}_i \cap X_i$.

Consider $\ms A_i' \in \mathcal{EXP}_i(\bar U_i)$ with
$\mathcal R_i = \redgame(\ms A_i', X_i, \varphi) \neq \false$ and let
$\bar U_j = (U_{j,1}, \ldots, U_{j,l}) = (R_1^{\ms A_i'}, \ldots, R_l^{\ms A_i'}) \cap X_j
\in \mathcal{AR}_j \cap X_j$ and $\mathcal R_j = \redgame(\ms A_i', X_j, \varphi)$.
Then, by Lemma~\ref{lem:emcg-eval-introduce} and Lemma~\ref{lem:eval-equals-reduce},
$\mathcal R_j \neq \false$.
Therefore, by the invariant for~$j$, there is $\mathcal R_j' \in S_j(\bar U_j)$ with
$\mathcal R_j' \cong \mathcal R_j$.
If $\mathcal R_j' = \mathcal R_j = \true$, then, by Lemma~\ref{lem:emcg-eval-forget},
also $\mathcal R_i = \true$.  Otherwise, the algorithm computes
$\mathcal R_i' = \forget(\mathcal R_j', x) \cong \mathcal R_j$.  Either case implies~(I).

Conversely, consider $\mathcal R_i \in S_i(\bar U_i)$. 
Then either $\mathcal R_i = \true$ and there is $\bar U_j \in \mathcal{AR}_j \cap X_j$ and
$\mathcal R_j \in S_j(\bar U_j)$ with $\mathcal R_j = \true$ and $\bar U_i = \bar U_j \cap X_i$,
or
there is $\bar U_j \in \mathcal{AR}_j \cap X_j$ and $\mathcal R_j \in S_j(\bar U_j)$, such that
$\bar U_i = \bar U_j \cap X_i$ and
$\mathcal R_i \cong \forget(\mathcal R_j, x)$.
By the invariant for~$j$, in either case $\mathcal{RED}_j(\mathcal R_j) \neq \emptyset$.
Therefore, there is $\mathcal R_j'  \in \mathcal{RED}_j(\bar U_j, \mathcal R)$, where
$\mathcal R_j' \cong \mathcal R_j$ and $\mathcal R_j' = \redgame(\ms A_j', X_j, \varphi)$,
for some $\ms A_j' \in \mathcal{EXP}_j(\bar U_j)$.

Let $\ms A_i' = \ms A_j'$.  If $\mathcal R_j = \true$, then, by Lemmas
\ref{lem:emcg-eval-forget} and~\ref{lem:eval-equals-reduce},
$\redgame(\ms A_i', X_i, \varphi) = \true \in S_i(\bar U_i)$.
Otherwise, $\redgame(\ms A_i', X_i, \varphi) \cong \forget(\mathcal R_j, x) \cong \mathcal R_i$
according to
Lemma~\ref{lem:forget}.  Either case implies~(II).

Finally, consider $\mathcal R_i \in S_i(\bar U_i)$ and 
let $\ms O_i \in \mathcal{EXP}_i(\bar U_i, \mathcal R_i)$ such that
\begin{multline*}
\sum_{k=1}^{l} \alpha_k |R_k^{\ms O_i} \setminus X_i|
= 
\min\Biggl
\{\,\sum_{k=1}^{l} \alpha_k |R_k^{\ms A_i'} \setminus X_i|
 \Biggm|
\\ 
 \ms A_i' \in \mathcal{EXP}_i(\bar U_i, \mathcal R_i) \wedge 
 \reduce(\ms A_i', X_i, \varphi) \neq \false \,\Biggr\}
\end{multline*}
and $\reduce(\ms O_i, X_i, \varphi) \neq \false$.
Let $\ms O_j = \ms O_i$. Then, by Lemmas~\ref{lem:emcg-eval-forget}
and~\ref{lem:eval-equals-reduce}, $\mathcal R_j = \redgame(\ms O_j, X_j, \varphi) \neq \false$.
By (II), there is $\mathcal R_j' \in S_j(\bar U_j)$ with
$\mathcal R_j' \cong \mathcal R_j$, where
$\bar U_j = (R_1^{\ms O_j} \cap X_j, \ldots, R_l^{\ms O_j} \cap X_j)$.
Analogue to the previous case, we obtain that $\ms O_j$ is optimal in
$\mathcal{RED}_i(\bar U_j, \mathcal R_j)$.  Therefore, by the induction hypothesis,
$$
\val_j(\mathcal R_j) = \sum_{k=1}^{l} \alpha_k |R_k^{\ms O_j} \setminus X_j|
= \sum_{k=1}^{l} \alpha_k |R_k^{\ms O_j} \setminus X_i| - \sum_{k=1}^{l} \alpha_k (x \in R_k^{\ms O_i} \setminus X_i),
$$
which implies~(III).

\end{proof}


\begin{lemma}
\label{lem:alg-join}
Let $i \in T$ be a join node of the tree decomposition 
with children~$j_1, j_2 \in T$.
If Invariant~\ref{inv:alg} holds for~$j_1$ and $j_2$ before the algorithm processes~$i$,
then it also holds for~$i$.
\end{lemma}
\begin{proof}
Note that $X_i = X_{j_1} = X_{j_2}$.
Let $\bar U = (U_1, \ldots, U_l) \in \mathcal{AR}_i \cap X_i = \mathcal{AR}_{j_1} \cap X_{j_1} = 
\mathcal{AR}_{j_2} \cap X_{j_2}$.

Consider $\ms A_i' \in \mathcal{EXP}_i(\bar U)$ and let, for $j \in \{j_1, j_2\}$, 
be $\ms A_j' = \ms A_i'[A_j]$.
If $\mathcal R_i = \redgame(\ms A_i', X_i, \varphi) \neq \false$, then, for $j \in \{j_1, j_2\}$,
also $\mathcal R_j = \redgame(\ms A_j', X_j, \varphi) \neq \false$ by
Lemma~\ref{lem:emcg-eval-union} and Lemma~\ref{lem:eval-equals-reduce}.
By the invariant for $j \in \{j_1, j_2\}$, $S_j(\bar U)$ therefore
contains exactly one~$\mathcal R_j'$ with $\mathcal R_j' \cong \mathcal R_j$.  
If $\mathcal R_j' = \true$, then $\mathcal R_i = \true$ by
Lemma~\ref{lem:emcg-eval-union} and Lemma~\ref{lem:eval-equals-reduce}.
Otherwise, the algorithm computes
$\mathcal R_i' = \combine(\mathcal R_{j_1}, \mathcal R_{j_2})$.
By Lemma~\ref{lem:combine}, $\mathcal R_i' \cong \mathcal R_i$, which implies~(I).

Conversely, consider $\mathcal R_i \in S_i(\bar U)$. 
Then either $\mathcal R_i = \true$ and there is $j \in \{j_1, j_2\}$ and
$\mathcal R_j \in S_j(\bar U)$ with $\mathcal R_j = \true$, or
there is $(\mathcal R_{j_1}, \mathcal R_{j_2}) \in S_{j_1}(\bar U) \times S_{j_2}(\bar U)$,
such that $\mathcal R_i \cong \combine(\mathcal R_{j_1}, \mathcal R_{j_2})$.
By the invariant for $j \in \{j_1, j_2\}$, we have
$\mathcal{RED}_j(\mathcal R_j) \neq \emptyset$, and therefore there is
$R_j'  \in \mathcal{RED}_j(\bar U, \mathcal R_j)$ with $\mathcal R_j' \cong \mathcal R_j$
and $\mathcal R_j' = \redgame(\ms A_j', X_j, \varphi)$,
where $\ms A_j' = (\ms A_j, R_1^{\ms A_j}, \ldots, R_l^{\ms A_j}) \in \mathcal{EXP}_j(\bar U)$.
Let $\ms A_i' = (\ms A_i, R_1^{\ms A_i}, \ldots, R_l^{\ms A_i}) \in \mathcal{EXP}_i(\bar U)$, such that
$R_k^{\ms A_i} = R_k^{\ms A_{j_1}} \cup R_k^{\ms A_{j_2}}$ for all $1 \le k \le l$.

If $\true \in \{\mathcal R_{j_1}, \mathcal R_{j_2}\}$, then, by Lemmas
\ref{lem:emcg-eval-union} and~\ref{lem:eval-equals-reduce},
$\redgame(\ms A_i', X_i, \varphi) = \true \in S_i(\bar U)$.
Otherwise, $\redgame(\ms A_i, X_i, \varphi) \cong \combine(\mathcal R_{j_1},
\mathcal R_{j_2})$ by Lemma~\ref{lem:combine}.  Either case implies~(II).

Now consider $\mathcal R_i \in S_i(\bar U)$ and
$\ms O_i \in \mathcal{EXP}_i(\bar U, \mathcal R_i)$ with $\reduce(\ms O_i, X_i, \varphi) \neq \false$ and
\begin{multline*}
\sum_{k=1}^{l} \alpha_k |R_k^{\ms O_i} \setminus X_i|
= \min\Biggl\{\,\sum_{k=1}^{l} \alpha_k |R_k^{\ms A_i'} \setminus X_i|
 \Biggm| \\
 \ms A_i' \in \mathcal{EXP}_i(\bar U, \mathcal R_i)
 \wedge \reduce(\ms A_i', X_i, \varphi) \neq \false
 \,\Biggr\}.
\end{multline*}
Let, for $j \in \{j_1, j_2\}$, $\ms O_j = \ms O_i[A_j]$. Then, by Lemma~\ref{lem:emcg-eval-union} and
Lemma~\ref{lem:eval-equals-reduce}, $\mathcal R_j = \redgame(\ms O_j, X_j, \varphi) \neq \false$.
Therefore, either $\mathcal R_j = \redgame(\ms O_j, X_j, \varphi) = \mathcal R_i = \true$ for
some $j \in \{j_1, j_2\}$, or $\combine(\mathcal R_{j_1}, \mathcal R_{j_2}) \cong
\redgame(\ms O_{j_1} \cup \ms O_{j_2}, X_i, \varphi) \cong \mathcal R_i$ by Lemma~\ref{lem:combine}.

Assume there were $j \in \{j_1, j_2\}$, say $j = j_1$,
and $\ms A_j' \in \mathcal{EXP}_j(\bar U, \mathcal R_j)$
with $\mathcal R_j' = \redgame(\ms A_j', X_j, \varphi)$, such that
$$
\sum_{k=1}^{l} \alpha_k |R_k^{\ms O_j} \setminus X_j|
> 
\sum_{k=1}^{l} \alpha_k |R_k^{\ms A_j'} \setminus X_j|
$$
and either
$\mathcal R_j' = \true$ or
$\mathcal R_i \cong \combine(\mathcal R_j', \mathcal R_{j_2}')$ for some
$\mathcal R_{j_2}' \in \mathcal{RED}_{j_2}(\bar U, \mathcal R_{j_2})$.
Since $A_{j_1} \cap A_{j_2} = X_i$, structures
$(\ms A_{j_1}', \ms A_{j_2}') \in
\mathcal{RED}_{j_1}(\bar U) \times \mathcal{RED}_{j_2}(\bar U)$
are compatible.
By the invariant, part~(II), we have $\mathcal R_{j_2}' \cong \redgame(\ms A_{j_2}', X_{j_2}, \varphi)$ for some
$\ms A_{j_2}' \in \mathcal{EXP}_{j_2}(\bar U)$.
Without loss of generality, we assume $\ms A_{j_2}' = \ms O_{j_2}$, since each $\ms A_{j_2}'$ with
$$
\sum_{k=1}^{l} \alpha_k |R_k^{\ms O_{j_2}} \setminus X_j|
\ge
\sum_{k=1}^{l} \alpha_k |R_k^{\ms A_{j_2}'} \setminus X_j|
$$
yields the same contradiction.
Therefore, since $\mathcal R_j' \cong \mathcal R_j$, we have
$$
\mathcal R_i' \cong \begin{cases}
\true & \text{if } \mathcal R_j' = \true \text{ or } \mathcal R_{j_2} = \true \text{ and}\\
\combine(\mathcal R_j', \mathcal R_{j_2}) & \text{otherwise}, \\
\end{cases}
$$
by Lemma~\ref{lem:combine},
where $\mathcal R_i' = \redgame(\ms A_j' \cup \ms O_{j_2}, X_i, \varphi)$.
Therefore, 
$$
\sum_{k=1}^{l} \alpha_k |R_k^{\ms O_i} \setminus X_i|
> 
\sum_{k=1}^{l} \alpha_k |R_k^{\ms A_j'} \setminus X_i|
+ \sum_{k=1}^{l} \alpha_k |R_k^{\ms O_{j_2}} \setminus X_i|
$$
a contradiction to the minimality of $\ms O_i$.
Therefore, for $j \in \{j_1, j_2\}$, $\ms O_j$ is optimal in
$\mathcal{EXP}_i(\bar U, \mathcal R_j)$, and
$$
\val_j(\mathcal R_j) = \sum_{k=1}^{l} \alpha_k |R_k^{\ms O_j} \setminus X_j|
$$
by the invariant for~$j$.  By (II), there is
$\mathcal R_j' \in S_j(\bar U)$ with $\mathcal R_j' \cong \mathcal R_j$,
which then implies~(III).
\end{proof}

\begin{lemma}
\label{lem:alg-root}
Let $r = \root(\mathcal T)$ be the root of the tree decomposition, where
$X_r = \emptyset$, and let Invariant~\ref{inv:alg} hold for~$r$.
Let $\bar U = (\emptyset, \ldots, \emptyset)$ and
$$\OPT = \linmsoopt$$ be an optimal solution for the \linmso-problem.
Then 
$$
\OPT = \min\{\,\val_r(\mathcal R) \mid \mathcal R \in S_r(\bar U) \wedge \eval(\convert(\mathcal R)) = \true\,\}.
$$
\end{lemma}
\begin{proof}
Note that $\ms A = \ms A_r$. 
Let $\ms A'$ be optimal, i.e., let $\ms A'$ be a
$\tau$-expansion of~$\ms A$, such that $\ms A' \models \varphi$
and
$$
\sum_{k=1}^{l} \alpha_k |R_k^{\ms A'} \setminus X_r| = \sum_{k=1}^{l} \alpha_k |R_k^{\ms A'}| = \OPT.
$$
Let $\mathcal R = \redgame(\ms A', X_r, \varphi)$.
We have $\eval(\mcg(\ms A', \varphi)) = \true$ since $\ms A' \models \varphi$, and therefore
$$
\eval(\convert(\mathcal R)) = \eval(\mcg(\ms A', \varphi)) = \true$$
by Lemma~\ref{lem:convert}.
Note that $X_r = \emptyset$ and therefore $\mathcal{AR}_r \cap X_r = \{(\emptyset, \ldots, \emptyset)\}$.
By Invariant~\ref{inv:alg}, part~(I), there is $\mathcal R' \in S_j(\bar U)$, such that
$\mathcal R' \cong \mathcal R$, which implies $\OPT = \val_r(\mathcal R')$
by part~(III) and the optimality of~$\ms A'$ for $\mathcal{EXP}_r(\bar U, \mathcal R)$.
Since $\eval(\convert(\mathcal R')) = \true$, we also have
$$
\OPT = 
\val_r(\mathcal R') 
\ge \min\{\,\val_r(\mathcal R'') \mid \mathcal R'' \in S_r(\bar U) \wedge \eval(\convert(\mathcal R'')) = \true\,\}.
$$
Conversely, let $\mathcal R \in S_r(\bar U)$, such that $\eval(\convert(\mathcal R)) = \true$ and
$$
\val_r(\mathcal R) = 
	\min\{\,\val_r(\mathcal R') \mid \mathcal R' \in S_r(\bar U) \wedge \eval(\convert(\mathcal R')) = \true\,\}.
$$
By part~(II) of the invariant, there is a $\tau$-expansion $\ms A''$ of~$\ms A$, such that
$\mathcal R \cong \reduce(\ms A'', X_r, \varphi)$.
Since $\eval(\convert(\mathcal R)) = \true$, we have $\ms A'' \models \varphi$ by
Lemma~\ref{lem:convert}.
Without loss of generality, we can assume by part~(III), 
that $\ms A''$ is optimal for $\mathcal{EXP}_r(\bar U, \mathcal R)$,
i.e., 
$\val_r(\mathcal R) = \sum_{k=1}^{l} \alpha_k |R_k^{\ms A''} \setminus X_r|$.
We then directly conclude
$$
\val_r(\mathcal R) = \sum_{k=1}^{l} \alpha_k |R_k^{\ms A''}| \ge \OPT
 = \sum_{k=1}^{l} \alpha_k |R_k^{\ms A'}|.
$$
\end{proof}

We can now prove Theorem~\ref{thm:main}.

\begin{proof}[Proof of~Theorem~\ref{thm:main}]
Using induction over the structure of the tree decomposition and Lemmas~\ref{lem:alg-leaf}--\ref{lem:alg-join}
for the respective nodes, we know that Invariant~\ref{inv:alg} holds for the root node of the tree decomposition
after the algorithm has finished Phase~1.
By Lemma~\ref{lem:alg-root}, the algorithm outputs the correct
solution in Phase~2.

For the running time, consider $i \in T$.  We have $|\mathcal{AR}_i \cap X_i| = O(2^{|X_i|l})$,
which for constant $l \le |\tau|$ and $|X_i|\le w+1$ is a constant.
For $\bar U \in \mathcal{AR}_i \cap X_i$, consider the set $S_i(\bar U)$.  Since the algorithm
only inserts games into $S_i(\bar U)$, if $S_i(\bar U)$ does not already contain 
an equivalent game, 
$$
|S_i(\bar U)| \le N_{X_i, \varphi} \leq \exp^{\qr(\varphi)+1}((|X_i|+1)^{O(\|\varphi\|)}),
$$
by Lemma~\ref{lem:number}, which for bounded $|X_i|$ is constant.  Furthermore, by Lemma~\ref{lem:size}, 
for each $\mathcal R \in S_i(\bar U)$, 
$$
|\mathcal R| \le \exp^{\qr(\varphi)+1}((|X_i|+1)^{O(\|\varphi\|)}),
$$
again a constant.  Finally, each position of each game is of the form $(\ms H, X_i, \psi)$,
where $\|\psi\| \le \|\varphi\|$ and $\|\ms H\| = O(|X_i| + \|\varphi\|)$, where
$\|\ms H\|$ denotes the size of a suitable encoding of~$\ms H$.
All operations on games, i.e., $\reduce()$, $\eval()$,
$\combine()$, $\forget()$, and $\convert()$,
therefore take constant time.

In total, at a node $i \in T$, a constant number of entries or pairs, respectively, is considered,
and each operation takes constant time.  The running time is therefore $O(|T|)$.
\end{proof}

\subsection{Extensions}
\label{sec:ext-linmso}

\paragraph{Semiring Homomorphisms}
Note that the algorithm implicitly used a homomorphism
$$
h\colon (U_1, \ldots, U_l) \mapsto \sum_{k=1}^l \alpha_k|U_k|
$$
from the semiring $(\Pow(\mathcal{AR}_r), \hat \uplus, \cup, \hat
\emptyset, \emptyset)$ into the semiring $(\mathbf Z_\infty, +, \min,
0, \infty)$.  Here, $\Pow(\mathcal{AR}_r)$ is the set of all possible
interpretations of the free relation symbols (i.e., a set of tuples of sets), $\hat \uplus$ is a
component-wise, disjoint union with neutral element $\hat \emptyset
= (\emptyset, \ldots, \emptyset)$, and
$\cup$ is the regular union of sets. 
The extension to other semiring homomorphisms, e.g., to count the
number of interpretations satisfying the MSO property $\varphi$,
is rather straightforward.  See~\cite{CM93} for a list of
many interesting semirings.

\paragraph{Many-sorted Structures}
In this article, we considered one-sorted structures, i.e., structures
whose universe contains objects of a single sort only.  The
corresponding theory is also called $\rm MS_1$-theory in the
literature and is strictly less powerful than corresponding logics
for multi-sorted structures.
For instance, recall from Example~\ref{ex:taugraph} that
a graph $G = (V, E)$ can in a natural way be identified with a structure
over the vocabulary $\taugraph = (\adj)$, where $V$ is identified with
the one-sorted universe of \emph{vertices}, and $\adj$ is interpreted as~$E$.
The \textsc{Hamiltonian Path} problem for graphs
cannot be expressed in $\mso(\taugraph)$, since this requires the
use of edge-set quantification (see~\cite{EF99}, for instance).

Fortunately, this poses no restriction in algorithmic applications.
Firstly, it is not hard to extend the techniques in this paper to
many-sorted structures.  Courcelle's original works~\cite{Cou90,Cou90a}
were already proven for many-sorted structures.  Secondly, one can
easily simulate many-sorted structures by introducing relation
symbols that distinguish the respective objects in a common universe
accordingly.  For example, one can consider the incidence graph of
a graph and introduce unary relation symbols $V$ and $E$, which
allow to distinguish objects of sort ``vertex'' or ``edge'', and
a new binary relation symbol $\inc$ for the incidence relation.
Transforming a structure and a corresponding tree decomposition
accordingly can be done efficiently and does not increase the width
of the decomposition.
Graphs with multi-edges can be represented similarly.

\section{Solving Concrete Problems}
\label{sec:concrete-problems}

In the analysis of the running time of the algorithm, we were rather
pessimistic w.r.t.\ the constants hidden in the $O(|T|)$.  Recall
that unless $\rm P = NP$, these cannot be bounded by an elementary
function, i.e., the running time of the algorithm cannot be
$O(f(\|\varphi\|, w) n)$ for a fixed function $f\colon \mathbf N
\times \mathbf N \to \mathbf N$ that is a nesting of exponentials
of bounded depth~\cite{FG04}.

The picture changes dramatically once we assume the \emph{problem}
is fixed, i.e., the problem description consisting of the vocabulary
$\tau$, a formula $\varphi \in \mso(\tau)$ and the integers $\alpha_1,
\ldots, \alpha_l \in \mathbf Z$ are constants.  Specialized and
comparably efficient algorithms exist for many problems, e.g., of
running time $O(2^w\poly(w) n)$ for the \textsc{Minimum Vertex Cover}
problem, or of $O(3^w \poly(w) n)$ for \textsc{Minimum Dominating Set}
and \textsc{3-Colorability}, cf.~\cite{TP97,RBR09},
where $\poly(w)$ is a fixed polynomial in~$w$.  Recent results
furthermore indicate that better running times are improbable~\cite{LMS10}.
Assuming small treewidth, such algorithms might still turn out to
be feasible in many practical applications, cf.~\cite{Bod98}.

In this section, we estimate the running times of our generic
approach for the three aforementioned problems.  Let $(\mathcal T,
\mathcal X)$ be a tree decomposition of the input graph structure $\ms A$ over
$\taugraph$,
where $\mathcal T = (T, F)$ and $\mathcal X = (X_i)_{i \in T}$ with
$|X_i| \le w$ for all $i \in T$, i.e., $\ms A$ has treewidth at
most~$w-1$.

\subsection{\textsc{Minimum Vertex Cover}}

\begin{figure}[tbp]
\begin{center}
\includegraphics[scale=0.8]{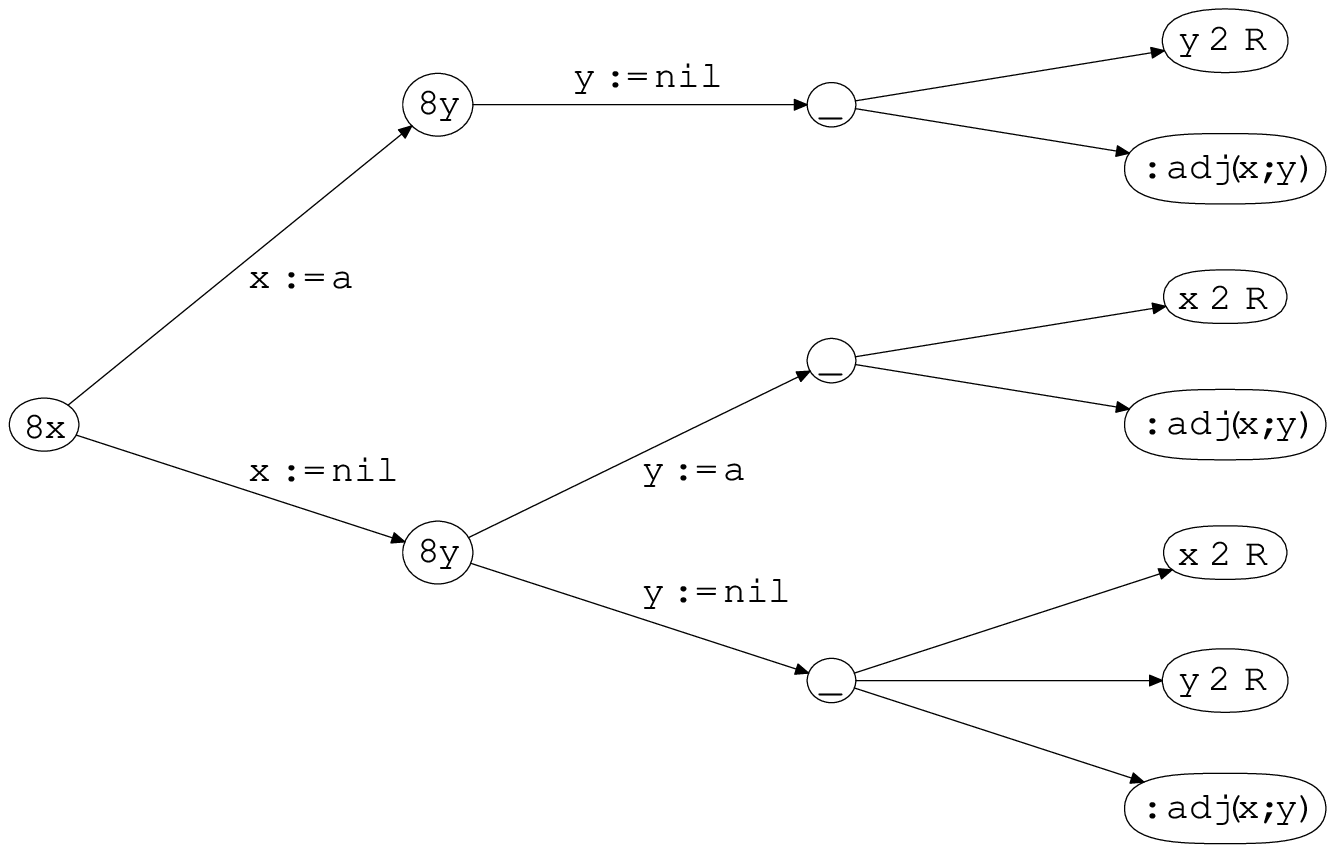}
\end{center}
\caption{Simplified schematic of $\reduce(\ms A, X, \vc)$,
where $\ms A$ has universe $A = \{a\}$, $X = A$ and $a \notin R^{\ms A}$.
If any of the symbols $x$ or $y$ remains uninterpreted
(cases $x := \nil$ and $y := \nil$ in the figure),
then some of the plays in $\emcg(\ms A, \emptyset, \vc)$ end with
a draw and still persist in the reduced game.  If $\ms A = \ms A_i$ and $X = X_i$
for a node $i$ of a tree decomposition, then this essentially means that it is
still open whether nodes in the ``future'' of~$i$ will be adjacent or whether
they will be contained in~$R$.\label{fig:vc-unfolding}}
\end{figure}

Recall from Example~\ref{ex:mso-examples} that the formula
$$
\vc = \forall x \forall y (\neg \adj(x,y) \vee x \in R \vee y \in R) \in \mso(\taugraph \cup \{R\}) 
$$
is true on a $(\taugraph, R)$-structure $(\ms G, U)$ if and only
if $U \subseteq G$ is a vertex cover for the graph~$\ms G$. 
Using the notation from the previous section, we claim
that for each $i \in T$ and for all $\bar U \in \mathcal{AR}_i \cap
X_i$, the set $S_i(\bar U)$ contains at most one entry $\mathcal
R$, and if $\mathcal R \in S_i(\bar U)$ for some $\bar U$, then
$|\mathcal R| = \poly(w)$.  To this end, consider arbitrary $\bar
U \in \mathcal{AR}_i \cap X_i$ and let $\ms A_i' \in \mathcal{EXP}_i(\bar
U)$.

For any $a \in A_i$, such that $a \in R^{\ms A_i'}$, the verifier
has a winning strategy on $\mathcal G = \emcg(\ms A_i'', X_i, \forall
y\ldots)$, where $\ms A_i'' = (\ms A_i', a)$ with $x^{\ms A_i''} = a$,
since the atomic formula~$x
\in R$ is always satisfied for all~$y$. 
Therefore, $\eval(\mathcal
G) = \redgame(\mathcal G) = \true$ and $\reduce()$ removes the
subgame $\mathcal G$ from $\emcg(\ms A_i', X_i, \vc)$.

Consider now a subgame $\emcg(\ms A_i'', X_i, \forall y \ldots)$,
where $\ms A_i'' = (\ms A_i', a)$ with $a \notin R^{\ms A_i''}$. 
If there is $b \in A_i$, such that $(a, b) \in \adj^{\ms A_i''}$ and $b \notin
R^{\ms A_i''}$, then the falsifier has a winning strategy on
$\emcg(\ms A_i', X_i, \vc)$ and consequently $\redgame(\ms A_i',
X_i, \vc) = \false$.
If otherwise for all $b \in A_i$ either $b \in R^{\ms A_i''}$ or 
$(a, b) \notin \adj^{\ms A_i''}$, then
we get $\redgame((\ms A_i', a, b), X_i, \ldots)) = \true$, and the
corresponding
subgame will be removed by $\reduce()$.
Therefore only the subgame on $\ms A_i''$ with $y^{\ms A_i''} = \nil$ remains
undetermined.  We conclude
$\emcg((\ms A_i', b_1), X_i, \forall y \ldots) \cong \emcg((\ms A_i', b_2), X_i, \forall y \ldots)$
for all $b_1,b_2 \in A_i \setminus X_i$.

Due to the symmetry of $x$ and $y$ in the vertex cover formula, we can argue
analogously for the cases where the roles of $x$ and $y$ have been interchanged. 
Therefore, $\mathcal R_1 \cong \mathcal R_2$ for all
$\mathcal R_1, \mathcal R_2 \in \mathcal{RED}_i(\bar U)$, from which we conclude
$|S_i(\bar U)| \le 1$.
Each game is of size $|\mathcal R| = O(w)$, since by above considerations
$$
|\subgames(\reduce(\ms A_i'', X_i, \forall y \ldots))| \le
\begin{cases}
|X_i| + 1 + 1	& \text{if $x^{\ms A_i''} = \nil$} \\
1		& \text{if $x^{\ms A_i''} \in A_i$}
\end{cases}
$$
and
$|\subgames(\reduce(\ms A_i', X_i, \vc))| \le |X_i| + 1 + 1$:
In both cases, we have $|X_i|$ subgames for the vertices in $X_i$,
one subgame for all vertices in $A_i \setminus X_i$
(since all of them are equivalent), and
one subgame for the case that $x$ and $y$, respectively, remain uninterpreted.
See Figure~\ref{fig:vc-unfolding} for an example.

It is not hard to see that $\reduce(\mathcal R)$, $\eval(\mathcal R)$,
$\convert(\mathcal R)$, $\forget(\mathcal R_1)$
and
$\combine(\mathcal R_1, \mathcal R_2)$
can be implemented in a way such that they
run in time polynomial in $|\mathcal R|$ and
$|\mathcal R_1| + |\mathcal R_2|$.
Hence, we immediately find that the generic algorithm introduced
in this article reaches, up to factors polynomial in~$w$, the running
time of $O(2^w n)$ of the specialized algorithm, since $|\mathcal{AR}_i
\cap X_i| = 2^{|X_i|}$ for all $i \in T$.

\subsection{\textsc{Minimum Dominating Set}}

The formula
$$
\ds = \forall x (x \in R \vee \exists y (y \in R \wedge \adj(x,y)))
\in \mso(\taugraph \cup \{R\})
$$
holds in $(\ms G, U)$ if and only if $U \subseteq G$ is a dominating set for
the graph~$\ms G$.
Let for each $i \in T$ and $\bar U = (U_1) \in \mathcal{AR}_i \cap X_i$
be $k = |X_i| - |U_1|$.  We claim that $|S_j(\bar U)| \le 2^k$.
To this end, let again $\ms A_i' \in \mathcal{EXP}_i(\bar U)$ and
$\mathcal R = \redgame(\ms A_i', X_i, \ds)$.  Let $U \subseteq A_i$ be such
that $\ms A_i' = (\ms A_i, U)$.

If $U$ dominates $a \in A_i$, then either $a \in U$ and
$\redgame((\ms A_i, a), X_i, {x \in R}) = \true$, or there is $b \in
U$ that is adjacent to~$a$, and $\redgame((\ms A_i, a), X_i, \exists
y \ldots) = \true$.  In both cases we get $\mathcal R' = \redgame((\ms
A_i, a), X_i, x \in R \vee \exists y \ldots) = \true$, and therefore
$\mathcal R' \notin \subgames(\mathcal R)$.

\begin{figure}[tbp]
\begin{center}
\includegraphics[scale=0.8]{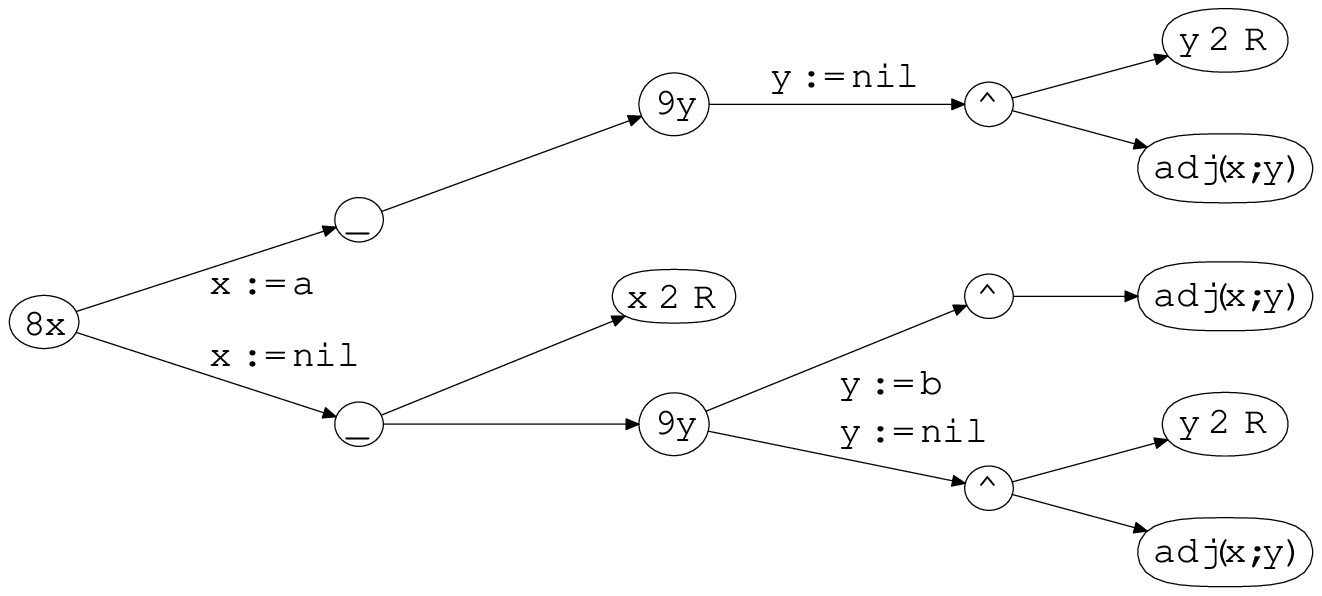}
\end{center}
\caption{Simplified schematic of $\reduce(\ms A, X, \ds)$, where
$\ms A$ has universe $A = \{a,b\}$, $X = A$, and $R^{\ms A} = \{b\}$, such
that $a$ and $b$ are not adjacent.
Then $a$ might still be dominated by a ``future'' vertex; the corresponding
plays (following the upper $y := \nil$ branch in the figure) end
with a draw and therefore persist in the reduced game.  Similarly,
the branch $x := \nil$ corresponds to the case that ``future''
vertices are chosen as interpretations for~$x$.  Such vertices can
also be dominated by $b$, which is represented by the $y := b$
branch in the figure.  \label{fig:ds-unfolding}}
\end{figure}

If $a \in A_i$ is not dominated by~$U$, then $\redgame((\ms A_i',
a), X_i, x \in R) = \false$ and $\redgame(\ms A_i'', X_i, y \in R \wedge
\adj(x, y)) = \false$ for all $\ms A_i''$ with $x^{\ms A_i''} = a$
and $y^{\ms A_i''} \in A_i$. 
These games are therefore removed by~$\reduce()$.
Only the game $\redgame(\ms A_i'', X_i, y \in R \wedge \adj(x, y)$
with $x^{\ms A_i''} = a$ and $y^{\ms A_i''} = \nil$ remains
undetermined.  Thus for all $a_1, a_2 \in
A_i \setminus X_i$ that are not dominated by~$U$ we have $\redgame((\ms
A_i', a_1), X_i, x \in R \vee \exists y \ldots) \cong \redgame((\ms
A_i', a_2), X_i, x \in R \vee \exists y \ldots)$.

For $\ms A_i''$ with $x^{\ms A_i''} = \nil$ the game
$\reduce(\ms A_i'', X_i, x \in R)$ remains undetermined.
For all $b \in A_i \setminus U$ we have $\redgame((\ms A_i'', b), X_i, 
y \in R \wedge \adj(x,y)) = \false$ due to the subformula $y \in R$;
the corresponding subgame is therefore removed from
$\emcg(\ms A_i'', X_i, \exists y \ldots)$.
For all $b_1, b_2 \in A_i \cap U$ we again have
$\redgame((\ms A_i'', b_1), X_i, y \in R \wedge \adj(x,y)) \cong
\redgame((\ms A_i'', b_2), X_i, y \in R \wedge \adj(x,y))$.

All in all, either two games $\mathcal R_1, \mathcal R_2 \in
\mathcal{RED}_i(\bar U)$ only differ w.r.t.\ the subset of
\emph{undominated} nodes in $X_i$.  Since there are $k$ nodes in
$X_i$ that are not contained in $U$, this bounds $|S_j(\bar U)| \le
2^k$.   For each of them, we have
$|\subgames(\reduce(\ms A_i'', X_i, \forall x
\ldots))| \le |X_i| + 1 + 1$ corresponding to at most $|X_i|$
undominated nodes in~$X_i$, at most one undominated node in $A_i
\setminus X_i$ and the subgame for $\ms A_i''$ with $x^{\ms A_i''}
= \nil$.
Furthermore, 
$|\subgames(\reduce(\ms A_i'', X_i, {x \in R} \vee \exists y \ldots))| = O(1)$
for $\ms A_i''$ with $x^{\ms A_i''} \neq \nil$ and
$|\subgames(\reduce(\ms A_i'', X_i, {x \in R} \vee \exists y \ldots))| \le |X_i| + 1 + 1$.
We conclude that $|\mathcal R| = O(|X_i|)$.
See Figure~\ref{fig:ds-unfolding} for an example.

In total, at a node $i \in T$, there are therefore at most
$$
\sum_{k = 0}^{w} \binom w k 2^k = 3^w
$$
entries stored, and each entry has size $|\mathcal R | = O(w)$.
Nodes $i \in T$ of type \emph{leaf}, \emph{forget}  are therefore processed
in time $O(3^w\poly(w))$.  
For \emph{join} nodes $i \in T$ with children $j_1, j_2$, every pair
in $S_{j_1}(\bar U) \times S_{j_2}(\bar U)$ is considered.  Therefore, at most
$$
\sum_{\bar U \in \mathcal{AR}_i \cap X_i} |S_{j_1}(\bar U)| \cdot |S_{j_2}(\bar U)| \le
\sum_{k = 0}^{w} \binom w k 2^k 2^k = 5^w
$$
entries are considered, which yields a running time of $O(5^w \poly(w) n)$.
This does not yet match the best specialized algorithm for the
\textsc{Minimum Dominating Set} problem~\cite{RBR09} with a running
time of $O(3^w \poly(w) n )$, but is still faster than combining all pairs
with a running time of $\Theta(9^w \poly(w) n)$.
We note that both the $O(3^w \poly(w) n)$ bound from~\cite{RBR09} and the
$O(4^w n)$ bound from~\cite{AN02} exploit a
certain ``monotonicity'' property of domination like problems, which
does not hold for all problems that are expressible in MSO (\textsc{Independent Dominating Set}
being an example).

\subsection{\textsc{3-Colorability}}

\def\Part{{\it part}}
\def\IS{{\it is}}

The formula
\begin{multline*}
\col = \exists R_1 \exists R_2 \exists R_3 \Bigg( 
\forall x \bigg(\bigvee_{i=1}^3(x \in R_i) \wedge \bigwedge_{i\neq j}
	(\neg x \in R_i \vee \neg x \in R_j)\bigg) \wedge \\
	\forall x \forall y\bigg(\neg \adj(x, y) \vee
	\bigwedge_{i=1}^3 (\neg x \in R_i \vee \neg y \in R_i)\bigg)\Bigg) \in \mso(\taugraph)
\end{multline*}
defining the \textsc{3-Colorability} problem has
no free symbols.  
Therefore $\mathcal{AR}_i = \{()\}$, where $()$ is the empty tuple, and
the table $S_j$ contains at most one entry
$\mathcal R = \redgame(\ms A_i, X_i, \col)$.  We estimate the size of $\mathcal R$.
For, let $\col = \exists R_1 \exists R_2 \exists R_3 \varphi$, where $\varphi = \Part \wedge \IS$.
Here, $\Part = \forall x \ldots$ expresses that the $R_i$ are a partition of the
universe, and
$\IS = \forall x \forall y \ldots$ ensures that each $R_j$ is an independent set.

If $\bar U = (U_1, U_2, U_3) \in \Pow(A_i)^3$ is not a partition of $A_i$, then
the falsifier wins $\emcg((\ms A_i, U_1, U_2, U_3), X_i, \Part)$, and therefore
$\redgame((\ms A_i, \bar U), X_i, \varphi) \notin \subgames(\mathcal R)$.
Otherwise, $\emcg((\ms A_i, \bar U, a), X_i, \Part) = \true$ for all
$a \in A_i$ and undetermined when $x$ remains uninterpreted.
Using the same arguments as for the similar vertex cover formula~$\vc$, we have
$\mathcal R_1 \cong \mathcal R_2$ for all
$\bar U_j = (U_{j,1}, U_{j,2}, U_{j,3}) \in \Pow(A_i)^3$ with
$\bar U_1 \cap X_i = \bar U_2 \cap X_i$
and $\mathcal R_j = \redgame((\ms A_i, \bar U_j, X_i, \IS) \neq \false$,
$1 \le j \le 2$.
This implies $\redgame((\ms A_i, \bar U_1), X_i, \varphi) \cong \redgame((\ms A_i, \bar U_2), X_i, \varphi)$.
Thus, $\subgames(\mathcal R)$ contains at most $O(3^w)$ subgames
$\mathcal R_i = \redgame((\ms A_i, \bar U), X_i, \ldots) \neq \false$,
which bounds $|\mathcal R| = O(3^w\poly(w))$.

Thus, assuming $\combine(\mathcal R_1, \mathcal R_2)$ requires time
$\Theta(|\mathcal R_1| \cdot |\mathcal R_2| \cdot (\|\varphi\| + |X_i|))$, we only can bound the total running time
by $O(9^w\poly(n))$.
This can probably be improved to $O(3^w\poly(n))$ using a similar approach as for the tables
$S_j(\bar U)$.

\section{Practical Experiments and Conclusion}
\label{sec:conclusion}

\begin{table}
\begin{center}
\hfil
\vbox{%
\tabskip=0pt\def\0{\phantom{0}}
\halign to \hsize{# \hfil & # \hfil \tabskip=2em plus 2em&
	\tabskip=0pt plus 2em\hfil # & \hfil # & \hfil # \tabskip=2em plus 2em&
	\tabskip=0pt plus 2em\hfil # & \hfil # & \hfil #\tabskip=0pt\cr
\omit \hfill \textsc{Minimum Vertex Cover} \hfill \span \span \span \span \span \span \span \cr
\noalign{\medskip \hrule height 0.4pt \smallskip} 
&& \omit \hfill time in seconds \hfill \span \span & \omit \hfill memory in MB \hfill \span \span \cr
dimension & runs & min & max & median & min & max & median \cr
\noalign{\smallskip \hrule height 0.4pt \smallskip } 
$\01 \times 200$ & 40 & 0.2 & 0.2 & 0.2 & 1 & 1 & 1 \cr
$\02 \times 100$ & 40 & 0.3 & 0.4 & 0.3 & 1 & 1 & 1 \cr
$\03 \times 66$ & 40 & 0.5 & 0.9 & 0.6 & 1 & 1 & 1 \cr
$\04 \times 50$ & 40 & 0.1 & 1.0 & 0.1 & 1 & 1 & 1 \cr
$\05 \times 40$ & 40 & 0.2 & 0.4 & 0.3 & 1 & 2 & 2 \cr
$\06 \times 33$ & 40 & 0.3 & 0.9 & 0.5 & 2 & 3 & 2 \cr
$\07 \times 28$ & 40 & 0.6 & 1.9 & 1.0 & 2 & 5 & 3 \cr
$\08 \times 25$ & 40 & 1.2 & 4.6 & 2.3 & 3 & 9 & 5 \cr
$\09 \times 22$ & 40 & 1.9 & 13.6 & 5.2 & 5 & 18 & 10 \cr
$10 \times 20$ & 40 & 4.4 & 41.4 & 13.7 & 9 & 36 & 19.5 \cr
$11 \times 18$ & 40 & 11.3 & 156.4 & 46.2 & 16 & 62 & 39 \cr
$12 \times 16$ & 40 & 28.2 & 642.4 & 185.2 & 27 & 128 & 76 \cr
$13 \times 15$ & 40 & 61.3 & 2644.9 & 679.4 & 42 & 268 & 145.5 \cr
$14 \times 14$ & 40 & 308.7 & 10257.3 & 3017.1 & 80 & 468 & 283 \cr
\noalign{\smallskip \hrule height 0.4pt \medskip}
\omit \hfill \textsc{Minimum Dominating Set} \hfill \span \span \span \span \span \span \span \cr
\noalign{\medskip \hrule height 0.4pt \smallskip} 
&& \omit \hfill time in seconds \hfill \span \span & \omit \hfill memory in MB \hfill \span \span \cr
dimension & runs & min & max & median & min & max & median \cr
\noalign{\smallskip \hrule height 0.4pt \smallskip } 
$1 \times 200$ & 40 & 0.3 & 0.3 & 0.3 & 1 & 1 & 1 \cr
$2 \times 100$ & 40 & 0.8 & 1.0 & 0.9 & 1 & 1 & 1 \cr
$3 \times 66$ & 40 & 0.2 & 0.3 & 0.2 & 1 & 1 & 1 \cr
$4 \times 50$ & 40 & 0.6 & 0.9 & 0.8 & 2 & 3 & 2.5 \cr
$5 \times 40$ & 40 & 2.2 & 3.2 & 2.8 & 4 & 6 & 6 \cr
$6 \times 33$ & 40 & 8.3 & 12.8 & 11.6 & 11 & 17 & 15 \cr
$7 \times 28$ & 40 & 40.3 & 85.4 & 71.0 & 27 & 47 & 42 \cr
$8 \times 25$ & 40 & 238.9 & 681.2 & 493.7 & 68 & 137 & 112 \cr
$9 \times 22$ & 35 & 1605.7 & 8588.2 & 5235.3 & 170 & 386 & 332 \cr
\noalign{\smallskip \hrule height 0.4pt \medskip}
\omit \hfill \textsc{3-Colorability} \hfill \span \span \span \span \span \span \span \cr
\noalign{\medskip \hrule height 0.4pt \smallskip} 
&& \omit \hfill time in seconds \hfill \span \span & \omit \hfill memory in MB \hfill \span \span \cr
dimension & runs & min & max & median & min & max & median \cr
\noalign{\smallskip \hrule height 0.4pt \smallskip } 
$1 \times 200$ & 20 & 0.5 & 0.6 & 0.5 & 1 & 1 & 1 \cr
$2 \times 100$ & 20 & 0.2 & 0.2 & 0.2 & 1 & 1 & 1 \cr
$3 \times 66$ & 20 & 0.7 & 1.6 & 0.9 & 2 & 2 & 2 \cr
$4 \times 50$ & 20 & 3.3 & 6.3 & 4.8 & 5 & 5 & 5 \cr
$5 \times 40$ & 20 & 15.3 & 38.3 & 29.3 & 10 & 15 & 14 \cr
$6 \times 33$ & 20 & 99.6 & 317.4 & 233.0 & 26 & 45 & 42 \cr
$7 \times 28$ & 20 & 771.6 & 2702.9 & 2262.0 & 70 & 139 & 123 \cr
$8 \times 25$ & 15 & 4029.2 & 26841.6 & 14032.5 & 146 & 373 & 268 \cr
\noalign{\smallskip \hrule height 0.4pt \smallskip } 
}}\hfil
\end{center}
\caption{Running times and memory usage on random subgraphs of grids with about 200 vertices%
\label{tab:runningtimes1}}
\end{table}

\begin{table}
\begin{center}
\hfil
\vbox{%
\tabskip=0pt\def\hw{===\hidewidth}%
\halign to \hsize{\hfil#\tabskip=0pt plus 1em&\hfil#\tabskip=2em plus 2em&
	\tabskip=0pt plus 2em\hfil#&\hfil#&\hfil#\tabskip=2em plus 2em&
	\tabskip=0pt plus 2em\hfil#&\hfil#&\hfil#\tabskip=0pt\cr
\omit \hfill \textsc{Minimum Vertex Cover} \hfill \span \span \span \span \span \span \span \cr
\noalign{\medskip \hrule height 0.4pt \smallskip} 
&& \omit \hfill time in seconds \hfill \span \span & \omit \hss memory in MB\hss \span \span \cr
width&runs&min&max&median&min&max&median\cr
\noalign{\smallskip \hrule height 0.4pt \smallskip} 
1 & 387 & 0.5 & 0.8 & 0.6 & 2 & 3 & 2 \cr
2 & 179 & 0.1 & 1.0 & 0.8 & 2 & 4 & 3 \cr
3 & 68 & 0.1 & 0.3 & 0.2 & 3 & 4 & 3 \cr
4 & 74 & 0.2 & 0.5 & 0.3 & 3 & 4 & 4 \cr
5 & 69 & 0.4 & 1.3 & 0.7 & 3 & 4 & 4 \cr
6 & 62 & 0.9 & 2.3 & 1.4 & 4 & 6 & 5 \cr
7 & 38 & 1.5 & 5.5 & 3.1 & 5 & 11 & 9 \cr
8 & 36 & 2.5 & 14.0 & 6.3 & 8 & 20 & 15 \cr
9 & 45 & 7.4 & 34.9 & 16.6 & 18 & 40 & 27 \cr
10 & 29 & 24.8 & 121.6 & 56.8 & 30 & 78 & 56 \cr
11 & 28 & 55.8 & 382.1 & 156.8 & 56 & 138 & 103 \cr
12 & 29 & 164.6 & 1495.9 & 392.7 & 100 & 293 & 160 \cr
\noalign{\smallskip \hrule height 0.4pt \bigskip}
\omit \hfill \textsc{Minimum Dominating Set} \hfill \span \span \span \span \span \span \span \cr
\noalign{\medskip \hrule height 0.4pt \smallskip} 
&& \omit \hfill time in seconds \hfill \span \span & \omit \hss memory in MB \hss \span \span \cr
width&runs&min&max&median&min&max&median\cr
\noalign{\smallskip \hrule height 0.4pt \smallskip } 
1 & 387 & 0.6 & 0.9 & 0.7 & 2 & 3 & 2 \cr
2 & 174 & 0.1 & 1.0 & 0.2 & 2 & 4 & 3 \cr
3 & 39 & 0.2 & 0.8 & 0.5 & 3 & 4 & 3 \cr
4 & 30 & 0.7 & 4.5 & 2.6 & 4 & 8 & 6 \cr
5 & 17 & 4.3 & 29.2 & 16.5 & 8 & 21 & 16 \cr
6 & 9 & 112.3 & 318.5 & 187.8 & 38 & 63 & 49 \cr
7 & 1 & 2403.9 & 2403.9 & 2403.9 & 162 & 162 & 162 \cr
8 & 3 & 35290.3 & 64922.8 & 51801.0 & 319 & 338 & 321 \cr
9 & 1 & 43228.4 & 43228.4 & 43228.4 & 347 & 347 & 347 \cr
\noalign{\smallskip \hrule height 0.4pt \bigskip}
\omit \hfill \textsc{3-Colorability} \hfill \span \span \span \span \span \span \span \cr
\noalign{\medskip \hrule height 0.4pt \smallskip} 
&& \omit \hfill time in seconds \hfill \span \span & \omit \hss memory
in MB \hss \span \span \cr
width&runs&min&max&median&min&max&median\cr
\noalign{\smallskip \hrule height 0.4pt \smallskip } 
1 & 387 & 0.1 & 0.2 & 0.1 & 2 & 3 & 2 \cr
2 & 174 & 0.2 & 0.8 & 0.5 & 2 & 4 & 3 \cr
3 & 39 & 0.8 & 3.6 & 2.2 & 3 & 6 & 5 \cr
4 & 30 & 4.1 & 22.2 & 15.0 & 7 & 16 & 13 \cr
5 & 17 & 35.3 & 156.9 & 99.9 & 19 & 52 & 37 \cr
6 & 9 & 485.8 & 1328.8 & 1168.8 & 81 & 150 & 125 \cr
8 & 2 & 33733.8 & 75099.9 & 54416.8 & 446 & 664 & 555 \cr
\noalign{\smallskip \hrule height 0.4pt \smallskip } 
}}\hfil
\end{center}

\caption{Running times and memory usage for some random graphs on 200 vertices, grouped by
the width of the tree decomposition used.%
\label{tab:runningtimes2}}
\end{table}

We started to implement the approach presented in this article in
C++.  The current version works for graphs over the vocabulary
$\taugraph = (\adj)$.  At certain places, the implementation varies
from the algorithms presented in this paper for increased efficiency.
For instance, $\reduce()$ is usually not called explicitly but computed directly where needed.

We list some running times and memory usage of the implementation when solving the three problems
discussed in the previous section.  Input graphs are randomly generated subgraphs of $n \times m$ grids and
Erd\H{o}s--R\'enyi random graphs.  All graphs have about 200 vertices and the probability to
include an edge ranges between $0.001$ and $0.015$.  For the grid-subgraphs
we used path decompositions of width~$n$.  Tree decompositions for the
random graphs were computed by a triangulation heuristics (cf.~\cite{Bod03}).
The tests were done under Linux 2.6.32 on a Intel Core 2 Quad CPU Q6600 (2.40GHz) with 4 GB RAM.

\section{Conclusion}

Motivated by a practical application, we present an alternative proof of Courcelle's Theorem.
Our proof is based on model checking games and tries to avoid expensive constructions such as the
power set construction for tree automata, which turned out to cause some problems in practice.

Let us mention that our approach could be made simpler if we applied it
to graphs of bounded clique-width.  The union operation for \emph{join}
nodes of a tree decomposition involves a ``fusion'' of elements and of
interpretations of nullary symbols.  The clique-width parse trees do
not use nullary symbols and the union is replaced by a disjoint union,
which simplifies many of the operations.  On the other hand, the lack
of suitable algorithms to compute the mandatory clique-width parse trees
favors treewidth based techniques for practical applications.

First experiments with our approach do indeed indicate practical feasibility.
An implementation based on our proof can solve the \textsc{3-Colorability} problem
for some graphs where the automata theoretic approach based on the well-known
MONA tool failed.  The running times of our generic implementation can still not compete
with specialized, hand-written algorithms that can easily solve problems such as, say
\textsc{3-Colorability}, for graphs of treewidth~15 and beyond.  
We are confident that further optimization can improve the feasibility of our generic approach in practical
applications even more.

\section{Acknowledgments}

The authors thank an anonymous referee for valuable comments and
suggestions that significantly helped to improve the quality of the
paper.  The authors thank Somnath Sikdar for useful discussions on
earlier drafts of this paper.  The third author thanks Bruno Courcelle
for pointing out how to use incidence graphs to handle edge set
quantifications in a simple way.


 \def\Springer{Springer}
  \def\LNCS{Lecture Notes in Computer Science}
  \def\SIDMA{SIAM Journal on Discrete Mathematics}
  \def\FSTTCS#1{Proceedings of the #1~Conference on Foundations of Software Technology and Theoretical Computer Science (FSTTCS)}
  \def\LIPIcs{Leibniz International Proceedings in Informatics (LIPIcs)}
  \def\Dagstuhl{Schloss Dagstuhl--Leibniz-Zentrum fuer Informatik}
  \def\ICALP#1{Proceedings of the #1~International Colloquium on Automata, Languages, and Programming (ICALP)}
  \def\LATIN#1{Proceedings of the #1~Symposium on Latin American Theoretical Informatics (LATIN)}

\end{document}